\documentclass[12pt]{article}
\usepackage{authblk}
\usepackage{geometry}
\usepackage{amsthm}
\usepackage{mathrsfs}
\usepackage{subfig}
\usepackage{caption}
\usepackage{customcommands}
\usepackage{bm}
\usepackage{Nettastyle}
\usepackage{amsmath}
\usepackage{amssymb}
\usepackage{amsfonts}
\usepackage{amsthm}
\usepackage{braket}
\usepackage[normalem]{ulem}
\usepackage{color}
\usepackage{ulem}
\usepackage{mathtools}
\usepackage{tensor} 

\newcommand{\ol}{\overline}
\newcommand{\sH}{\mathcal{H}}
\newcommand{\wt}{\widetilde}
\DeclareMathOperator{\tr}{tr}

\theoremstyle{plain}
\newtheorem*{thm-restate}{Theorem \ref{thm:qms_exact}}

\newcommand{\bfig} {\begin{figure}\begin{center}}
\newcommand{\efig}{\end{center}\end{figure}}
\newcommand{\bi}{\begin{itemize}}
\newcommand{\ei}{\end{itemize}}

\newcommand{\Hl}{\mathcal{H}_\ell}
\newcommand{\Hr}{\mathcal{H}_r}
\newcommand{\HL}{\mathcal{H}_L}
\newcommand{\HR}{\mathcal{H}_R}
\newcommand{\HB}{\mathcal{H}_B}
\newcommand{\Hf}{\mathcal{H}_f}
\newcommand{\HP}{\mathcal{H}_P}
\newcommand{\Hall}{\mathcal{H}_L\otimes\mathcal{H}_\ell\otimes\mathcal{H}_r\otimes\mathcal{H}_R}
\newcommand{\ran}{\rangle}
\newcommand{\lan}{\langle}

\newcommand{\PR}{\mathrm{Pr}}

\newtheorem{mydef}{Definition}[section]
\newtheorem{lemma}{Lemma}[section]

\newcommand{\LP}{\mathcal{L}_\Phi}
\newcommand{\ph}{\psi_{Hawk}}
\newcommand{\chin}{\chi^{\{in\}}}
\newcommand{\chot}{\chi^{\{out\}}}

\title{The black hole interior from non-isometric codes and complexity}

\author[1]{Chris Akers,}
\author[1]{Netta Engelhardt,}
\author[1]{Daniel Harlow,}
\author[2,3]{Geoff Penington,}
\author[1]{Shreya Vardhan}

\affiliation[1]{Center for Theoretical Physics,\\
Massachusetts Institute of Technology, Cambridge, MA 02139, USA}
\affiliation[2]{Center for Theoretical Physics,\\ University of California, Berkeley, CA 94720 USA}
\affiliation[3]{Institute for Advanced Study, 1 Einstein Dr, Princeton, NJ 08540 USA}

\emailAdd{cakers@mit.edu}
\emailAdd{engeln@mit.edu}
\emailAdd{harlow@mit.edu}
\emailAdd{geoffp@berkeley.edu}
\emailAdd{vardhan@mit.edu}

\abstract{Quantum error correction has given us a natural language for the emergence of spacetime, but the black hole interior poses a challenge for this framework: at late times the apparent number of interior degrees of freedom in effective field theory can vastly exceed the true number of fundamental degrees of freedom, so there can be no isometric (i.e. inner-product preserving) encoding of the former into the latter.  In this paper we explain how quantum error correction nonetheless can be used to explain the emergence of the black hole interior, via the idea of ``non-isometric codes protected by computational complexity''.  We show that many previous ideas, such as the existence of a large number of ``null states'', a breakdown of effective field theory for operations of exponential complexity, the quantum extremal surface calculation of the Page curve, post-selection, ``state-dependent/state-specific'' operator reconstruction, and the ``simple entropy'' approach to complexity coarse-graining, all fit naturally into this framework, and we illustrate all of these phenomena simultaneously in a soluble model. 
}

\begin{document}
\maketitle

\section{Introduction}\label{sec:intro}
Understanding the quantum behavior of black holes is a longstanding problem in theoretical physics.  The key insight from including perturbative quantum corrections to gravity is that black holes behave as quantum systems with a finite number of degrees of freedom given in Planck units by \cite{Bekenstein:1973ur,Hawking:1975vcx}
\be\label{SBH}
S=\frac{\mathrm{Area}}{4}.
\ee
This leads to the idea that spacetime itself is only an approximate or ``emergent'' notion, valid in some situations but not in others.  This is especially clear in the context of the black hole information problem, where some kind of non-locality which is not present in the classical theory is required if unitarity is to be preserved \cite{Hawking:1976ra}. 

\begin{figure}
   \begin{center}
    \includegraphics[height=3.5cm]{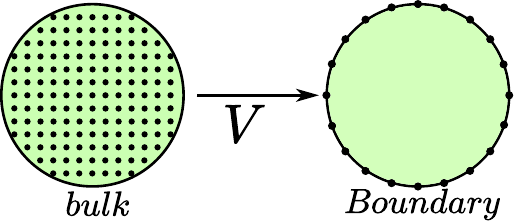}
    \caption{The holographic encoding map in AdS/CFT: an approximate isometry $V:\mathcal{H}_{bulk}\to\mathcal{H}_{Boundary}$ maps the effective bulk degrees of freedom to the fundamental boundary degrees of freedom.}\label{Vadsfig}
    \end{center}
\end{figure}

It is one thing to say that spacetime should be emergent, but quite another to know what this actually means.  So far our best-understood example of a theory with emergent spacetime is the AdS/CFT correspondence \cite{Maldacena:1997re}.   A major breakthrough towards understanding the mechanism responsible for spacetime emergence in AdS/CFT was the discovery of the Ryu-Takayanagi formula for boundary entropy in terms of bulk geometry \cite{Ryu:2006bv}, which was eventually refined into the quantum extremal surface (QES) formula of Engelhardt and Wall \cite{Engelhardt:2014gca}.  Together with progress on the topic of ``bulk reconstruction'', which is the problem of representing bulk operators on the dual CFT Hilbert space \cite{Banks:1998dd,Hamilton:2006az,Heemskerk:2012mn,Bousso:2012mh,Czech:2012bh,Wall:2012uf,Headrick:2014cta}, this led to the realization that the emergence of the gravitational spacetime can and should be formulated as a problem in quantum error correction \cite{Almheiri:2014lwa,Pastawski:2015qua,Dong:2016eik,Harlow:2016vwg}.  The basic idea is shown in figure \ref{Vadsfig}: the low-energy Hilbert space $\mathcal{H}_{bulk}$ of bulk effective field theory is mapped into the boundary Hilbert space $\mathcal{H}_{Boundary}$ by an approximately isometric map $V:\mathcal{H}_{bulk}\to\mathcal{H}_{Boundary}$ (an isometry is a linear map $V$ from one Hilbert space to another that preserves the inner product, or equivalently a linear map which obeys $V^\dagger V=I$).  In quantum error correction language $\mathcal{H}_{bulk}$ is the ``logical'' Hilbert space and $\mathcal{H}_{Boundary}$ is the ``physical'' Hilbert space, and standard properties of the correspondence such as bulk locality \cite{Almheiri:2014lwa}, entanglement wedge reconstruction \cite{Dong:2016eik}, and the quantum extremal surface formula \cite{Harlow:2016vwg}, follow from the error-correcting properties of the approximate isometry $V$ (see also \cite{Hayden:2016cfa,Cotler:2017erl,hayden2017alphabits,Akers:2019wxj,Akers:2020pmf,Akers:2021fut}).  In particular if we split the boundary theory into the degrees of freedom in a spatial region $B$ and the degrees of freedom in its complement $\ol{B}$, which we will treat as a inducing a tensor product decomposition $\mathcal{H}_{Boundary}=\mathcal{H}_B\otimes\mathcal{H}_{\ol{B}}$, then there is a similar decomposition $\mathcal{H}_{bulk}=\mathcal{H}_b\otimes \mathcal{H}_{\ol{b}}$ in the bulk, where $b$ is the entanglement wedge of $B$ and $\ol{b}$ is the entanglement wedge of $\ol{B}$, such that information in $b$ is accessible in $B$ and information in $\ol{b}$ is accessible in $\ol{B}$ (see figure \ref{subfig}).\footnote{Mathematically it is better to formulate these decompositions in terms of commuting algebras rather than tensor factors, since the latter don't usually exist unless the algebras act on entire connected components of the bulk/boundary space.  Having already committed this sin once, we will continue to commit it without further comment.}  Moreover in this standard picture, for any state $\rho$ on $\mathcal{H}_{bulk}$ we have the QES formula
\be
S(\tr_{\ol{B}}(V\rho V^\dagger))\approx\frac{\mathrm{Area}(X_B^{\rm min})}{4}+S(\rho_b).
\ee

\bfig
\includegraphics[height=4cm]{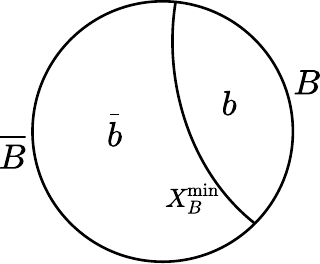}
\caption{Subregion duality in AdS/CFT: given a boundary subregion $B$, its minimal quantum extremal surface $X_B^{\rm min}$ divides the bulk space into the entanglement wedge $b$ of $B$ and the entanglement wedge $\ol{b}$ of $\ol{B}$, and all bulk observables in $b$ (or $\ol{b}$) on this low-energy subspace can be represented as operators in $B$ (or $\ol{B}$).}\label{subfig}
\efig
There is more to $\mathcal{H}_{Boundary}$ than just low-energy states of course, but we can understand the rest of the Hilbert space along similar lines as follows: including more energetic states in the domain of $V$ allows for black holes in generic microstates to be present, provided that the reconstruction excludes the black hole interiors.  In this way we can give a mathematical picture of the emergence of all parts of spacetime which are not behind black hole horizons in all states of the boundary theory \cite{Almheiri:2014lwa,Pastawski:2015qua}.  Unfortunately the more high-energy states are included, the less of the spacetime can be accounted for.  

\bfig
\includegraphics[height=4cm]{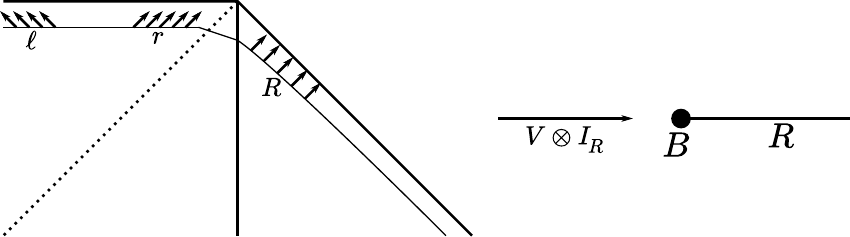}
\caption{Holography for a black hole $B$ radiating into a reservoir $R$.  In the effective picture there are left- and right-moving modes $\ell$ and $r$ on a ``nice slice'' Cauchy surface, and the holographic map $V:\Hl\otimes\Hr\to\HB$  tells us how these interior modes are encoded in the fundamental degrees of freedom.  We will always take $V$ to be a linear map, but in general it cannot be an isometry.}\label{evapbhmapfig}
\efig
What then can we say about the emergence of the spacetime behind black hole horizons?  It has been understood for some time that there are obstacles to understanding this in the same way as was done for the exterior \cite{Mathur:2009hf,Marolf:2012xe,Almheiri:2012rt,Verlinde:2012cy,Almheiri:2013hfa,Marolf:2013dba,Papadodimas:2013jku,Harlow:2014yka}. The basic problem is that the number of interior degrees of freedom which is suggested by effective field theory can vastly exceed the entropy \eqref{SBH}.  A prominent example of such a situation is a black hole in flat space which has been evaporating for a long time.  As a simple model we can treat the black hole as a quantum system $B$ with Hilbert space $\HB$, whose dimensionality $|B|$ we think of as being analogous to $e^{\frac{\mathrm{Area}}{4}}$, which is coupled to a ``reservoir'' $R$ with Hilbert space $\mathcal{H}_R$ in such a way that Hawking radiation from the black hole can propagate out into $R$.  We will refer to this description of the system as the \textit{fundamental description}. For example $B$ could describe the states in some energy band of a holographic CFT, with energy high enough that the bulk interpretation consists mostly of  ``big'' black hole states, and $R$ could be a free field theory on a half space in one dimension higher than the CFT \cite{Rocha:2008fe,Penington:2019npb,Almheiri:2019psf}.  For a real evaporating black hole $R$ includes weakly-coupled gravitons, but the radiation cloud is sufficiently diffuse that no holographic bounds are close to being saturated so a weakly-interacting Fock-space description of its Hilbert space should be accurate.\footnote{In the language of this paper, the holographic map from the reservoir into whatever is the fundamental description of gravity in flat space should be an approximate isometry which does not have any null states, with errors that are smaller than any inverse power of $|B|$, so we can just think of the radiation as fundamental.  We can illustrate this separation in AdS/CFT as follows: if we create a small black hole in $AdS_5$ whose entropy is of order $S\sim N^{1/4}$ and then wait for it to evaporate, the non-perturbative errors involved in encoding the resulting radiation state into the dual CFT are $O(e^{-N^2})$, which is much smaller than the small errors of order $e^{-S}$ which are relevant for the information problem for this black hole.}  We have used the term ``reservoir'' to denote the system $R$ into which the black hole radiates, but going forward we will mostly refer to $R$ as the ``radiation'' to match the usual parlance.

In addition to the fundamental description we can also introduce an \textit{effective description}, which is given in terms of the same radiation system $R$ but now coupled to the black hole interior degrees of freedom on a ``nice slice'' in effective field theory: there are left-movers $\ell$ which keep track of how the black hole was made, and right-movers $r$ which in Hawking's picture of the dynamics are entangled with outgoing modes in $R$.\footnote{In this paper we will ignore interactions between $\ell$ and $r$: these can be handled perturbatively provided we do not try to evolve a post-measurement state on a nice slice backwards in time (we won't), but at a substantial cost in clarity.}  We can then try to formulate a linear holographic map
\be
V:\Hl\otimes\Hr\to\HB
\ee
whose tensor product with the identity on $R$ tells us how states in the effective description are mapped into the fundamental description (see figure \ref{evapbhmapfig} for an illustration).
In this language the core problem is that in the effective description we eventually have
\be\label{oldbh}
S(\rho_r)=S(\rho_R)\gg \log|B|,
\ee
but this is not compatible with $V$ being even approximately an isometry. Indeed \eqref{oldbh} implies that $|r|\gg |B|$, so for any linear map $V:\Hl\otimes\Hr\to\HB$ there must necessarily be a sizeable subspace of ``null states'' which are annihilated by $V$.  

One may wonder why we have taken the full map to have the form $V\otimes I_R$, rather than some more complicated map which is not a product.  The reason is that we are thinking of the reservoir system as both part of the effective description and part of the fundamental description, so operators on $R$ should have the same interpretation in both descriptions and thus should commute with the holographic map.  In particular when the effective description state is a product between $\ell r$ and $R$ then the holographic map should preserve the state of $R$.  After all the reservoir is just one of potentially many quantum systems that could be in the same universe of the black hole, and if it is uncorrelated with all of them then how would we know which one the map should mix with the black hole?  We emphasize however that a holographic map of the form $V\otimes I_R$ \textit{can} act nontrivially on the state of the reservoir \textit{provided} it is entangled with  $\ell r$: since $V$ is not an isometry, it can teleport information from the interior out into $R$.  We will see below that this is the basic mechanism by which information is conserved in black hole evaporation.

How are we to think about the emergence of spacetime in a situation where the encoding of the effective description into the fundamental description is not isometric?  One possibility is to say that spacetime does not actually emerge in this situation, as was indeed argued in \cite{Almheiri:2012rt}.  On the other hand, in \cite{Penington:2019npb,Almheiri:2019psf} it was shown that if we do take the effective description seriously at late times, as in figure \ref{evapbhmapfig}, then by applying the quantum extremal surface formula to $R$ we can (1) obtain a Page curve \cite{Page:1993wv} which is consistent with unitarity and (2) also reproduce the Hayden-Preskill scrambling argument and ``black holes as mirrors'' effect \cite{Hayden:2007cs}. Inspired by these successes, in this paper we take seriously the idea that the emergence of the black hole interior is described by a linear but non-isometric holographic map $V$ from the effective description on a nice slice to the fundamental black hole Hilbert space as in figure \ref{evapbhmapfig}.
The essence of our proposal is the following:\footnote{The observation that \eqref{oldbh} poses an obstruction to an encoding of the interior degrees of freedom on a nice slice into the microstates of a black hole, and in particular the idea that there might be some ``invalid" states in such a map, has a long history.  See \cite{Susskind:1993if,Kiem:1995iy,Jafferis:2017tiu,Hayden:2018khn,Almheiri:2018xdw,Penington:2019kki,Marolf:2020xie},\cite{Akers:2021fut} for a sampling of places where this has been discussed.  Our main contributions here are 1) to emphasize the crucial role of complexity theory in making a non-isometric encoding work and 2) to give models which illustrate it in detail and show how it relates to other ideas about quantum black holes.}
\begin{center}
\textbf{There is a large set of ``null states'' in the Hilbert space of effective field theory inside a black hole, each of which is annihilated by the holographic map to the fundamental degrees of freedom.  This however cannot be detected by any observer who does not perform an operation of exponential complexity.}
\end{center}
Here ``exponential'' means ``exponential in the black hole entropy'', or equivalently ``exponential in $\log|B|$''.
The idea that computational complexity may give limits on the validity of gravitational effective field theory was introduced in \cite{Harlow:2013tf}, and explored further in \cite{Susskind:2014rva,Susskind:2015toa,Engelhardt:2017aux,Engelhardt:2018kcs,Brown:2019rox,Bouland:2019pvu,Kim:2020cds,Engelhardt:2021mue,Engelhardt:2021qjs}.  Our goal in this paper is to turn it into precise mathematics, giving concrete examples of non-isometric codes which realize it in a setting where many explicit computations are possible.  Indeed we will see that our models achieve the following: 
\bi
\item Our non-isometric map $V\otimes I_{R}$ preserves to exponential precision the inner product between all states of sub-exponential complexity in the effective description.  In other words no failure of isometry can be detected without doing something exponentially complex.
\item The entropy of the radiation can be computed either directly in the fundamental description or using the QES formula in the effective description, with ``islands'' appearing where needed to ensure consistency with the purity of the total state in the fundamental description.
\item There are formulas in the fundamental description for the probabilities of measurement outcomes and the encoded post-measurement states for any sub- exponential observable in the effective description, including those in the black hole interior, and these agree (starting in sub-exponential states) with the usual rules of quantum mechanics in the effective description up to exponentially small errors.
\item Observables of sub-exponential complexity in the effective description can be reconstructed in the fundamental description in accordance with the rules of entanglement wedge reconstruction, so in particular at late times sub-exponential observables in the black hole interior can be reconstructed on the radiation system alone. 
\item Coarse-grained observables in the fundamental description can be given a geometric intrepretation in terms of the ``outermost wedge'' of Engelhardt and Wall \cite{Engelhardt:2017aux,Engelhardt:2018kcs}, and in particular the calculation of the radiation entropy which follows from Hawking's formalism can be given a fundamental interpretation as the ``simple entropy'', which is coarse-grained over observables of exponential complexity.
\ei

\paragraph{A Simple Model:} To give an illustration of some of these results, we here present a simplified version of one of our models, which suppresses the left moving modes $\ell$.  Motivated by the chaotic nature of black hole dynamics, we define the holographic map $V:\Hr\to\HB$ by
\be
V|n\ran_r\equiv \frac{1}{\sqrt{|B|}}\sum_b e^{i\theta(n,b)}|b\ran_B,
\ee
where $|n\ran_r$ and $|b\ran_B$ are bases for $r$ and $B$ and $e^{i\theta(n,b)}$ are a bunch of randomly chosen phases.  When $|B|\gg |r|$ this map is an approximate isometry, while it clearly can't be isometric when $|B|\ll |r|$.  Nonetheless we can easily see that $V$ approximately preserves the inner products of all basis states:
\begin{align}\nonumber
\lan n'|V^\dagger V|n\ran&=\frac{1}{|B|}\sum_b e^{i\theta(n,b)-i\theta(n',b)}\\
&=\begin{cases} 1 & n=n'\\ O(1/\sqrt{|B|}) & n\neq n' \end{cases}.
\end{align}
The second line follows because when $n\neq n'$ we can think of the sum over phases as a random walk, which will typically generate something of order $\sqrt{|B|}$.  In particular this argument works in the regime where $|B|\ll |r|$, at least provided that we do not take $|r|$ to be exponentially large in $|B|$ (in that case the random walk can occasionally give a bigger inner product).  Perhaps surprisingly, we can have far more approximately orthogonal states in a Hilbert space than its dimensionality would naively suggest.  

We can also illustrate the QES formula in this model.  We can model Hawking's picture of the system in the effective description as an entangled state 
\be
|\psi_{\rm Hawk}\ran=\sum_n \sqrt{D_n}|n\ran_r|n\ran_R, 
\ee
where the $D_n$ are non-negative and sum to one.  The encoded state is
\be
(V\otimes I_R)|\psi_{\rm Hawk}\ran=\frac{1}{\sqrt{|B|}}\sum_{n,b}\sqrt{D_n}e^{i\theta(n,b)}|b\ran_B|n\ran_R,
\ee
so the reduced state on the radiation is
\be
\rho_R=\frac{1}{\sqrt{|B|}}\sum_{n,n',b}\sqrt{D_n D_{n'}}e^{i\theta(n,b)-i\theta(n',b)}|n'\ran\lan n|_R.
\ee
The second Renyi entropy is of the radiation is given by
\begin{align}\nonumber
e^{-S_2(\rho_R)}&=\frac{1}{|B|^2}\sum_{n,n',b,b'}D_n D_{n'} e^{i\Big(\theta(n,b)-\theta(n',b)+\theta(n',b')-\theta(n,b')\Big)}\\
&\approx \sum_n D_n^2+\frac{1}{|B|},\label{phaseS2}
\end{align}
where in the second line we have used that the sum is dominated either by the terms where $n=n'$ or the terms where $b=b'$, since these are the only terms where the phases cancel.\footnote{In this paper we intentionally avoid any use or discussion of Euclidean quantum gravity until the final section, but we note in passing that the second term in the second line of \eqref{phaseS2} is precisely the contribution to the second Renyi which arises from a ``replica wormhole'' in the Euclidean approach \cite{Penington:2019kki,Almheiri:2019qdq}. We here have obtained it directly in the Hilbert space formalism.}  We thus have
\be\label{S2intro}
S_2(\rho_R)\approx \min\Big[S_2((\psi_{\rm Hawk})_R),\log |B|\Big],
\ee
which is just what is expected from the QES formula (up to this being a Renyi entropy instead of a von Neumann entropy, which we will fix below).  Indeed the entropy is the minimum of the entropy in Hawking's description and the black hole entropy, just as is needed to respect the purity of the state.   The non-isometric nature of $V$ is crucial for this result: if $V$ were an isometry then we would have
\be
\tr_B\Big((V\otimes I_R)|\psi_{\rm Hawk}\ran\lan \psi_{\rm Hawk}|(V^\dagger\otimes I_R)\Big)=\left(\psi_{\rm Hawk}\right)_R,
\ee
so the radiation entropy would have to agree in the effective and fundamental descriptions.  The non-isometry of $V$ thus lies at the heart of the QES calculations of the Page curve in \cite{Penington:2019npb,Almheiri:2019psf}.

The idea of realizing the black hole interior using a non-isometric code protected by computational complexity has important implications for the black hole information problem.  The key ingredients of the problem as formulated by Hawking are the following:
\bi
\item[(1)] A finite black hole entropy with a state-counting interpretation.
\item[(2)] A unitary black hole S-matrix
\item[(3)] A black hole interior which is described to a good approximation by gravitational effective field theory, including the entanglement between the outgoing interior and exterior modes $r$ and $R$.
\ei
Hawking argued that one cannot have all three of these things in the same theory.  The picture we have been discussing so far explains how (1) and (3) can be reconciled, provided that we only demand that effective field theory is valid for sub-exponential observables.  What about (2)?  We gave some indirect evidence for this via the match \eqref{S2intro} to the QES formula, but can we realize it directly?  Indeed we can, but to do so we need to add dynamics to the story.  

To formulate a dynamical holographic map, we first need to acknowledge that at different times the length of the interior, and thus the number of degrees of freedom in $\ell$ and $r$, is different, as is the size of the black hole, and thus the number of degrees of freedom in $|B|$.  We therefore have a sequence of holographic maps
\be
V_t:\sH_{\ell_t}\otimes \sH_{r_t}\to \sH_{B_t},
\ee
which can be combined into one big map $V$ on the direct sum of interiors of all sizes mapping into the set of black holes of all sizes.  At early times $V_t$ should be close to an isometry, but at late times (i.e. after the Page time) it necessarily becomes highly non-isometric.  We then need a dynamical rule that tells us how to unitarily evolve in the fundamental description from one time slice to the next, and this rule needs to be compatible with the holographic map in the sense the we can either evolve then encode or encode then evolve (this is called equivariance).  We are indeed able to construct a family of models obeying these expectations (the basic idea is shown in figure \ref{dynevapfig} for readers who want to see it now):
\bi
\item Time evolution is exactly unitary in both the fundamental and effective descriptions, with no nonlocal interactions. In the fundamental description, the black hole degrees of freedom have $k$-local, all-to-all couplings. In the effective description, time evolution consists of Hawking pair production for the outgoing modes, while ingoing modes simply fall into the black hole.
\item There is a nonlocal holographic map $V_t$ which maps the effective description into the fundamental description at each time $t$, preserving all sub-exponential observables, and this map is indeed equivariant with respect to time evolution.
\item The entropy of the Hawking radiation follows the Page curve as a function of time, and information escapes the interior according to the Hayden-Preskill decoding criterion.
\ei

Our plan for the rest of the paper is the following: in section \ref{modelsec} we introduce our main model, and also a simplified model slightly generalizing the one described above that can illustrate many of the same features with simpler calculations.  In section \ref{complexitysec} we introduce ideas from the theory of measure concentration and use them to establish the ability of our model to hide its non-isometric nature from any sub-exponential observer.  In section \ref{pagesec} we compute entropies in the fundamental description, finding in all cases that they are compatible with the QES formula in the effective description.  In section \ref{reconstructionsec} we demonstrate how entanglement wedge reconstruction works in our main model, making use of recent ideas on ``state-specific reconstruction'' from \cite{Akers:2021fut} (see also~\cite{Hayden:2018khn}).  In section \ref{measurementsec} we combine the results of the previous sections to develop a measurement theory for observers in the black hole interior.  In section \ref{dynamicsec} we present the dynamical version of our model, showing how the holographic maps at different times are related by unitary time evolution in the fundamental description.  In section \ref{coarsesec} we explore the relation between our model and complexity coarse-graining of \cite{Engelhardt:2017aux,Engelhardt:2018kcs} and the ``Python's lunch'' conjecture of \cite{Brown:2019rox}. In section \ref{adsonesec} we explain how our ideas can be applied to non-evaporating one-sided AdS black holes.  In section \ref{multisec} we explain how they can be applied to situations with more than one black hole.  Finally in section \ref{discussion} we explain in more detail how previous ideas such as Euclidean quantum gravity, the Horowitz/Maldacena black hole final state proposal \cite{Horowitz:2003he}, the Papadodimas-Raju construction \cite{Papadodimas:2013jku}, and the ``ghost operators'' of \cite{Kim:2020cds} relate to our models.  We also give an ``FAQ'' for black hole information experts who want to quickly see our answers to a list of standard questions about the quantum mechanics of black holes.  In the appendices we develop a number of technical tools and results for use in the main text. In particular we give an overview of unitary integration using Weingarten functions and a detailed introduction to the remarkable phenomenon of measure concentration.

\subsection{Notation}
We here establish some conventions and notation.  This can either be read now or referred back to as needed.  Quantum systems are labelled by both upper case and lower case letters, i.e. $A,B,a,b,\ldots$, with their associated Hilbert spaces indicated by $\sH_A,\sH_B,\ldots$.  The dimensionalities of $\sH_A,\sH_B,\ldots$ are written as $|A|, |B|,\ldots$. We will for the most part only consider finite-dimensional Hilbert spaces, and when in doubt this can be assumed.  In any Hilbert space $\sH$ we denote the vector norm of $|\psi\ran\in \sH$ as
\be
|||\psi\ran||=\sqrt{\lan \psi|\psi\ran}.
\ee
We use the term ``state'' to refer both to normalized vectors $|\psi\ran$ obeying $|||\psi\ran||=1$ and to operators $\rho$ on $\sH$ which are positive semi-definite and obey $\tr \rho=1$, with the term ``pure state'' referring to either the former or the rank-one case of the latter.  We will \textit{never} refer to vectors or positive semi-definite operators which are not normalized as states, despite the fact that we often act on states with maps that are not norm/trace preserving.

For any linear operator $X:\sH_A\to\sH_B$, the Schatten $p$-norm is defined for $p\in [1,\infty]$ by
\be
||X||_p\equiv \left(\tr\left((X^\dagger X)^{p/2}\right)\right)^{1/p}.
\ee
Particularly interesting are the trace norm $||X||_1$, which measures the distinguishability of quantum states, the Frobenius norm $||X||_2$, which is the natural norm associated to the inner product
\be
\lan X,Y\ran=\tr (X^\dagger Y),
\ee
and the operator norm 
\be
||X||_\infty=\max_{|||\psi\ran||=1}||X|\psi\ran||=\max_{|||\psi\ran,|||\phi\ran||=1}  |\lan \phi|X|\psi\ran|.
\ee
For any $p\in [1,\infty]$ the Schatten norm is indeed a norm in the sense of being linear (meaning that  $||\lambda X||_p=|\lambda|\,||X||_p$ for all $\lambda\in\mathbb{C}$), positive semi-definite,  vanishing if and only if $X=0$, and obeying the triangle inequality 
\be
||X+Y||_p\leq ||X||_p+||Y||_p.
\ee
The last is nontrivial to prove: it is a consequence of the operator version of H\"older's inequality, which says that for any $p,q\in [1,\infty]$ obeying $\frac{1}{p}+\frac{1}{q}=1$ and any operator $X$ we have
\be\label{Holder}
||X||_p=\max_{||Y||_q=1}|\tr(Y^\dagger X)|.
\ee
The proof of this follows from combining the classical H\"older inequality with von Neumann's trace inequality $\tr(AB)\leq \sum_i \sigma_i(A)\sigma_i(B)$, where $\sigma_i(X)$ are the singular values of $X$ listed in non-increasing order.  The Schatten norms of different $p$ are related by the fact that for all $1\leq p \leq q \leq \infty$ and all $X:\sH_A\to\sH_B$, we have
\be \label{normineq}
||X||_q\leq ||X||_p\leq \mathrm{rank}(X)^{\frac{1}{p}-\frac{1}{q}}||X||_q.
\ee
There is also a useful triple inequality
\be\label{triplebound}
||XYZ||_p\leq ||X||_\infty||Y||_p||Z||_\infty, 
\ee
from which (together with \eqref{normineq} we immediately see that the Schatten norm is submultiplicative:
\be
||XY||_p\leq ||X||_p||Y||_p.
\ee

We often integrate over the unitary group using the invariant Haar measure, which we denote $dU$ and normalize so that $\int dU=1$.  When discussing probability measures (such as the Haar measure) we often use the notation $\mathrm{Pr}[A]$, which means the probability that the statement $A$ is true.  For example we make frequent use of the union bound $\mathrm{Pr}[A\cup B]\leq \mathrm{Pr}[A]+\mathrm{Pr}[B]$ and also the fact that if $A\implies B$, then $\mathrm{Pr}[B]\leq \mathrm{Pr}[A]$.

We sometimes compare our results to those which can be obtained from quantum extremal surface methods, so we here recall some basic definitions related to quantum extremal surfaces.  We'll begin in AdS/CFT.  Given any boundary subregion $A$, a codimension-two surface $X_A$ in the bulk is \textit{homologous} to $A$ if there exists a \textit{homology hypersurface} $H_A$, meaning a codimension one surface such that $\partial H_A=X_A\cup A$.  The bulk domain of dependence of any such $H_A$ is called the \textit{outer wedge} of $X_A$.  $X_A$ is called a \textit{quantum extremal surface} (QES) for $A$ if it is homologous to $A$ and also extremizes the generalized entropy
\be\label{Sgen}
S_{\rm gen}[X_A]\equiv \frac{\mathrm{Area}(X_A)}{4}+S(\rho_{H_A}),
\ee
where $\rho_{H_A}$ is the state of the bulk fields on $H_A$ and the variations are restricted to preserve the homology constraint.  The \textit{quantum extremal surface formula} says that the von Neumann entropy of the CFT state on $A$ is equal to the generalized entropy of whichever QES homologous to $A$ has smallest generalized entropy:
\be
S(\rho_A)=S_{\rm gen}[X_A^{\rm min}].
\ee
The outer wedge of $X_A^{\rm min}$ is called the \textit{entanglement wedge}.  

In recent times it has been understood that the QES machinery can be used beyond AdS/CFT \cite{Penington:2019npb,Almheiri:2019psf,Almheiri:2019hni}.  The full range of validity is not yet known, but one situation which is now well-understood is when we have a holographic CFT $B$ coupled to a non-gravitational reference system $R$.  We can then ask for a way of computing the von Neumann entropy in the fundamental description of a region $A=A_B\cup A_R$, where $A_B\subset B$ and $A_R\subset R$.  The key is to find a generalization of the homology constraint to this situation.  The rule which seems to work is the following: we look for a surface $X_A$ in the gravitational part of the effective description such that we can find a hypersurface $\hat{H}_{A}$, also in the gravitational region, such that $\partial \hat{H}_A=A_B\cup X_A$. We then define the ``true'' homology hypersurface to be $H_A=\hat{H}_A\cup A_R$, and the outer wedge is now the domain of dependence of $H_A$.  The generalized entropy is given by \eqref{Sgen} as before, and $X_A$ is a QES if it is homologous to $A$ in this more general sense and also extremizes $S_{\rm gen}$.  The von Neumann entropy of $\rho_A$ is again computed by the generalized entropy of $X_A^{\rm min}$.  These definitions will be sufficient for our purposes, but it would be interesting to understand how they generalize to the situation where $R$ is weakly gravitating.  

\section{A model of the interior}\label{modelsec}
\bfig
\includegraphics[height=5cm]{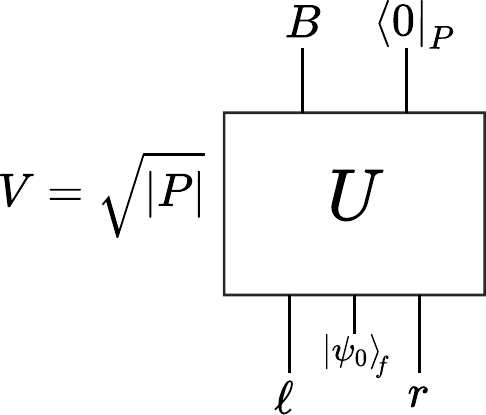}
\caption{Our model holographic map $V$: a Haar-typical unitary $U$ acts on the effective field theory degrees of freedom $\ell$ and $r$ in the black hole interior, together with some fixed additional state $|\psi_0\ran_f$, and then post-selects on some auxilliary degrees of freedom $P$, resulting in a state of the fundamental degrees of freedom $B$.}\label{Vdeffig}
\efig
We would like to define a map
\be
V:\Hl\otimes \Hr\to \mathcal{H}_B,
\ee
where $\Hl$ and $\Hr$ are the Hilbert spaces of left/right-moving modes in the black hole interior and $\HB$ is the fundamental Hilbert space of black hole microstates as in figure \ref{evapbhmapfig}. Our proposal first considers a larger Hilbert space $\Hl\otimes \Hf\otimes \Hr$, where $f$ are some extra degrees of freedom such that
\be
|f||\ell||r|=|P||B|
\ee 
for some positive integer $|P|$, then introduces an alternative tensor decomposition
\be
\Hl\otimes\Hf\otimes \Hr=\HB\otimes\HP,
\ee
and then defines 
\be\label{Vdef}
V\equiv \sqrt{|P|}\, \lan 0|_P U |\psi_0\ran_f.
\ee
Here $|\psi_0\ran_f$ is some fixed state on $\Hf$, $|0\ran_P$ is some fixed state on $\HP$, and $U$ is a typical sample from the Haar measure on unitary transformations of $\Hl\otimes\Hf\otimes\Hr$.  The role of the prefactor $\sqrt{|P|}$ will be clarified shortly. We can think of $|\psi_0\ran_f$ as keeping track of any extra effective field theory degrees of freedom that we are not interested in varying, e.g. short-distance modes that we integrated out or modes just outside of the black hole, although this interpretation will not be crucial in what follows.  We illustrate the map $V$ in figure \ref{Vdeffig}.  The main novelty is the postselection on $|0\ran_P$, which prevents $V$ from being an isometry from $\ell r$ to $B$.

Why should $U$ be Haar random?  One reason is that this facilitates computation, but a more principled reason is that on general grounds we expect the holographic map from the black hole interior to the fundamental description to be rather complicated and this is most easily achieved by using randomness.  That said, we should acknowledge that unitaries drawn from the Haar measure are surely ``too random'': they do not capture many important features of gravity such as Lorentz invariance, the dimensionality of spacetime, etc. They have nonetheless been able to qualitatively reproduce many of the key aspects of quantum black hole physics \cite{Page:1993df, Hayden:2007cs}, and  we will see that this continues to be the case here. 

We emphasize that we are interested in a single sample from the Haar measure; we are \textit{not} viewing $U$ as a random variable in the fundamental description. In what follows we will use averages over $U$ as a way to diagnose what happens in a typical sample, but these averaged results only give a good picture of the typical case for quantities whose fluctuations are small.  Quantities of this type are sometimes called self-averaging, see section 7 of \cite{Penington:2019kki} for more discussion of this idea.

\bfig
\includegraphics[height=7cm]{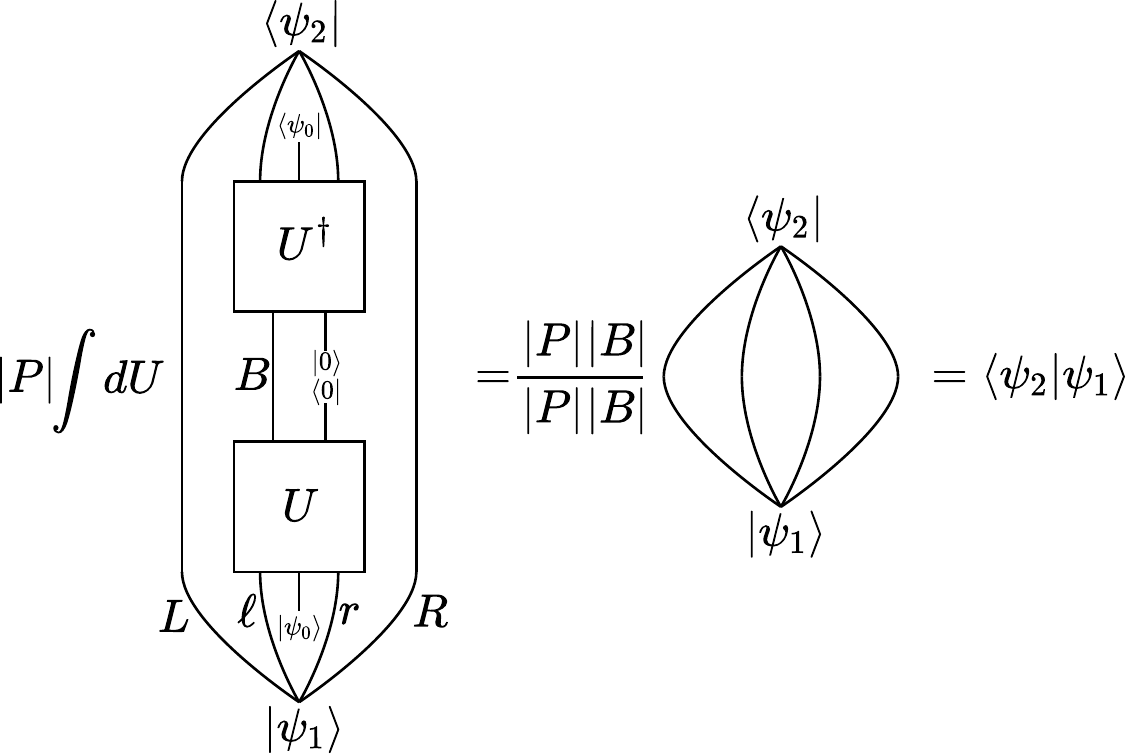}
\caption{Computing the typical overlap between $V|\psi_1\ran$ and $V|\psi_2\ran$.  By equation \eqref{Uresults} the unitary integral connects the indices going into the top/bottom of $U$ to those going into the bottom/top of $U^\dagger$, and then divides by the total dimensionality $|P||B|$.  This factor of the dimensionality is cancelled by the factor of $|P|$ in $V^\dagger V$ and a factor of $|B|$ coming from a $B$ index loop.}\label{overlapfig}
\efig
A key property of the map \eqref{Vdef}  is that on average it preserves the inner product between any two states $|\psi_1\ran,|\psi_2\ran\in\Hl\otimes\Hr$:
\be
\int dU \lan \psi_2|V^\dagger V|\psi_1\ran=\lan \psi_2|\psi_1\ran.
\ee
Moreover if we introduce a reference system $L$ with $|L|=|\ell|$ and a reservoir $R$ with $|R|=|r|$,\footnote{Mathematically $L$ and $R$ are similar, but physically they are quite different: $R$ is the reservoir into which the black hole evaporates, while $L$ is a potentially unphysical auxiliary system that we introduce to purify mixed states on $\ell$.} and take $|\psi_1\ran,|\psi_2\ran\in\HL\otimes\Hl\otimes\Hr\otimes \HR$, then we have:
\be\label{overlapav}
\int dU \lan \psi_2|(V^\dagger V\otimes I_{LR})|\psi_1\ran=\lan \psi_2|\psi_1\ran.
\ee
This calculation uses standard Haar integration technology which we review in appendix \ref{Uapp}; we illustrate it in figure \ref{overlapfig}.  

\bfig
\includegraphics[height=6.5cm]{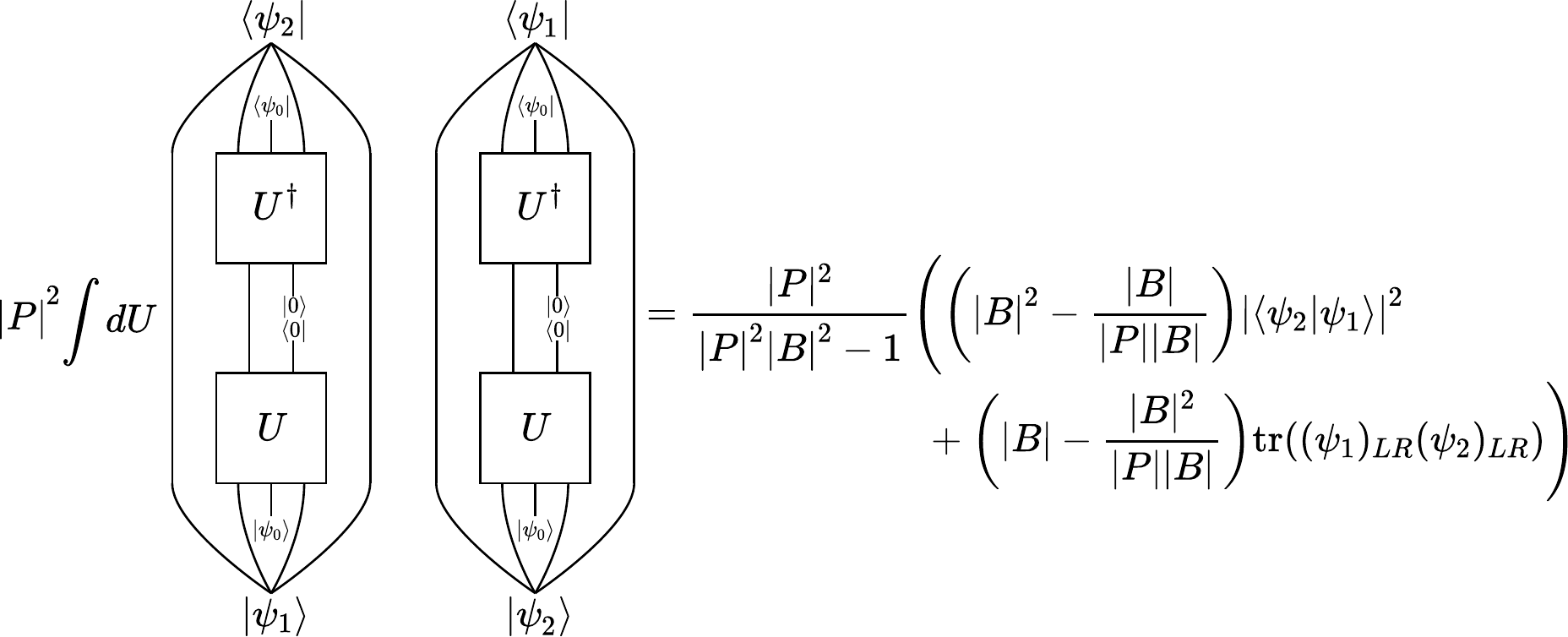}
\caption{Computing the nontrivial term in the fluctuation of the overlap.  The integral connects the indices going into the top/bottom of each $U$ to the indices going into the bottom/top of a $U^\dagger$, but now there are four ways to do these connections: two for the $\ell fr$ indices and two for the $BP$ indices, with the different choices are weighted as in \eqref{Uresults}.}\label{overlapflucfig}
\efig
It may seem surprising that a linear map which is not an isometry can on average preserve the inner product.  We can clarify the situation by computing the fluctuation in the inner product:
\begin{align}\label{overlapfluc}
\int dU \Big|\lan \psi_2|(V^\dagger V\otimes I_{LR})|\psi_1\ran-\lan \psi_2|\psi_1\ran\Big|^2=\frac{|P|-1}{|P|^2|B|^2-1}\Bigg[|P||B|\tr\big((\psi_1)_{LR}(\psi_2)_{LR}\big)-|\lan \psi_2|\psi_1\ran|^2\Bigg].
\end{align}
Here $(\psi_1)_{LR}$, $(\psi_2)_{LR}$ are the reduced states of $|\psi_1\ran$, $|\psi_2\ran$ on $LR$. The cross terms on the left hand side can be computed using \eqref{overlapav}, and the calculation of the nontrivial term is shown in figure \ref{overlapflucfig}. The right-hand side of \eqref{overlapfluc} vanishes when $|P|=1$, as it must since in this case there is no post-selection and $V$ is an isometry.  Otherwise it is nonzero, confirming that $V$ indeed is not in general an isometry. For any states $|\psi_1\ran,|\psi_2\ran$ we have $0\leq\tr\big((\psi_1)_{LR}(\psi_2)_{LR}\big)\leq 1$, so using $\frac{1}{x-1}\leq \frac{2}{x}$ for $x\geq 2$ we have
\be
\int dU \Big|\lan \psi_2|(V^\dagger V\otimes I_{LR})|\psi_1\ran-\lan \psi_2|\psi_1\ran\Big|^2\leq \frac{2}{|B|}\label{flucbound}
\ee
for $|B|>1$.  Since $|B|$ is exponential in the black hole entropy, this means that, although $V$ is not an isometry, the fluctuations in the overlap of $(V\otimes I_{LR})|\psi_1\ran$ and $(V\otimes I_{LR})|\psi_2\ran$ away from that of $|\psi_1\ran$ and $|\psi_2\ran$ are exponentially small in the entropy: $V$ almost always preserves the inner product, even for an ``old'' black hole where  
\be
1\ll|B|\ll |\ell||r|.
\ee

How can \eqref{flucbound} be compatible with the fact that when $|B|\ll |\ell||r|$ there are necessarily a large number of states which are annihilated by $V$?  A first observation is that in the derivation of \eqref{flucbound},  we had assumed that the states $|\psi_1\ran$, $|\psi_2\ran$ were independent of $U$.  Had we allowed them to depend on $U$, as the set of null states surely does, we would have gotten a different answer than \eqref{flucbound}.  From a physics point of view this says the following: states in the effective description that are prepared by an observer who does not know the holographic map are very unlikely to be states that can detect the failure of $V$ to be an isometry.  In particular one might hope that for a typical $U$ this should be true for any states of sufficiently low circuit complexity: in the following section we will show that this is indeed the case.

We can develop some more intuition for \eqref{flucbound} by introducing a simplified model to which we will refer as the phase model.  To motivate it, note that if $|i\ran_\ell$ and $|n\ran_r$ are complete orthonormal bases for $\Hl$ and $\Hr$ and $|b\ran_B$ is a complete basis for $\HB$, then we have
\be
V|in\ran_{\ell r}=\sqrt{|P|}\sum_b |b\ran_B\lan b|_B\lan 0|_P U|\psi_0\ran_f|in\ran_{\ell r}.
\ee
We can think of each matrix element in the sum as a random overlap between two states in $\HB\otimes \HP$, which is typically of order $\frac{1}{\sqrt{|P||B|}}$ times a random phase.  We define the phase model by simply taking this to define the holographic map:\footnote{In the introduction we already discussed an even simpler version of this model which ignored the left movers $\ell$, we now put them back.} 
\be\label{phasedef}
V_{phase}|in\ran_{\ell r}\equiv \frac{1}{\sqrt{|B|}}\sum_be^{i\theta(i,n;b)}|b\ran_B,
\ee
where the phases $e^{i\theta(i,n;b)}$ are given by a single sample from $|\ell||r||B|$ independent copies of the uniform distribution on $U(1)$.  We then have
\begin{align}\nonumber
\lan in|V^\dagger_{phase}V_{phase}|jm\ran&=\frac{1}{|B|}\sum_b e^{i\theta(j,m;b)-i\theta(i,n;b)}\\
&=\begin{cases} 1 & (i,n)=(j,m)\\ O\left(\frac{1}{\sqrt{|B|}}\right) & (i,n)\neq (j,m)\end{cases},\label{phaseoverlap}
\end{align}
where the second line follows because adding up $|B|$ independent phases at large $|B|$ typically gives something of order $\sqrt{|B|}$.  \eqref{phaseoverlap} is nicely compatible with \eqref{flucbound}, and we emphasize again that they both hold even in the regime where $1\ll|B|\ll |\ell||r|$.\footnote{On the other hand if we have $|\ell||r|\gg e^{|B|}$, then purely by unlucky chance there will be occasional pairs of orthogonal basis states in $\Hl \otimes \Hr$ whose images through $V$ have an inner product which is $O(1)$.  We comment further on this rather extreme regime in the following section.}  It may not be possible to injectively map an orthonormal basis of a larger Hilbert space to an orthonormal basis of a smaller one, but we can do it up to errors that are exponentially small. 

We close this section by observing that the particular structure of our holographic map $V$ defined by \eqref{Vdef}, that it is an isometry followed by a post-selection, is in fact completely general in the sense that any linear map $V$ can be put in this form.  Indeed given any $V: \mathcal{H}_{\ell} \otimes \mathcal{H}_r \to \mathcal{H}_B$ of finite operator norm $\lVert V \rVert_\infty$ we can define $\hat{V} \equiv V / \lVert V \rVert_\infty$ and
\begin{equation}
    W \equiv \ket{0}_P \otimes \hat{V} + \ket{1}_P \otimes \sqrt{I - \hat{V}^\dagger \hat{V}}~,
\end{equation}
where $\mathcal{H}_P$ is an ancillary qubit.
$W$ satisfies $W^\dagger W = I$ and so is an isometry, so we can write it as $W = U \ket{0}_f$ for some unitary $U$ and ancillary system $f$.
We then have
\begin{equation}
    V = \lVert V \rVert_\infty \bra{0}_P U \ket{0}_f~.
\end{equation}

\section{Measure concentration and sub-exponential states}\label{complexitysec}
We have now seen that our non-isometric holographic map $V$ defined by \eqref{Vdef} has the property that $V\otimes I_{LR}$ is likely to preserve the inner product between any particular pair of states in $\HL\otimes\Hl\otimes \Hr\otimes \HR$ up to exponentially small errors.  For the phase model \eqref{phasedef} we further saw that we can achieve this for a complete basis of orthonormal states, at least as long as $|\ell||r|\ll e^{|B|}$.  But what about superpositions of $O(1)$ numbers of these basis states?  Or more ambitiously, what about states that are obtained from basis states by acting with a polynomial quantum circuit?  In this section we will see that in fact $V\otimes I_{LR}$ is likely to preserve to exponential accuracy the overlaps between all states on $\Hall$ whose circuit complexities are sub-exponential. Our discussion is closely related to the ``Johnson-Lindenstrauss lemma'' of computer science, which says that for any set of $m$ points in $\mathbb{R}^N$ there is a linear map to $\mathbb{R}^n$ which approximately preserves their distances to accuracy $\epsilon$ provided that $n>8 \frac{\log m}{\epsilon^2}$.  This is not quite what we need however, as we want a statement on $\HL\otimes \Hl\otimes \Hr\otimes \HR$ but only want to allow $V$ to act on $\Hl\otimes \Hr$, so we find it more natural to develop things from scratch.  A preliminary version of these ideas was reported as theorem 5.1 in \cite{Akers:2021fut}.

\subsection{Measure concentration}
Our argument relies on the remarkable phenomenon of measure concentration, which is essentially the statement that for certain probability distributions on certain high-dimensional manifolds the fluctuations of any reasonable observable will be exponentially small in the dimensionality of the manifold. The classic example of this phenomenon is Levy's lemma, which says that this is true for real functions on the sphere:
\begin{lemma}
(Levy) Let $F:\mathbb{S}^d\to \mathbb{R}$ be $\kappa$-Lipschitz, meaning that $|F(x)-F(y)|\leq \kappa |x-y|$ for all $x,y\in \mathbb{S}^d$, with $|x-y|$ being the Euclidean (chordal) distance on $\mathbb{S}^d$.  Then in the uniform probability measure on $\mathbb{S}^d$ we have
\begin{align}\nonumber
\PR\Big[F\geq\lan F\ran +\epsilon \Big] &\leq \exp\left(-\frac{(d-1)\epsilon^2}{2\kappa^2}\right)\\\nonumber
\PR\Big[F\leq\lan F\ran-\epsilon \Big] &\leq \exp\left(-\frac{(d-1)\epsilon^2}{2\kappa^2}\right)\\
\PR\Big[|F-\lan F\ran|\geq\epsilon \Big] &\leq 2\exp\left(-\frac{(d-1)\epsilon^2}{2\kappa^2}\right)
\end{align}
where $\lan F\ran$ is the expectation value of $F$ and $\PR[\cdot]$ indicates the probability that ``$\cdot$'' is true.\label{Levy}
\end{lemma}
Many people at first find this counter-intuitive.  One way of coming to terms with it is the following. We first observe that the median value $F_\text{med}$ of our function naturally cuts the sphere into two equal-volume halves, with $F > F_\text{med}$ and $F < F_\text{med}$ respectively. In a small neighbourhood of this cut, it follows from the Lipschitz-continuity of $F$ that $F \approx F_\text{med}$. Hence the probability of a large fluctuations in $F$ is bounded by the fractional volume of the sphere that is \emph{not} close to the cut.  Intuitively, the volume close to the cut is minimised, and hence the opportunity for fluctuations maximised, if the cut lies on an equator of the sphere, rather than oscillating in some complex way; see Figure 
\ref{fig:levysketch}.\footnote{Proving that this is indeed the case, a result known as the isoperimetric inequality, is the main technical step in one approach to proving Levy's lemma, although not the one we use in Appendix \ref{measureapp}.} But  the fractional volume of any ball of opening angle $\theta\in [0,\pi/2)$ on $\mathbb{S}^d$ is 
\be
\frac{\Omega_{d-1}}{\Omega_d}\int_0^\theta d\theta'(\sin \theta')^{d-1}\leq \frac{\Gamma((d+1)/2)}{\sqrt{\pi}\Gamma(d/2)}\theta (\sin \theta)^{d-1}\leq \sqrt{\frac{d}{2\pi}}\theta (\sin \theta)^{d-1},
\ee
which vanishes exponentially at large $d$ for any $\theta\in[0,\pi/2)$. We conclude that almost all of the volume of $\mathbb{S}^d$ is very close to the equator. It follows that the probability of a large fluctuation in $F$ is exponentially small, even in the ``worst-case scenario'' where the $F = F_\text{med}$ cut lies on an equator.
\bfig
\includegraphics[height=7cm]{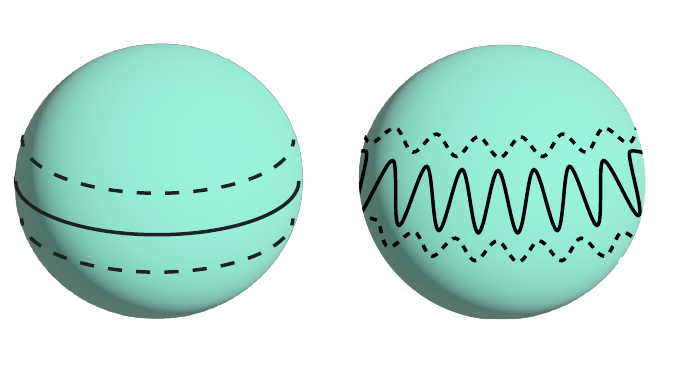} 
\caption{Given a cut $\gamma$ separating the sphere into two equal-volume halves, the volume of the small neighbourhood $\gamma_\varepsilon = \{x : \exists y \in \gamma . \,\,  ||x - y || \leq \varepsilon \}$ is minimized when $\gamma$ lies on an equator of the sphere (left) rather than oscillating in a complex way (right). Somewhat counterintuitively, in high dimensions the fractional volume of $\gamma_\varepsilon$ is exponentially close to one, even in the former case.}\label{fig:levysketch}
\efig

Theorems of this type are usually called deviation bounds.  For our work here we need an analogous deviation bound for functions on the unitary group $U(N)$ \cite{meckes2019random}:
\begin{lemma} (Meckes) Let $F:U(N)\to \mathbb{R}$ be $\kappa$-Lipschitz in the sense that $|F(U_1)-F(U_2)|\leq \kappa ||U_1-U_2||_2$, with $||X||_2\equiv\sqrt{\tr X^\dagger X}$.  Then in the Haar measure on $U(N)$ we have
\begin{align}\nonumber
\PR\Big[F\geq \lan F\ran+\epsilon\Big]&\leq \exp\left(-\frac{\epsilon^2N}{12\kappa^2}\right)\\\nonumber
\PR\Big[F\leq \lan F\ran-\epsilon\Big]&\leq \exp\left(-\frac{\epsilon^2N}{12\kappa^2}\right)\\
\PR\Big[|F-\lan F\ran |\geq\epsilon\Big]&\leq 2\exp\left(-\frac{\epsilon^2N}{12\kappa^2}\right).
\end{align}\label{Meckes}
\end{lemma}
The machinery that goes into proving deviation bounds of this type is somewhat involved; we give a self-contained introduction in appendix \ref{measureapp} which includes proofs of both of these lemmas.  The key fact is that a sufficient condition for such concentration results to hold on a $d$-dimensional Riemannian manifold with probability distribution
\be
d\mu=e^{-\Phi(x)}\sqrt{g}d^d x
\ee
is that we have
\be\label{baksuff}
R_{\mu\nu}+\nabla_{\mu}\nabla_\nu \Phi\geq \frac{1}{C} g_{\mu\nu}
\ee
for some $C>0$, so we can think of measure concentration as a consequence of positive curvature and/or convexity of $\Phi$.  The resulting concentration inequalities then have $e^{-\frac{\epsilon^2}{2C\kappa^2}}$ on the right-hand side, so e.g. for the round sphere we can take $C=\frac{1}{d-1}$, leading to lemma \ref{Levy}.

In this paper our main application of measure concentration is establishing the following theorem:
\begin{thm}\label{setconcthm}
Let $V$ be defined as in \eqref{Vdef}, and let $S$ be a collection of $N_S$ states in $\Hall$. Then for all $|B|\geq 16$ and $0<\gamma<1/2$ we have
\be
\PR\Bigg[\sup_{|\psi\ran,|\phi\ran\in S}\Big|\lan \psi|V^\dagger V\otimes I_{LR}|\phi\ran-\lan\psi|\phi\ran\Big|>\sqrt{18}|B|^{-\gamma}\Bigg]\leq 6 N_S(N_S-1)\exp\left(-\frac{|B|^{1-2\gamma}}{24}\right). \label{36}
\ee
\end{thm}
Thus as long as 
\be\label{meascond}
\frac{\log N_S}{|B|^{1-2\gamma}}\ll 1,
\ee
which can be true for rather large sets $S$ of states, then encoding by $V\otimes I_{LR}$ is very likely to preserve all pairwise overlaps of states in $S$ to exponential precision in the black hole entropy $\log |B|$.  In particular if we take $S$ to be an orthonormal basis for $\HL\otimes\Hl\otimes\Hr\otimes \HR$, then all overlaps will likely be preserved provided that $|\ell|||r|$ is not exponentially large in $|B|$, just as we found in the phase model and as suggested by the Johnson-Lindenstrauss lemma.  In the following subsection we will see that this conclusion can be extended to all states of sub-exponential complexity.  The proof of theorem \ref{setconcthm} is somewhat lengthy, so we defer it to appendix \ref{proofapp}.  The core idea is a judicious application of lemma \ref{Meckes}. 

We can also use theorem \ref{setconcthm} to study when $V$ approximately preserve the overlaps of \textit{all} states, i.e. when $V$ is approximately an isometry in the sense that
\be\label{isomeps}
||V^\dagger V-I_{\ell r}||_\infty <\delta
\ee
for some $\delta$ which is of order $|B|^{-\gamma}$ with $0<\gamma<1/2$.  From \eqref{Vdef} the only ``obvious'' upper bound is
\be\label{naivebound}
||V^\dagger V-I_{\ell r}||_\infty\leq |P|-1,
\ee
which is saturated if there is an input state $|\psi\ran_{\ell r}|\psi_0\ran_f$ which is in the image of $U^\dagger(I_B\otimes |0\ran\lan 0|_P)U$, but we can hope that there are no such input states.  In appendix \ref{appxapp} we use theorem \ref{setconcthm} to show that a sufficient condition for \eqref{isomeps} to hold with $\delta$ of order $|B|^{-\gamma}$ is that 
\be\label{isomcond}
\frac{|\ell||r|}{|B|} |B|^{2\gamma}\left(\gamma \log|B|+\log |P|\right)\ll 1.
\ee
In particular taking $\gamma$ to be small (but not parametrically small), we see that up to logarithmic factors $V$ is likely to be approximately isometric whenever it is allowed to be by counting (i.e. when $|\ell||r|\ll |B|$).

Unfortunately a theorem analogous to \ref{setconcthm} cannot be established in the same way for the phase model \eqref{phasedef}.  Measure concentration of the type characterized by theorem \ref{mcthm} does not occur on the high-dimensional torus $T^d$: the metric is flat,  so \eqref{baksuff} cannot hold for any $C>0$.  For example we can consider a function which varies smoothly over one of the circles in $T^d$ and is constant on the others.  Since the space is a product, we can just integrate out the other circles without affecting the distribution on the circle the function varies over, and thus this function will never concentrate.  This is why we have taken the model \eqref{Vdef} to be our main model, and relegated the phase model to a secondary status even though computations are easier there.\footnote{In the phase model one can still argue that if $|\ell||r|\ll |B|$ then $V$ is likely an isometry.  The idea is to construct the ``worst-case'' state for violating $V^\dagger V=1$ and then show that this state nonetheless is mapped almost to itself by $V^\dagger V$.}

\subsection{Circuit complexity and overlaps of sub-exponential states}\label{subexpsec}
In order to extend the validity of low-energy effective field theory as far as possible, we would like it to be the case that the detection of any violation thereof requires a difficult operation.  The standard way to formalize such requirements is using the theory of computational complexity, which in this case means quantum computational complexity.  Namely we would like to say that the failure of $V$ to be an isometry cannot be detected with operations of low computational complexity.  
To that end, we now show that $V\otimes I_{LR}$ is very likely to approximately preserve the overlap of any two states $|\psi_1\ran,|\psi_2\ran\in \Hall$ which are prepared by sub-exponentially complex quantum circuits.  To do this, we need to define more carefully what we mean by the latter.  

To define a notion of quantum complexity, it is necessary to introduce a model of quantum computation.  We will take our model to be built out of ``physical qudits'', each of which has a Hilbert space dimension $q>1$ which we think of as being $O(1)$ (for example it could be two).  To encode a quantum system $\mathcal{H}_A$ into qudits we need at least 
\be
m=\log_q \lceil A \rceil
\ee
qudits, where $\lceil A\rceil$ is defined to be the smallest integer power of $q$ which is greater than or equal to $|A|$.  By definition we have
\be
|A|\leq \lceil A \rceil <q|A|.
\ee
We will further assume that this encoding can be done in a ``simple'' way, so that operations which are simple in terms of the qudits are also simple to implement physically on $A$.  For example any local structure in $A$ should be respected.  We may then consider a model of computation with some $O(1)$ number $\Gamma$ of fundamental unitary gates, which can be applied in any order to any pair of the physical qudits. A quantum circuit $C$ is an ordered sequence of gates $g_1,g_2,\ldots$, with the number $\mathscr{C}(C)$ of gates being called the \textit{circuit complexity}.  Asymptotic complexity is then defined for sequences $C_1,C_2,\ldots$ of circuits acting on sequences of systems $A_1,A_2,\ldots$, and in particular a circuit family $C_n$ is called \textit{sub-exponential} if for any $\alpha>0$ we have
\be\label{subexp1}
\mathscr{C}(C_n)\leq (\lceil A_n \rceil)^\alpha
\ee
for all but finitely many $n$.  A family of states $|\psi_1\ran\in \mathcal{H}_{A_1},|\psi_2\ran\in \mathcal{H}_{A_2},\ldots$ is called sub-exponential if there is a sub-exponential family of circuits which prepares the encoded version of each starting on the all-zero state of the qudits.

Next we would like to get some sense of ``how many'' sub-exponential states there are.  We first ask how many states of complexity $\mathscr{C}_n$ there are on $\mathcal{H}_{A_n}$: this is upper bounded by the number $N_{\mathscr{C}_n}$ of circuits of complexity $\mathscr{C}_n$ on $\mathcal{H}_{A_n}$, which is given by
\be
N_{\mathscr{C}_n}=\left(\begin{pmatrix} \log_q \lceil A_n \rceil \\ 2\end{pmatrix}\Gamma\right)^{\mathscr{C}_n}\leq e^{\left(2\log \log_q \lceil A_n \rceil +\log \Gamma\right)\mathscr{C}_n} \label{311}
\ee
since for each gate we need to pick which gate we want and which two qudits to act on.  Since for any $\alpha>0$ a subexponential circuit family always eventually obeys \eqref{subexp1}, for any $\alpha>0$ we therefore eventually have
\be
N_{\mathscr{C}_n}\leq e^{\left(2\log \log_q \lceil A_n \rceil +\log \Gamma\right)\left(\lceil A_n \rceil\right)^\alpha}.
\ee

To apply these definitions to the problem at hand, we need to think a bit more about which quantity we should use to measure complexity.  There are several potentially large numbers around, namely $|\ell|$, $|r|$, and $|B|$.  At early times we have $|r|\ll |B|$, while at late times we have $|r|\gg |B|$, and we also have some flexibility in how to choose $|\ell|$.  We certainly do not want effective field theory to break down at early times or when we choose $|\ell|$ to be small, so we should at least hope that any breakdown requires an operation whose complexity is exponential in $\log |B|$.  As we discuss momentarily we will not consider situations where $|\ell|$ or $|r|$ is exponentially bigger than $|B|$, and so it is most natural for us to define a sub-exponential circuit in our problem to be one such that for any $\alpha>0$ we eventually have $\mathscr{C}(C_n)\leq |B_n|^\alpha$.  Dropping the $n$'s to conform to our previous notation, for any $\alpha>0$ the number $N_{sub-exp}$ of sub-exponential states will always eventually obey\footnote{This formula has the usual complexity-theoretic problem that for any fixed problem size we can't be sure that the particular circuit family we are interested is yet to obey the asymptotic bound.  For example for many values of $n$  an algorithm which runs in time $10^{500} n$ is slower than an algorithm which runs in time $n^{10}$, even though the latter is asymptotically slower.  In practice we just have to assume that the particular sub-exponential states we are actually interested in have already reached the asymptotic regime for the black holes we are interested in.}
\be\label{subexpcount}
N_{sub-exp}\leq e^{\left(2\log \log_q(\lceil |l|^2|r|^2\rceil)+\log \Gamma\right)(\lceil B\rceil)^\alpha}\sim e^{|B|^\alpha\log \log (|l||r|)}.
\ee
By theorem \ref{setconcthm} we then have
\be\label{subexpconc}
\PR\Bigg[\sup_{|\psi\ran,|\phi\ran \, \mathrm{sub-exp}}\Big|\lan \psi|V^\dagger V\otimes I_{LR}|\phi\ran-\lan\psi|\phi\ran\Big|>\sqrt{18}|B|^{-\gamma}\Bigg]\lesssim \exp\Big(|B|^\alpha\log \log (|l||r|)-\frac{|B|^{1-2\gamma}}{24}\Big),
\ee
so choosing $\alpha<1-2\gamma$ we see that the overlaps between all sub-exponential states will almost surely be preserved to exponentially small precision for sufficiently large black holes \textit{provided} that $|l||r|$ is not doubly exponentially large in $|B|$.  Assuming that this is not the case, we have as an immediate corollary that any null state must be exponentially complex.

How concerned should we be that the bound \eqref{subexpconc} breaks down when $|\ell||r|$ is doubly exponential in $|B|$?  The short answer is ``not very''.  In such a situation the nice slice on which we are defining the effective description is extremely long (see figure \ref{evapbhmapfig}), and almost entirely inaccessible to any infalling observer. Already when $|\ell||r|$ is merely exponentially large in $|B|$, the exponential time required to prepare the black hole from collapse makes any state exponentially complex (and hence out of scope for the effective description) if we define complexity using a reference state that does not itself contain a black hole. Even if we try to sidestep this issue by taking the reference state to be e.g. the Hawking state, we can never have an entire basis of states for $|\ell||r|$ with subexponential complexity when $|\ell||r|$ is itself exponential in $|B|$.\footnote{For more on why some kind of breakdown of effective field theory is expected at singly exponential times  see e.g. \cite{Maldacena:2001kr,Susskind:2015toa,Cotler:2016fpe,Harlow:2018tqv,Susskind:2020wwe}.} At the doubly exponentially long times necessary for \eqref{subexpconc} to break down, even more drastic breakdowns of effective field theory are expected. For example, in the (admittedly non-evaporating) thermofield double state in $AdS$, one encounters bizarre recurrences where the black hole returns to approximately its original state. 

One possible objection to the content of this section is that perhaps the reason that $V\otimes I_{LR}$ approximately preserves inner products of sub-exponential states is merely that a generic $U$ has exponential complexity, and therefore this is unlikely to be true for a more realistic map.  For example we will see below that more realistic $U$ have only polynomial complexity.  There are two reasons to be skeptical of this objection.  The first is that the calculation leading to \eqref{flucbound} does not really require Haar-random unitaries; it is enough to use $k$-designs with $k>2$, which are sets of sub-exponential unitaries which have the same low moments as Haar-random unitaries do (see appendix \ref{Uapp} for a brief introduction to $k$-designs).  The second reason is that attempting to directly create null states of sub-exponential complexity merely by assuming that $U$ is sub-exponential does not work as long as $\log |f|$ is polynomial in $\log |B|$, as we explain in section \ref{adsonesec} below.  On the other hand it is not currently possible to prove something like \eqref{subexpconc} for $k$-designs, since they have not been shown to obey any measure concentration result analogous to lemma \ref{Meckes}.  This is a pity, as such a result would immediately resolve many open problems in complexity theory.  That is also likely an indication that such a proof would be quite difficult to find.  This however should not be viewed as an indication that no such result exists; many statements in complexity theory, such as \textbf{P}$\neq$\textbf{NP}, are widely expected to be true but seem hopelessly difficult to prove using current techniques.

\section{Quantum extremal surfaces and entropy in the Hawking state}\label{pagesec}
So far we have discussed general properties of the holographic encoding map $V$, defined by \eqref{Vdef}, from states in the effective description to states in the fundamental description.  We learned that this map typically preserves the inner product between all states of sub-exponential complexity, which we will see in the following two sections is sufficient to ensure that the effective description is valid for any observable of sub-exponential complexity.  For example if $O_R$ is a radiation observable of sub-exponential complexity, which for now we'll just define by saying it is a linear combination of two sub-exponential unitaries (see section \ref{measurementsec} for a better definition which implies this one), and $|\psi\ran\in \HL\otimes\Hl\otimes \Hr\otimes\HR$ is a sub-exponential state in the effective description, then we have
\begin{align}\nonumber
\lan \psi|(V^\dagger \otimes I_{LR})O_R(V\otimes I_{LR})|\psi\ran&=\lan \psi|(V^\dagger V\otimes I_{LR})O_R|\psi\ran\\
&\approx\lan \psi|O_R|\psi\ran.\label{ORexp}
\end{align}
In the second line we used equation \eqref{subexpconc}, and also that $O_R|\psi\ran$ is a linear combination of two sub-exponential states. On the other hand we know that the effective description cannot really give an accurate picture for all observables: if it did, then Hawking's prediction of information loss would be correct.  Mathematically there is indeed a large deviation from Hawking's description: at late times there is a large number of null states which are annihilated by $V$.  One place where we can surely detect this is in expectation values of exponentially complex observables.  The great insight of Page however was that we can get at the same physics much more simply by studying entropies instead of observables \cite{Page:1993wv}.  In this section we therefore study the relationship between the entropy of the radiation system $R$ in the fundamental description and its entropy in effective description.  We emphasize that these entropies do not need to be the same: conjugation by $V\otimes I_{LR}$ need not preserve the state on $R$ (or its entropy) since $V$ is not an isometry.  We will find that the entropy in the fundamental description is computed by the quantum extremal surface formula in the effective description, which we view as giving a Hilbert space interpretation to the ```Page curve'' computation of \cite{Penington:2019npb,Almheiri:2019psf}.\footnote{Some of the results of this section can be understood as consequences of the random tensor network formalism of \cite{Hayden:2016cfa}, but we find it convenient to develop things from scratch.}

We begin by introducing what we call the \textit{Hawking state}
\be\label{HawkS}
|\psi_{Hawk}\ran=|\chin\ran_{L\ell}\otimes |\chot\ran_{rR},
\ee
which we use to model the state in the effective description of a black hole which has been evaporating for some time. We can think of $|\chin\ran$ as being the purification of the (possibly mixed) state of the infalling matter which created the black hole onto a reference system $L$ and $|\chot\ran$ as the entangled state of the interior and exterior outgoing Hawking modes.  We will sometimes assume that these states are highly entangled in the sense that $S(\chin_\ell)$ and $S(\chot_R)$ are of order $\log |L|$ and $\log |R|$ respectively.  We define the notation
\be\label{hawkingencode}
\Psi_{LB R}(U)\equiv (V\otimes I_{LR})|\ph\ran\lan\ph|(V^\dagger\otimes I_{LR})
\ee
for the encoded version of the Hawking state, 
and our goal is to compute the Renyi and von Neumann entropies of $\Psi_R(U)$.

\bfig
\includegraphics[height=7cm]{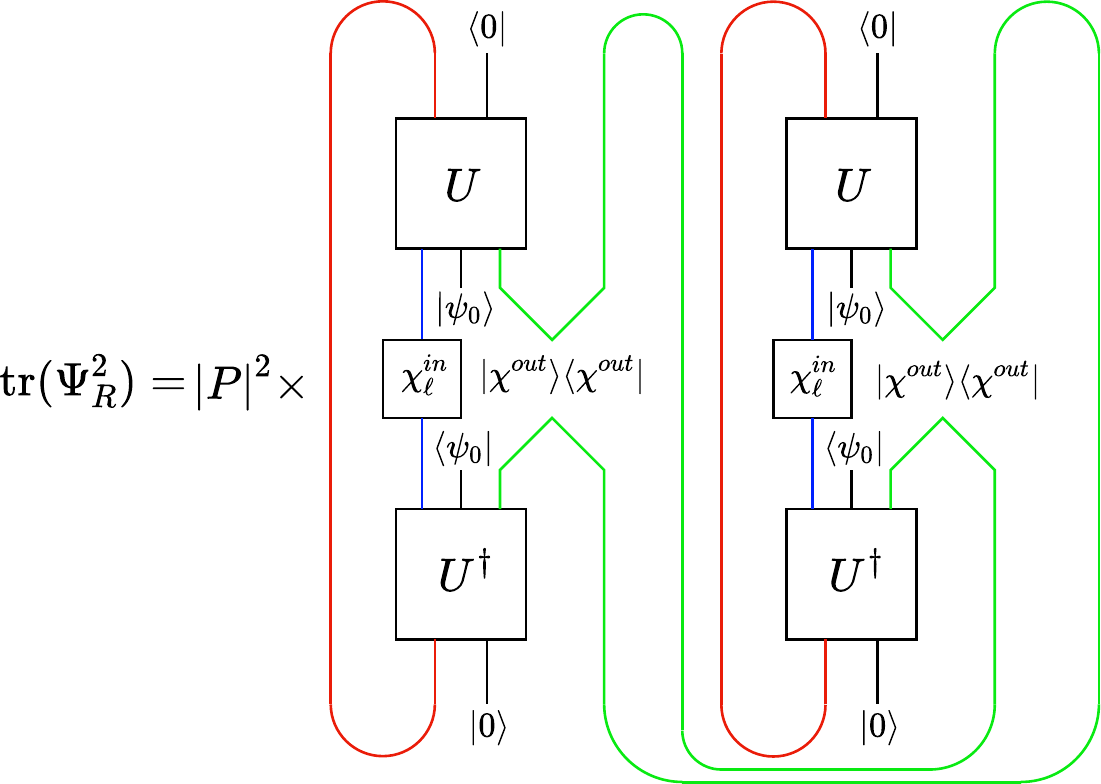}
\caption{Computing the second Renyi entropy of the radiation.  $B$ indices are shaded red, $L,\ell$ indices are shaded blue, and $r,R$ indices are shaded green.  The four different ways of contracting the indices through the unitaries as in \eqref{Uresults} lead to the four terms in \eqref{renyi2av}.}\label{renyi2inoutfig}
\efig
We'll first compute the second Renyi entropy of the radiation.  Using the second equation from \eqref{Uresults} we have (see figure \ref{renyi2inoutfig})
\begin{align}\nonumber
\int dU e^{-S_2\left(\Psi_R(U)\right)}&=\int dU \tr\left(\Psi_R(U)^2\right)\\
&=\frac{|P|^2|B|^2}{|P|^2|B|^2-1}\Bigg[\left(1-\frac{1}{|P||B|^2}\right)e^{-S_2\left(\chot_R\right)}+\left(1-\frac{1}{|P|}\right)\frac{e^{-S_2\left(\chin_\ell\right)}}{|B|}\Bigg],\label{renyi2av}
\end{align}
so when $|B|,|P|\gg 1$ we have\footnote{One can use the method we introduce below to compute the higher Renyi entropies to show that the fluctuations of $S_2(\Psi_R)$ about its average are small, so \eqref{renyi2R} is typically true for each particular $U$.}
\be\label{renyi2R}
S_2(\Psi_R)\approx \min\Big[S_2\left(\chot_R\right),\log |B|+S_2\left(\chin_\ell\right)\Big].
\ee

\bfig
\includegraphics[height=4.5cm]{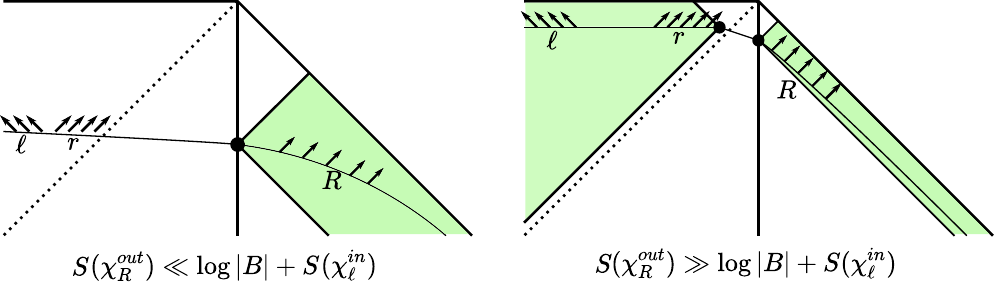}
\caption{The quantum extremal surface calculation of the radiation entropy: at early times the entanglement wedge of the radiation includes only the domain of dependence of $R$, shaded green in the left diagram,  while at late times it also includes an ``island'' in the black hole interior.  The ``rim'' of the island, shown as the second black dot in the right diagram, is the new quantum extremal surface found in \cite{Penington:2019npb,Almheiri:2019psf}.  It has an area which is approximately that of the horizon on this time slice, and thus has $\frac{\mathrm{Area}}{4}\approx \log |B|$.}\label{evapbhislandfig}
\efig
Except for being a Renyi entropy instead of the von Neumann entropy, \eqref{renyi2R} is just what is expected from the quantum extremal surface formula \cite{Penington:2019npb,Almheiri:2019psf}. At early times the entanglement wedge of the radiation $R$ consists just of the radiation, so its generalized entropy matches that of the radiation in the Hawking state \eqref{HawkS}, while at late times this entanglement wedge also includes an ``island'' containing $\ell$ and $r$, so its generalized entropy is given by the area in Planck units of the rim of the island, which is just $\log |B|$, plus the entropy of $\ell$ in the Hawking state ($r$ and $R$ purify each other and don't contribute).  See figure \ref{evapbhislandfig} for an illustration.  The crossover in \eqref{renyi2R} from $S_2\left(\chot_R\right)$ to $\log |B|+S_2\left(\chin_\ell\right)$ when the former becomes larger than the latter is precisely what is needed to ensure consistency with the purity of the total state $\Psi_{LBR}$.

\bfig
\includegraphics[height=8cm]{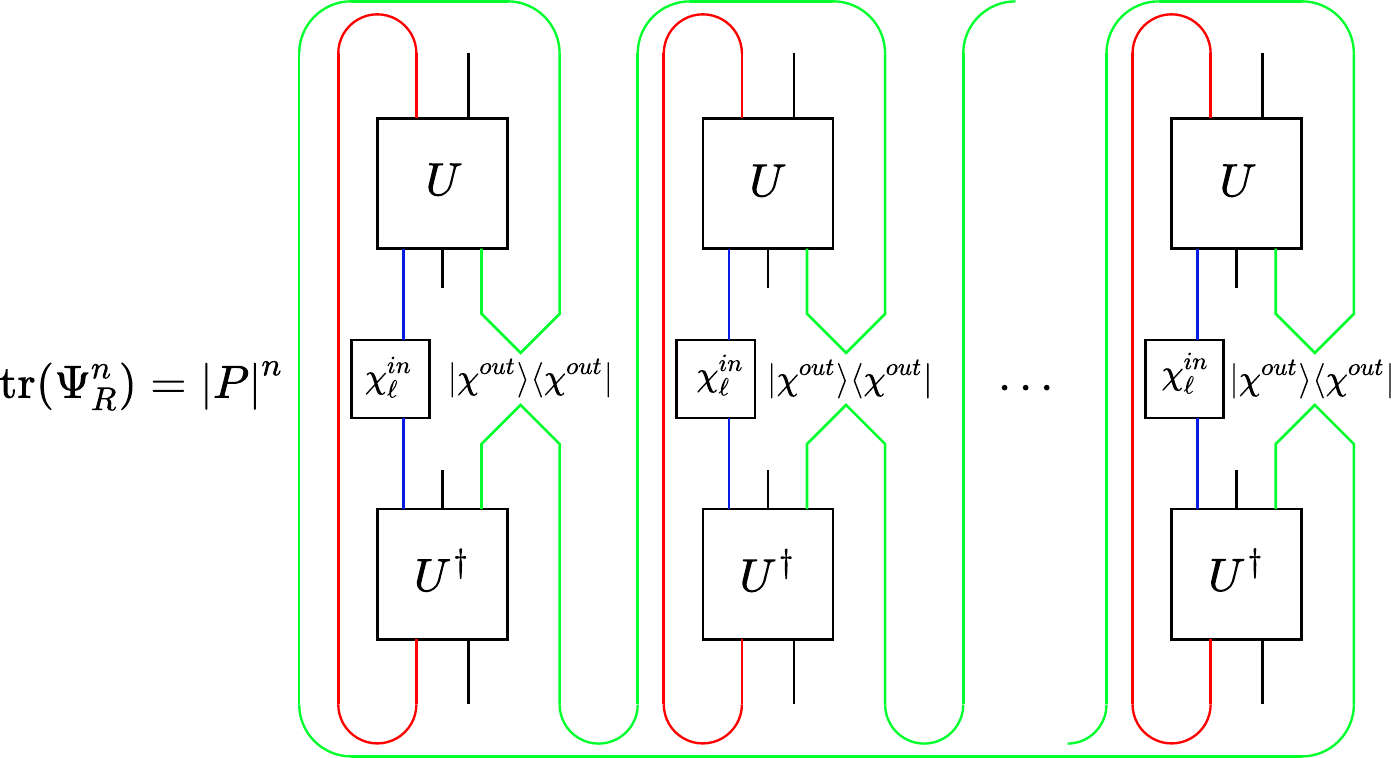}
\caption{Computing the $n$th Renyi entropy.  There are two terms which can dominate: the one which contracts the ingoing/outgoing indices for each $U$ with the outoing/ingoing indices of the $U^\dagger$ just below it, which maximizes the number of red loops, and the one which contracts the ingoing/outgoing indices of each $U$ with the outoing/ingoing indices of the $U^\dagger$ which is in the previous column (wrapping around to the last $U^\dagger$ for the first $U$), which maximizes the number of green loops.}\label{nrenyifig}
\efig
To really match the quantum extremal surface result however, we need to compute the von Neumann entropy instead of the second Renyi entropy.  We can do this by computing the higher Renyis
\be
S_n(\Psi_R)\equiv -\frac{1}{n-1}\log \tr \left(\Psi_R^n\right)
\ee
and then using (subject to the usual caveats about the replica trick)
\be
S(\Psi_R)=\lim_{n\to 1^+}S_n(\Psi_R).
\ee
In principle we can compute the average of $\tr\left(\Psi_R^n\right)$ exactly using the known expression \eqref{Wgresult} for the Weingarten function, but is more instructive to first assume that all dimensionalities are large.  We can then use the asymptotic expression \eqref{Wgasymp}, which tells us that we only need to consider ``diagonal'' contractions in figure \ref{nrenyifig} where the ingoing and outgoing indices of each $U$ are contracted with outgoing and ingoing indices of the same $U^\dagger$.\footnote{One might worry that loop contractions could enhance subleading terms in \eqref{Wgasymp}, but we never have $P$-loops so these subleading terms will at least be suppressed by powers of $|P|$, as in the fourth term in \eqref{renyi2av}.}  There are still $n!$ ways of doing these ``diagonal'' contractions, but if we assume that either $S_n(\chot_R)\gg \log |B|+S_n(\chin_\ell)$ or  $S_n(\chot_R)\ll \log |B|+S_n(\chin_\ell)$, and also assume that $|\chin\ran$ and $|\chot\ran$ are highly entangled, then there are only two contractions which can be dominant: the one which maximizes the number of $B$-loops, which dominates at early times, and the one which maximizes the number of $R$ loops, which dominates at late times (see figure \ref{nrenyifig}). For example the terms in the average of $\tr(\Psi_R^3)$ which arise from diagonal contractions are (see \eqref{tripleWg})
\begin{align}\nonumber
\int dU \tr(\Psi_R^3)\supset \frac{|P|^2|B|^2(|P|^2|B|^2-2)}{(|P|^2|B|^2-1)(|P|^2|B|^2-4)}\Bigg[&e^{-2S_3\left(\chot_R\right)}+\frac{3}{|B|} e^{-S_2\left(\chot_R\right)-S_2\left(\chin_\ell\right)}\\
&+\frac{1}{|B|^2}\left(e^{-2S_3\left(\chot_R\right)-2S_3\left(\chin_\ell\right)}+e^{-2S_3\left(\chin_\ell\right)}\right)\Bigg].
\end{align}
The prefactor rapidly approaches 1 when $|P|,|B|\gg 1$, and when $|\chin\ran$ and $|\chot\ran$ are highly entangled only the first or last term can dominate. Thus in this approximation we have
\be \label{nrenyiexp}
\int dU e^{-(n-1)S_n(\Psi_R(U))}\approx \max\left[\tr\left(\left(\chot_R\right)^n\right), \frac{1}{|B|^{n-1}}\tr\left(\left(\chin_\ell\right)^n\right) \right],
\ee
and therefore
\be\label{nrenyis}
S_n(\Psi_R(U))\approx \min\left[S_n(\chot_R),\log|B|+S_n(\chin_\ell)\right]
\ee
for the Renyi entropies and 
\be\label{pageresult}
S(\Psi_R(U))\approx \min\left[S(\chot_R),\log|B|+S(\chin_\ell)\right]
\ee
for the von Neumann entropy.  This last result is precisely what is expected from the quantum extremal surface calculation, see figure \ref{evapbhislandfig}.  The fact that the same formula holds for the Renyi entropies as well, in particular with the ``area'' term $\log |B|$ being independent of $n$, shows that our model has ``fixed area'' in the sense of \cite{Akers:2018fow,Dong:2018seb}.\footnote{We could easily generalize our model to describe non fixed-area states, for example simply by generalizing $V$ to have a direct sum structure $V = \oplus_\alpha V_\alpha$ with each $V_\alpha$ defined like our $V$, and with input and output Hilbert spaces also block decomposing like $\mathcal{H} = \oplus_\alpha \mathcal{H}_{\alpha}$. Each $\alpha$-block would correspond to a different area eigensector, and the $\log |B|$ term in \eqref{nrenyis} would pick up an $n$-dependence from the differing support on the different sectors. However this is not an important aspect of the physics, since semiclassical states can be approximated to exponential accuracy by states where the area  is fixed to leading order.}

It is instructive to see how these same entropies arise in the phase model \eqref{phasedef}.  Choosing bases for $\HL$, $\Hl$, $\Hr$, and $\HR$ such that
\begin{align}\nonumber
|\chin\ran&=\sum_i\sqrt{C_i}|i\ran_L|i\ran_\ell\\
|\chot\ran&=\sum_m\sqrt{D_m}|m\ran_r|m\ran_R
\end{align}
for some non-negative constants $C_i, D_m$, from \eqref{phasedef} we have
\begin{align} \nonumber
\Psi_R&= \sum_{m, m'}\left(\sum_{i=1}^{|L|} C_i \sqrt{D_m D_{m'}} \braket{im'|V_{phase}^{\dagger}V_{phase}|im}\right)\, \ket{m}\bra{m'}_R  \\\label{psir1}
&=\frac{1}{|B|}\sum_{i,b,m,m'}C_i\sqrt{D_mD_{m'}}e^{i\theta(i,m;b)-i\theta(i,m';b)}|m\ran\lan m'|_R.
\end{align} 
The second Renyi is given by
\be
e^{-S_2(\Psi_R)}=\frac{1}{|B|^2}\sum_{i,i',b,b',m,m'}C_i C_{i'}D_mD_{m'}e^{i\left(\theta(i,m;b)-\theta(i,m';b)+\theta(i',m';b')-\theta(i',m;b')\right)}.
\ee
Generically the dominant contributions to this sum come when the phases cancel in pairs.  There are two ways to do this: we can have $m=m'$, or we can have $(i,b)=(i',b')$.  Including only these contributions we thus have
\be
e^{-S_2(\Psi_R)}\approx e^{-S_2\left(\chot_R\right)}+\frac{1}{|B|}e^{-S_2\left(\chin_\ell\right)},
\label{s2terms}
\ee
which is equivalent to \eqref{renyi2R}.  As discussed in the introduction, the second term exists only because of the non-isometry of $V_{phase}$. One way to think about it is as arising from small but non-zero overlaps  in the images of orthogonal sub-exponential states in the effective description after being mapped to the fundamental picture, as in the second line of \eqref{phaseoverlap}. These overlaps lead to non-zero off-diagonal elements in \eqref{psir1}. While these off-diagonal elements are individually suppressed by a factor of $O(1/\sqrt{|B||L|})$  relative to the diagonal ones (assuming for simplicity that $C_i \sim 1/L$), there are $O(|R|^2)$ of them contributing to  $e^{-S_2(\Psi_R)}=\sum_{ij} |(\Psi_R)_{ij}|^2$. Hence, their total contribution in the second term of \eqref{s2terms} can be comparable to the first term coming from diagonal elements.  

We can also compute the higher Renyi entropies
\begin{align}\nonumber
e^{-(n-1)S_n(\Psi_R)}=\frac{1}{|B|^n}\sum_{i_1,\ldots,i_n,b_1,\ldots, b_n,m_1,\ldots m_n}&C_{i_1}\ldots C_{i_n}D_{i_1}\ldots D_{i_n}\\
&\hspace{-1.5cm}\times e^{i\left(\theta(i_1,m_1;b_1)-\theta(i_1,m_2;b_1)+\ldots+ \theta(i_n,m_n;b_n)-\theta(i_n,m_1;b_n)\right)},
\end{align}
which are dominated either by the terms where $m_1=m_2=\ldots=m_n$ or the terms where $(i_1,b_1)=(i_2,b_2)=\ldots=(i_n,b_n)$, leading to
\be
e^{-(n-1)S_n(\Psi_R)}\approx e^{-(n-1)S_n\left(\chot_R\right)}+\frac{1}{|B|^{n-1}}e^{-(n-1)S_n\left(\chin_L\right)},
\ee
which is equivalent to \eqref{nrenyis}.

An alternative way of understanding the contributions to the $n$-th Renyi entropy which restore unitarity, such as the second term in \eqref{s2terms}, is in terms of the ``equilibrium approximation'' of \cite{eqapprox}. It was argued in \cite{eqapprox} that if under the evolution of a chaotic quantum many-body system, a time-evolved pure state macroscopically resembles an equilibrium density matrix $\Phi$ \footnote{For example, the expectation values of few-body observables in the two states are close.}, then the $n$-th Renyi entropy of the pure state at late times can also be expressed in terms of $\Phi$ using a set of rules explained in \cite{eqapprox}.  In our model we can obtain a guess for the equilibrium state by averaging the encoded Hawking state,
\be
\Phi_{LBR}\equiv \int dU \Psi_{LBR}(U)=\chin_L\otimes \frac{I_B}{|B|}\otimes \chot_R,\label{Phidef}
\ee
and in section \ref{simpentsec} below we will show that all sub-exponential observables in the fundamental description have expectation values in $\Psi_{LBR}(U)$ which agree with those in $\Phi_{LBR}$ to exponential precision.\footnote{In section \ref{simpentsec} we don't include $L$, but it can be restored by replacing $R\to LR$ and $|\chot\ran_{rR}\to |\chin\ran_{L\ell}\otimes |\chot\ran_{rR}$ in the calculations.}  We have checked that applying the rules of the equilibrium approximation to $\Phi_{LBR}$ indeed gives Renyi entropies that agree with those we have computed in this section.

\section{Reconstruction of effective field theory operators}\label{reconstructionsec}
We now turn to the question of how gravitational effective field theory operators can be represented in the fundamental description, a process which is called reconstruction.  In this section we will study this as a purely mathematical problem; the results will then be used in the following section to formulate the measurement theory of an observer near/in the black hole in terms of the fundamental degrees of freedom. 

In a conventional quantum error-correcting code with isometric encoding $V$, there is a canonical way to represent logical operators using physical degrees of freedom: for any logical operator $O$ we define an encoded operator
\be\label{canO}
\wt{O}\equiv VOV^\dagger,
\ee
in terms of which for any logical state $|\psi_1\ran$ we have
\be\label{stateaction}
\wt{O}V|\psi_1\ran=VO|\psi_1\ran.
\ee
Thus it does not matter whether we first encode and then act with the operator or first act with the operator and then encode.  Moreover if $|\psi_1\ran,|\psi_2\ran$ are two logical states then we have
\be\label{codematrix}
\lan \psi_2|V^\dagger \wt{O}V|\psi_1\ran=\lan \psi_2|O|\psi_1\ran.
\ee

The validity of \eqref{stateaction} and \eqref{codematrix} crucially relies on the fact that $V^\dagger V=I$, so if $V$ is not an isometry then it is not clear what to expect.  In general we will say that an operator $\wt{O}$ gives a reconstruction of a logical operator $O$ if \eqref{stateaction} and \eqref{codematrix} both hold for $|\psi_1\ran,|\psi_2\ran$ in some appropriate set of states, and we will say that $\wt{O}$ is an approximate reconstruction of $O$ if \eqref{stateaction} and \eqref{codematrix} hold in some approximation.  In the remainder of this section we will explain what kinds of families of states and in what approximations this can be done.  

\subsection{Attempts at simple reconstruction}
In this subsection we will try two simple ideas for how to do operator reconstruction in the code defined by \eqref{Vdef}.  In both cases we will run into some complications, which motivate us to consider a more restricted kind of reconstruction in the following subsection.

We'll first consider operators $O_{LR}$ with support only on $L$ and $R$.  The situation starts out well: for any $|\psi\ran\in \Hall$ we have
\be
O_{LR}\left(V\otimes I_{LR}\right)|\psi\ran=\left(V\otimes I_{LR}\right)O_{LR}|\psi\ran,
\ee
so \eqref{stateaction} holds provided we use $O_{LR}$ as its own reconstruction.  \eqref{codematrix} however is more subtle: for general states $|\psi\ran_1,|\psi_2\ran\in \HL\otimes \Hl\otimes \Hr \otimes \HR$ and general operators $O_{LR}$ it will \textit{not} hold that $\lan \psi_2|(V^\dagger\otimes I_{LR})O_{LR}(V\otimes I_{LR}) |\psi_1\ran\approx \lan \psi_2| O_{LR}|\psi_1\ran$.  For example at late times the entanglement entropy of the Hawking state $|\ph\ran$ on $R$ is much larger than the entanglement entropy of its encoded image $(V\otimes I_{LR})|\ph\ran$, so by choosing $O_{LR}$ to be an appropriate projection we can arrange for the expectation value of $O_{LR}$ in $V|\ph\ran$ to be much smaller than its expectation value in $|\ph\ran$.  On the other hand things are much better if we restrict to sub-exponential operators $O_{LR}$ and sub-exponential states.  Sub-exponential unitaries and sub-exponential states we have already defined in section \ref{subexpsec}, and in the beginning of section \ref{pagesec} we provisionally defined a sub-exponential observable to be an observable which is a linear combination of two sub-exponential unitaries.  Any operator is a linear combination of two hermitian operators (its hermitian and anti-hermitian parts), so here we provisionally define a general sub-exponential operator to be one which is a linear combination of four sub-exponential unitaries.  We then have
\begin{align}\nonumber
\lan \psi_2|(V^\dagger\otimes I_{LR})O_{LR}(V\otimes I_{LR})|\psi_1\ran&=\lan \psi_2|(V^\dagger\otimes I_{LR})(V\otimes I_{LR})O_{LR}|\psi_1\ran\\
&\approx \lan \psi_2|O_{LR}|\psi_1\ran,\label{simpleR}
\end{align}
where the approximation in the last line is to exponential precision in the black hole entropy and the result follows from equation \ref{subexpconc} and the fact that $O_{LR}|\psi_1\ran$ is a linear combination of four sub-exponential states.  Thus we see that $O_{LR}$ gives a good reconstruction of itself \textit{provided} that 1) it is sub-exponential and 2) we use it only on sub-exponential states. 

What about operators $O_{\ell r}$ that act on the interior degrees of freedom?  To simplify calculations we can take $O_{\ell r}$ to be traceless, as we can always restore its trace by adding a multiple of the identity, and we surely know how to reconstruct the identity.  Let's first try using the canonical reconstruction $V O_{\ell r}V^\dagger \otimes I_{LR}$.  We can test its (approximate) validity by computing
\begin{align}\nonumber
\int dU ||(VO_{\ell r}V^\dagger\otimes I_{LR})(V\otimes I_{LR})|\psi\ran-VO_{\ell r}\otimes I_{LR}|\psi\ran||^2=&\int dU \tr (\psi_{\ell r}V^\dagger V O_{\ell r}^\dagger V^\dagger V O_{\ell r}V^\dagger V)\\\nonumber
&-2\mathrm{Re}\int dU \tr(\psi_{\ell r}O_{\ell r}^\dagger V^\dagger VO_{\ell r}V^\dagger V)\\
&+\lan \psi|(O_{\ell r}^\dagger O_{\ell r}\otimes I_{LR})|\psi\ran.
\end{align}
Using \eqref{Uresults} it is straightforward to show that 
\be
\int dU \tr(\psi_{\ell r}O_{\ell r}^\dagger V^\dagger VO_{\ell r}V^\dagger V)=\frac{|P|^2|B|^2}{|P|^2|B|^2-1}\left(1-\frac{1}{|P||B|^2}\right)\lan \psi|(O_{\ell r}^\dagger O_{\ell r}\otimes I_{LR})|\psi\ran,
\ee
so the nontrivial task is to compute the term involving three copies of $V^\dagger V$.  As in our computation of the Renyi entropies, when all dimensionalities are large the average will be dominated by the diagonal contractions.  Computing these using \eqref{tripleWg} we have
\begin{align}\nonumber
\int dU \tr (\psi_{\ell r}V^\dagger V O_{\ell r}^\dagger V^\dagger V O_{\ell r}V^\dagger V)=\frac{|P|^2|B|^2(|P|^2|B|^2-2)}{(|P|^2|B|^2-1)(|P|^2|B|^2-4)}\Bigg[&\lan \psi|(O_{\ell r}^\dagger O_{\ell r}\otimes I_{LR})|\psi\ran\\\nonumber
&+\frac{1}{|B|}\tr(O_{\ell r}^\dagger O_{\ell r})\\\nonumber
&+\frac{\lan \psi|(O_{\ell r} O_{\ell r}^\dagger\otimes I_{LR})|\psi\ran}{|B|^2}\Bigg]\\
&+\ldots,
\end{align}
where ``$\ldots$'' indicates nondiagonal contractions, and thus
\begin{align}\nonumber
\int dU ||(VO_{\ell r}V^\dagger\otimes I_{LR})(V\otimes I_{LR})|\psi\ran-VO_{\ell r}\otimes I_{LR}|\psi\ran||^2&\approx\frac{1}{|B|}\tr(O_{\ell r}^\dagger O_{\ell r})\\
&\leq ||O_{\ell r}||_\infty^2 \frac{|L||R|}{|B|}.
\end{align}
The terms which have been neglected here are at most of order $||O_{\ell r}||^2_\infty/|B|^2$.  Thus we see that when $|L||R|\ll |B|$, the canonical reconstruction $VO_{\ell r}V^\dagger\otimes I_{LR}$ acts in the right way on encoded states up to an exponentially small error.  This is not too surprising, as in this regime $V$ is close to being an isometry.  On the other hand if $|L||R|\gg |B|$, so that there are many null states and $V$ is far from being an isometry, then we see that the canonical reconstruction breaks down and some new idea is needed.  This breakdown is natural from the point of view of entanglement wedge reconstruction: at late times $O_{\ell r}$ lies in the entanglement wedge of the radiation (see figure \ref{evapbhislandfig}), so it would be a surprising if we could reconstruct it on $B$. 

\subsection{State-specific reconstruction} \label{sec:SSR}

We just saw that the canonical reconstruction $\wt{O}$ of operators $O_{\ell r}$ does not work when $V$ is highly non-isometric, even if we restrict to sub-exponential $O_{\ell r}$ and only act on sub-exponential states $|\psi\ran$.  One might hope that we could come up with some other reconstruction which does work, but if we want the same reconstruction to work on all sub-exponential states then in the non-isometric regime it turns out this is impossible:

\begin{thm}\label{nogo5}
    Let $\mathcal{H}_\ell, \mathcal{H}_r, \mathcal{H}_B$, $\mathcal{H}_L$, and $\mathcal{H}_R$ be finite-dimensional Hilbert spaces, which for this theorem we take to be tensor products of qubits, and let $V: \mathcal{H}_\ell \otimes \mathcal{H}_r \to \mathcal{H}_B$ a linear map.  Furthermore let $\log(|L||\ell||r||R|)$ be sub-exponential in $\log|B|$.  Assume that for every sub-exponential (in $\log |B|$) unitary $W_{L \ell r R}$, there exists a unitary $\wt{W}_{LBR}$ such that for all sub-exponential (in $\log |B|$) states  $\ket{\psi}_{L\ell r R}$ we have
    \begin{equation}\label{Wbound1}
	\lVert \wt{W}_{LBR} V \ket{\psi} - V W_{L \ell r R} \ket{\psi} \rVert \le \epsilon~.
    \end{equation}
    Then
    \begin{equation}\label{512}
        |B| \ge |\ell||r| (1 - \frac{4 \epsilon}{1+\delta} \log |\ell||r||L||R|) ~,
    \end{equation}
    where
    \be
    1+\delta \equiv ||(V\otimes I_{LR})|\mathrm{MAX}\ran_{L\ell r R,\ol{L}\ol{\ell}\ol{r}\ol{R}}||,
    \ee
    with $|\mathrm{MAX}\ran_{L\ell r R,\ol{L}\ol{\ell}\ol{r}\ol{R}}$ being the canonical maximally-entangled state in the computational basis for $L\ell r R$.  Moreover \eqref{512} still holds if \eqref{Wbound1} is only required to hold when $W_{L\ell r R}$ is a single-site Pauli operator.
\end{thm}
Hence, if we want the errors to be suppressed as $\epsilon,\delta\leq |B|^{-\gamma}$ for some $\gamma>0$, then since we require that $|L||R| |\ell||r|$ is sub-exponential in $|B|$, \eqref{512} implies that $|B|$ must be $\gtrsim |\ell||r|$.
The idea behind this theorem is that a reconstruction that works on all sub-exponential states would work well on an entire basis, and therefore also on the maximally entangled state. But if the linear map has a large kernel, then the unitary reconstruction on that entangled state would be able to dramatically change the state of the purifying reference. We give the full proof in appendix \ref{app:rec_no_go}.

We thus need to lower our hopes in terms of how powerful reconstructions of interior operators can be.  Since we can't have a single reconstruction that works on all sub-exponential states, and any particular subset of the sub-exponential states would be arbitrary, we might as well consider the smallest nontrivial domain of reconstruction: given a logical operator $O$, an encoding map $V$, and a particular state $|\psi\ran$, we can ask for a reconstruction $\wt{O}$ such that \eqref{stateaction} holds.
In \cite{Akers:2021fut} this was called ``state-specific'' reconstruction, since the reconstruction $\wt{O}$ is allowed to depend on the state $|\psi\ran$ and there is no guarantee that it will work on any other state.\footnote{The idea that the reconstruction of interior operators should be allowed to depend on the state which is being considered has a long history, and was most forcefully advocated in \cite{Papadodimas:2013jku,Papadodimas:2015jra}.  In section \ref{discussion} we discuss in more detail the relationship between that proposal and what we do here.}  In order for this very limited form of reconstruction to be useful however it is important to add an additional constraint: we require that if the logical operator $O$ is unitary then its reconstruction $\wt{O}$ is also unitary \cite{Akers:2021fut}.  The reason that this requirement is important is that unitary transformations represent physical operations that can be performed by an observer in the effective description, and we would like these to map to physical operations in the fundamental description.  Moreover without this requirement there is too much ambiguity in how to think about the support of the reconstructed operator, and state-specific reconstruction becomes too trivial to be interesting.  For example in quantum field theory, given any finite-energy states $|\psi_1\ran,|\psi_2\ran$, the Reeh-Schlieder theorem tells us that in any spatial region we can always find an operator whose action on $|\psi_1\ran$ gives a state which is arbitrarily close to $|\psi_2\ran$. On the other hand this operator will usually not be unitary, since after all we shouldn't be able to create the moon instantaneously by doing something on Earth.  

Approximate state-specific reconstruction is characterized by the following theorem, which is a generalization of theorem 3.7 from \cite{Akers:2021fut} to non-isometric situations:
\begin{thm}\label{dcthm}
Let $\mathcal{H}_a$, $\mathcal{H}_B$, and $\mathcal{H}_{\ol{B}}$ be finite-dimensional Hilbert spaces, $L:\mathcal{H}_a\to\mathcal{H}_B\otimes\mathcal{H}_{\ol{B}}$ a linear map, $W$ a unitary operator on $\mathcal{H}_a$, and $|\psi\ran$ an element of $\mathcal{H}_a$.  Then the following conditions are equivalent:
\bi
\item[(1)] There exists a unitary operator $W_B$ on $\mathcal{H}_B$ such that
\be
||W_BL|\psi\ran-LW|\psi\ran||\leq \epsilon_1
\ee
\item[(2)] We have the decoupling condition
\be\label{decoupling}
||\tr_B(LW|\psi\ran\lan\psi|W^\dagger L^\dagger)-\tr_B(L|\psi\ran\lan\psi| L^\dagger)||_1\leq \epsilon_2.
\ee
\ei
The equivalence is that  (1)$\implies$(2) with $\epsilon_2\leq \sqrt{2(\lan \psi|L^\dagger L|\psi\ran+\lan\psi|W^\dagger L^\dagger L W|\psi\ran)}\epsilon_1$ and (2)$\implies$(1) with $\epsilon_1\leq \sqrt{\epsilon_2}$.  Here $||X||_1\equiv \tr(\sqrt{X^\dagger X})$ is the trace norm of $X$.
\end{thm}
This is an example of what is usually called a ``decoupling theorem'': it shows that in order for us to be able to approximately reconstruct $W$ on $B$, it must be that $\ol{B}$ has little information about whether or not we act with $W$.  More precisely, by the Holevo-Helstrom theorem condition (2) is equivalent to the statement that there is no measurement on $\ol{B}$ that can tell whether or not we acted with $W$ with success probability greater than $\frac{1}{2}+\frac{\epsilon_2}{4}$ (see e.g. theorem 3.4 in \cite{watrous_2018}).  The proof of theorem \ref{dcthm} is given in appendix \ref{proofapp2}.

Our first application of theorem \ref{dcthm} will be to show that for any sub-exponential state $|\psi\ran\in \Hall$ and any sub-exponential unitary $W$ on $\Hall$, there is a state-specific reconstruction $W_{LBR}$ of $W$ on $\HL\otimes\HB\otimes\HR$ that works to exponential precision in $\log |B|$.  By theorem \ref{dcthm} this will follow if we can show that
\be
|\lan\psi|W^\dagger(V^\dagger V\otimes I_{LR})W|\psi\ran-\lan\psi|(V^\dagger V\otimes I_{LR})|\psi\ran|\leq \epsilon
\ee
for some $\epsilon$ which is exponentially small in $\log|B|$.  This however follows immediately from the triangle inequality and equation \eqref{subexpconc}: we are very likely to have
\begin{align}\nonumber
\Big|\lan\psi|W^\dagger(V^\dagger V\otimes I_{LR})W|\psi\ran-\lan\psi|(V^\dagger V\otimes I_{LR})|\psi\ran\Big|&\leq\Big|\lan \psi|W^\dagger(V^\dagger V\otimes I_{LR})W|\psi\ran-\lan \psi|W^\dagger W|\psi\ran\Big|\\\nonumber
&\phantom{\leq}+\Big|\lan\psi|(V^\dagger V\otimes I_{LR})|\psi\ran-\lan \psi|\psi\ran\Big|\\
&\leq 2\sqrt{18}|B|^{-\gamma}
\end{align}
with $0<\gamma<\frac{1}{2}$.

This however is a somewhat uninteresting kind of reconstruction: the reference system $L$ need not have a physical interpretation, so we are more interested in reconstructing sub-exponential unitaries $W_{\ell rR}$ with support only on $\Hl\otimes\Hr\otimes\HR$ and we'd like them to have reconstructions $W_{BR}$ with support just on $\HB\otimes\HR$.  Defining
\be\label{LBRdef}
\Psi_{LBR}(U,W)=(V\otimes I_{LR})W|\psi\ran\lan \psi|W^\dagger(V^\dagger\otimes I_{LR}),
\ee
with $|\psi\ran$ a sub-exponential state, we'd like to show that
\be
||\Psi_L(U,W_{\ell rR})-\Psi_L(U,I)||_1\leq \epsilon
\ee
for some $\epsilon$ which is exponential small in $\log |B|$.  We will first show that on average this is true.  We can observe that 
\begin{align}\nonumber
\int dU ||\Psi_L(U,W_{\ell r R})-\Psi_L(U,I)||_1&\leq \sqrt{|L|}\int dU ||\Psi_L(U,W_{\ell r R})-\Psi_L(U,I)||_2\\
&\leq \sqrt{|L|\int dU ||\Psi_L(U,W_{\ell r R})-\Psi_L(U,I)||_2^2},
\end{align}
with the first inequality following from $||X_A||_1\leq \sqrt{|A|}||X_A||_2$ (see \eqref{normineq}) and the second following from Jensen's inequality.  Our usual unitary integration technology then gives
\begin{align}\nonumber
\int dU ||\Psi_L(U,W_{\ell r R})-\Psi_L(U,I)||_2^2&=\int dU \tr_L\left(\Psi_L(U,W_{\ell r R})^2-2\Psi_L(U,W_{\ell r R})\Psi_L(U,I)+\Psi_L(U,I)^2\right)\\\nonumber
&=\frac{|P|^2|B|^2}{|P|^2|B|^2-1}\left(1-\frac{1}{|P|}\right)\frac{2}{|B|}\Bigg[e^{-S_2\left(\tr_{L\ell r}(W_{\ell r R}|\psi\ran\lan\psi|W_{\ell r R}^\dagger)\right)}\\\nonumber
&\hspace{5.5cm}-||\tr_{\ell r}(W_{\ell r R} \psi_{\ell r R})||_2^2\Bigg]\\
&\leq \frac{4}{|B|}e^{-S_2\left(\tr_{L\ell r}(W_{\ell r R}|\psi\ran\lan\psi|W_{\ell r R}^\dagger)\right)},
\end{align}
and thus
\be
\int dU ||\Psi_L(U,W_{\ell r R})-\Psi_L(U,I)||_1\leq 2\sqrt{\frac{|L|}{|B|}}.\label{Ldcbound}
\ee
So far we have not assumed much about the relative sizes of $|B|$ and $|L|$, but $L$ is designed to keep track of whatever matter we threw into the black hole to create it, and at least as long as this process was not adiabatic (i.e. it happened quickly) then we have $|L|\ll |B|$ and so we can apply theorem \ref{dcthm} to conclude that we can likely give a state-specific reconstruction on $BR$ for any particular sub-exponential $W_{\ell rR}$ on a particular sub-exponential state $|\psi\ran$.

As in our discussion of the overlap, we'd like to use measure concentration to strengthen the conclusion of the previous paragraph from likely applying to any particular sub-exponential $W_{\ell rR}$ and $|\psi\ran$ to likely applying for \textit{all} sub-exponential $W_{\ell rR}$ and $|\psi\ran$.  Based on our experience deriving \eqref{subexpconc}, the natural way to attempt this would be to view $||\Psi_L(U,W_{\ell r R})-\Psi_L(U,I)||_1$ as a Lipschitz function of $U$ and then apply Lemma \ref{Meckes} and equations \eqref{subexpcount}, \eqref{Ldcbound} to conclude that all sub-exponential states and unitaries are likely to obey the decoupling theorem \ref{dcthm}.  Unfortunately however $||\Psi_L(U,W_{\ell r R})-\Psi_L(U,I)||_1$ does not have a nice enough Lipschitz constant for this proof to work.  It is possible however to ``fix it up'' so that it does, and thus to indeed conclude that we are very likely to be able to give a state-specific reconstruction on $BR$ of any sub-exponential unitary $W_{\ell rR}$ acting on any sub-exponential state $|\psi\ran$.  The details of this argument are given in appendix \ref{fixdcapp}.  

\subsection{Entanglement wedge reconstruction} \label{sec:EWR}
We've now seen that by acting on any sub-exponential state $|\psi\ran\in \Hall$ we can reconstruct any sub-exponential unitary $W_{\ell rR}$ with a unitary $W_{BR}$ on $\HB\otimes \HR$.  Under what circumstances can this reconstruction live just on $\HB$ or $\HR$?  We've seen that when $W$ has support \textit{only} on $R$ then it can always be reconstructed on $R$ alone (via the trivial reconstruction), so the natural remaining thing to consider is a unitary $W_{\ell r}$ which acts only on the interior modes.  The entanglement wedge reconstruction proposal \cite{Czech:2012bh,Wall:2012uf,Headrick:2014cta,Dong:2016eik,Penington:2019kki} says that if the entanglement wedge of $R$ includes the interior island shown in figure \ref{evapbhislandfig}, then we should be able to reconstruct $W_{\ell r}$ on $R$, while if it does not include the island then when we should be able to reconstruct $W_{\ell r}$ on $B$.  In fact this expectation follows from a general theorem proven in \cite{Akers:2021fut}, which shows that state-specific reconstruction is essentially equivalent to the validity of the QES formula for non-isometric codes.  In this subsection we confirm this directly by verifying that the decoupling bound \eqref{decoupling} holds where appropriate.  

We first consider the possibility that an interior unitary $W_{\ell r}$ can be reconstructed on $B$.  In the notation of the previous subsection, we would like to give an upper bound for $||\Psi_{LR}(U,W_{\ell r})-\Psi_{LR}(U,I)||_1$ and then apply theorem \ref{dcthm}.  The calculation is similar to that leading to \eqref{Ldcbound}, so we will be brief.  We again first observe that
\be
\int dU ||\Psi_{LR}(U,W_{\ell r})-\Psi_{LR}(U,I)||_1\leq \sqrt{|L||R|\int dU ||\Psi_{LR}(U,W_{\ell r})-\Psi_{LR}(U,I)||_2^2},
\ee
and evaluating the unitary integral we have
\begin{align}\nonumber
\int dU ||\Psi_{LR}(U,W_{\ell r})-\Psi_{LR}(U,I)||_2^2&=\frac{|P|^2|B|^2}{|P|^2|B|^2-1}\left(1-\frac{1}{|P|}\right)\frac{2}{|B|}\left(1-|\lan \psi W_{\ell r}|\psi\ran|^2\right)\\
&\leq \frac{4}{|B|}
\end{align}
and thus
\be\label{LRdcbound}
\int dU ||\Psi_{LR}(U,W_{\ell r})-\Psi_{LR}(U,I)||_1\leq 2\sqrt{\frac{|L||R|}{|B|}}.
\ee
Therefore when $|L||R|\ll |B|$ we can reconstruct $W_{\ell r}$ on $B$. 

We next consider the possibility that $W_{\ell r}$ can be reconstructed on $R$.  As before we have
\be
\int dU  ||\Psi_{LB}(U,W_{\ell r})-\Psi_{LB}(U,I)||_1\leq \sqrt{|L||B|\int dU ||\Psi_{LB}(U,W_{\ell r})-\Psi_{LB}(U,I)||_2^2},
\ee
and evaluating the unitary integral we find\footnote{We emphasize here that $\psi_R=\tr_{\ell r}|\psi\ran\lan \psi|$, $\psi_R$ is not the reduced state of $\Psi_{LBR}$ on $R$.}
\begin{align}\nonumber
\int dU ||\Psi_{LB}(U,W_{\ell r})-\Psi_{LB}(U,I)||_2^2=&2\frac{|P|^2|B|^2}{|P|^2|B|^2-1}\left(1-\frac{1}{|P||B|^2}\right)\\\nonumber
&\times\left(e^{-S_2\left(\psi_R\right)}-\tr\left(\psi_{L\ell r}W_{\ell r}^\dagger \psi_{L\ell r}W_{\ell r}\right)\right)\\
\leq&4 e^{-S_2\left(\psi_R\right)}
\end{align}
and thus
\be\label{LBdcbound}
\int dU  ||\Psi_{LB}(U,W_{\ell r})-\Psi_{LB}(U,I)||_1\leq 2 \sqrt{\frac{|L||B|}{e^{S_2\left(\psi_R\right)}}}.
\ee
Therefore we are likely to have a reconstruction of $W_{\ell r}$ onto $R$ provided that $S_2\left(\psi_R\right)\gg \log |L|+\log |B|$.

Let's now compare with what is expected from entanglement wedge reconstruction. We expect a reconstruction on $B$ if 
\be\label{Brec}
S(\psi_R)\ll S(\psi_L)+\log|B|
\ee
and a reconstruction on $R$ if 
\be\label{Rrec}
S(\psi_R)\gg S(\psi_L)+\log|B|.
\ee
The right-hand side of \eqref{LRdcbound} being small indeed implies \eqref{Brec}, and the right-hand side of \eqref{LBdcbound} being small indeed implies \eqref{Rrec} since $S_2(\psi) \leq S(\psi)$ for any $\psi$ by the concavity of $-\log  x$.  Thus our decoupling results are compatible with entanglement wedge reconstruction.  

There is an intermediate regime where neither the right-hand side of \eqref{LRdcbound} nor the right-hand side of \eqref{LBdcbound} is small.  This is unsurprising since the bounds we derived were somewhat crude; for example we could have replaced $|L||R|$ by the rank of $\Psi_{LR}$ and in \eqref{LBdcbound} we could replace $|L||B|$ by the rank of $\Psi_{LB}$. One might wonder whether in fact the bounds could be improved all the way to \eqref{Brec} and \eqref{Rrec}, but this is not the case. Except for specific classes of quantum states, \eqref{Brec} and \eqref{Rrec} are necessary, but not sufficient, conditions for entanglement wedge reconstruction to be possible. Instead, the optimal bounds involve tools from one-shot quantum Shannon theory, and feature an intermediate regime where neither reconstruction on $B$ nor $R$ is possible \cite{Akers:2020pmf}. 

We close by observing that we established the decoupling bounds \eqref{LRdcbound}, \eqref{LBdcbound} for particular sub-exponential unitaries $W_{\ell r}$ and states $|\psi\ran$, but a measure concentration argument which is analogous to the one given in appendix \ref{proofapp2} shows that in fact they likely hold for all sub-exponential unitaries and states.  

\subsection{Subspace-dependent reconstruction} \label{sec:subspacerecon}
We have now established various results about ``state-specific'' reconstruction, which allows the reconstruction of an effective field theory operator to be different for each state that it acts on.  Mathematically this is a fine thing to study, but if we want to interpret the reconstructed operators as observables in the fundamental description then the idea is in strong tension with the linearity of quantum mechanics.  Observables in quantum mechanics correspond to linear operators on the Hilbert space, and for a given measurement apparatus we don't get to change the operator depending on the state of the system.  In the following section we will confront this problem head-on, but here we first point out an alternative way of ameliorating the problem which has some relation to the proposal of  \cite{Papadodimas:2013jku} and also to the ``alpha-bits'' proposal of \cite{Hayden:2018khn}.  The idea is to show that if we happen to only be interested in states in some effective-description subspace $\sH_C\subset \Hall$, of dimension $|C|=|B|^{\beta}$ with $0<\beta<1$, then we can likely find a reconstruction of any operator on this subspace which works for all states in $\sH_C$ (we will need to assume that $|P|$ is at most sub-exponential in $|B|$, with the same justification as given for assuming this about $|L||R|$ at the end of section \ref{subexpsec}). In general we find this less appealing than our state-specific reconstruction on general sub-exponential states, since we see no reason to focus on any particular subspace, but it is nonetheless worth mentioning. 

Our argument borrows some techniques from appendix \ref{appxapp}.  Namely from \eqref{epsnet}, we can find an $\epsilon$-net $S$ for $\sH_C$ with 
\be
N_S\sim e^{2|B|^\beta \log \frac{1}{\epsilon}}.
\ee
From theorem \ref{setconcthm} we see that $V$ is likely to preserve the inner product of all elements of $S$ up to errors of order $|B|^{-\gamma}$ provided that 
\be\label{reconc}
\frac{|B|^{2\gamma}}{|B|^{1-\beta}} \log \frac{1}{\epsilon}\ll 1,
\ee
which will be the case at large $|B|$ if $\epsilon$ is not too small.  By the same argument as in \eqref{nettoall}, if we take
\be
\epsilon=\frac{1}{|P||B|^\gamma}
\ee
then $V$ will preserve the inner product of all elements of $\sH_C$ up to errors of order $|B|^{-\gamma}$. Assuming that $|P|$ is at most sub-exponential in $|B|$, then we have $\log \frac{1}{\epsilon}\ll |B|^\alpha$ for any $\alpha>0$ and thus  \eqref{reconc} will hold at large $|B|$ provided that we choose $\gamma<\frac{1-\beta}{2}$.  We therefore have
\be
||\hat{V}^\dagger \hat{V}-I_C||_\infty<\delta,
\ee
where $\hat{V}\equiv VP_C$, with $P_C$ the projection onto $\sH_C$, and $\delta$ can be taken to be of order $|B|^{-\gamma}$ with $0<\gamma<\frac{1-\beta}{2}$.

Since $\hat{V}$ is thus an approximate isometry from $\sH_C$ to $\HL\otimes\HB\otimes \HR$, it is natural to expect that we can use it to implement a canonical reconstruction of any observable $O$ on $\sH_C$.  Indeed defining
\be
\wt{O}\equiv \hat{V}O\hat{V}^\dagger,
\ee
for any state $|\psi\ran\in \sH_C$ we have
\begin{align}\nonumber
||(\wt{O}V-VO)|\psi\ran||&=||\hat{V}O(\hat{V}^\dagger \hat{V}-I_C)|\psi\ran||_\infty\\\nonumber
&\leq ||\hat{V}||_\infty||O||_\infty ||\hat{V}^\dagger \hat{V}-I||\\
&\leq ||O||_\infty \delta\sqrt{1+\delta}.
\end{align}
In the last step we have used that
\begin{align}\nonumber
||\hat{V}||_\infty^2&=||\hat{V}^\dagger \hat{V}||_\infty\\\nonumber
&=||\hat{V}^\dagger \hat{V}-I+I||_\infty\\
&\leq 1+\delta.
\end{align}
Thus $\wt{O}$ gives a good reconstruction of $O$ on all of $\sH$ which works to exponential precision in $\log |B|$.  We emphasize that due to the presence of $P_C$ in $\hat{V}$, this reconstruction will in general have support on all of $LBR$. Whether or not a reconstruction can be given with smaller support depends on more details of $O$ and $\sH_C$.

\section{Measurement theory for the black hole interior}\label{measurementsec}
We have seen that the holographic map $V\otimes I_{LR}$ defined by \eqref{Vdef} preserves the inner product between all sub-exponential states, reproduces the QES prescription, and allows for state-specific reconstruction of all sub-exponential observables in a way that is compatible with entanglement wedge reconstruction. On the other hand our construction is in some tension with the principles of quantum mechanics. State-specific reconstructions of effective description operators are non-linear on the fundamental Hilbert space (their definition depends on the state they act on), so they cannot be interpreted as observables according to the usual rules.  Moreover the non-isometric nature of our code ensures that there will be large numbers of states in the effective description that are annihilated by $V$, which seems to introduce a large ambiguity in how to assign effective-description interpretations to fundamental-description states.  How then are we to think about measurements in the black hole interior from the point of view of the fundamental description?  Are the outcomes of such measurements even well-defined?  In this section we will see that it is indeed possible to construct a self-consistent theory of interior measurements in the fundamental description, provided that the interior observer is restricted to measuring sub-exponential observables in sub-exponential states.  The basic idea is to introduce a measurement apparatus outside of the black hole and then reconstruct the (non-local) effective-description unitary which measures an interior observable and writes the answer onto this apparatus.

\subsection{Review of measurement in quantum mechanics} \label{subsec:measQM}
Let's first recall the standard measurement protocol in quantum mechanics.  To measure a hermitian observable $X$ on a quantum system $S$, we introduce an ``apparatus'' system $A$ whose dimensionality is the same as the number of distinct eigenvalues of $X$.  $A$ is initialized in some fixed state $|0\ran_A$, and then the measurement is implemented by a unitary $U^{\{X\}}$ which acts as
\be\label{measU}
U^{\{X\}}|i\ran_S|0\ran_A=|i\ran_S|x_i\ran_A.
\ee
Here $|i\ran_S$ are the eigenstates of $X$, with
\be
X|i\ran_S=x_i|i\ran_S.
\ee
Given any initial state $\rho_S$, the state of the apparatus after acting with $U^{\{X\}}$ is 
\be\label{appstate}
\rho_A=\sum_x p_x |x\ran\lan x|,
\ee
with the sum being over distinct eigenvalues of $X$.  Here 
\be
p_x=\tr(\rho_S P_x),
\ee
with $P_x$ being the projection onto the $x$-eigenspace of $X$.  The state \eqref{appstate} describes a classical probability distribution over measurement outcomes $x$, with the probabilities given by $p_x$.  The state of the full system after the measurement result $x$ becomes known is
\be
\rho^{\{x\}}_{SA}=\frac{1}{p_x}P_x\rho_S P_x\otimes |x\ran\lan x|_A. \label{px}
\ee
Tracing out the apparatus $A$ gives the usual Copenhagen rules, but when in doubt (as we soon will be) it should be included.

There is a generalization of the above protocol which will be useful in what follows. A measurement at its core involves an interaction between a system $S$ and an apparatus $A$ initialized to some fixed state $|0\ran_A$, where the result is subsequently read off from $A$.  So far we have considered interactions of the form \ref{measU}, which entangle a standard basis of $A$ with the eigenbasis of some hermitian operator $X$, and in terms of which we can write the measurement probabilities and post-measurement state as 
\begin{align}\nonumber
p_x&=\tr(\rho_S M_x M_x^\dagger)\\
\rho_{SA}^{\{x\}}&=\frac{1}{p_x}(M_x\rho_S M_x^\dagger)\otimes |x\ran\lan x|_A.\label{Mmeas}
\end{align}
with
\be
M_x\equiv \lan x|_A U^{\{X\}}|0\ran_A.
\ee
A generalized measurement simply allows the interaction between the system and the apparatus to consist of an arbitrary unitary $U_{meas}$, not necessarily obtained from any hermitian operator $X$ as in \eqref{measU}. We then define
\be
M_x\equiv \lan x|_A U_{meas}|0\ran_A,
\ee
with $x$ now just labelling some basis for $A$, and we determine the measurement probabilities and the post-measurement state from \eqref{Mmeas} as before.  The measurement probabilities add up to one since
\be
\sum_x M_x^\dagger M_x=I.
\ee
From this point of view, measurements associated to hermitian observables as in \eqref{measU} are referred to as ``projective measurements''.  Projective measurements are characterized by the condition that the $M_x$ form a complete set of mutually orthogonal projectors, i.e. they are all hermitian and obey
\be
M_x M_{x'}=\delta_{xx'}M_x.
\ee

In the beginning of section \ref{pagesec} we gave a preliminary definition of what is meant by a ``sub-exponential observable'': we said that an observable is sub-exponential if it can be written as a linear combination of two sub-exponential unitaries.  A better definition, which more clearly characterizes the hardness of performing the measurement, is that the measurement unitary $U^{\{X\}}$ (or more generally $U_{meas}$) can be implemented with a sub-exponential quantum circuit.  In appendix \ref{complexityapp} we show in lemma \ref{measeqlem} that this definition implies the provisional one, so from now on we define a sub-exponential observable using the complexity of the measurement unitary.  

So far we have been somewhat vague about when exactly a definite outcome for the measurement is realized, i.e. when the wave function collapses and the state becomes \eqref{px}.  Most quantum mechanics textbooks are deliberately unclear on this point, since in most situations it does not matter in practice.  We however are considering a situation where an external observer can have access to the complete set of fundamental degrees of freedom, and so we need to be more careful.  The basic problem is usually illustrated using the anecdote of ``Wigner's friend'' \cite{wigner1995remarks}, which goes as follows.  We have a quantum mechanical system $S$ on which we would like to measure some observable $X$, an apparatus $A$ which will be acted on by the measurement unitary $U^{\{X\}}$, a friend of Wigner, whom we will refer to as Dirac and denote $D$, and Wigner himself, whom we will denote $W$.  Prior to the measurement we begin in the state
\be\label{state1}
\sum_i C_i |i\ran_S|0\ran_A|0\ran_D|0\ran_W,
\ee
where both Dirac and Wigner have been initialized in a pure state which has no information about either $S$ or $A$ (it is not necessary for the state to be pure but it simplifies things).  Here $|i\ran_S$ are the eigenstates of the observable $X$.  We now act with the measurement unitary $U^{\{X\}}$, producing the state
\be\label{state2}
\sum_i C_i |i\ran_S|x_i\ran_A|0\ran_D|0\ran_W.
\ee
Should we say that a definite measurement outcome has happened?  In the Copenhagen approach one would say yes, but the state \eqref{state2} remains a coherent superposition in which $SA$ is uncorrelated with $DW$, and, depending on the nature of $S$ and $A$, Wigner and Dirac (or more likely an experimental colleague) might even be able to undo the measurement unitary and return to the state \eqref{state1}, where it seems quite clear that no definite outcome has happened.  As a matter of principle we therefore had better say that in the state \eqref{state2} no definite measurement outcome has happened.  

Let's now consider what happens once Dirac looks at the apparatus to see what was measured.  The state of the system becomes
\be\label{state3}
\sum_i C_i|i\ran_S|x_i\ran_A|x_i\ran_D|0\ran_W.
\ee
From Dirac's point of view it is now clear that a definite measurement outcome has been realized.  In each branch of the wave function Dirac is certain which measurement outcome was obtained, and if he applies the measurement again he is sure to obtain the same result.  On the other hand from Wigner's point of view, Dirac can just be viewed as part of the apparatus and so by the same argument as in the previous step it is best to say that no definite measurement outcome has happened.  Who is right? The most sensible response is ``they both are'', from which we learn that the answer to the question of whether or not a specific measurement outcome happened is subjective.  Wigner is only allowed to assume the wave function has collapsed in the state \eqref{state3} to the extent that he is unable to act coherently on Dirac. Of course once Wigner either looks at the apparatus himself or asks Dirac what he saw, then the state becomes
\be\label{state4}
\sum_i C_i|i\ran_S|x_i\ran_A|x_i\ran_D|x_i\ran_W,
\ee
after which to measure the phase coherence Wigner would have to somehow perform a unitary operation which erased his own memories. It is difficult to see how one could do science in the context of such operations.  Thus the safest rule is the following: from the point of view of any particular observer, the wave function does not collapse until that observer learns the measurement outcome.

\subsection{Exterior measurements in the fundamental description}
\bfig
\includegraphics[height=7cm]{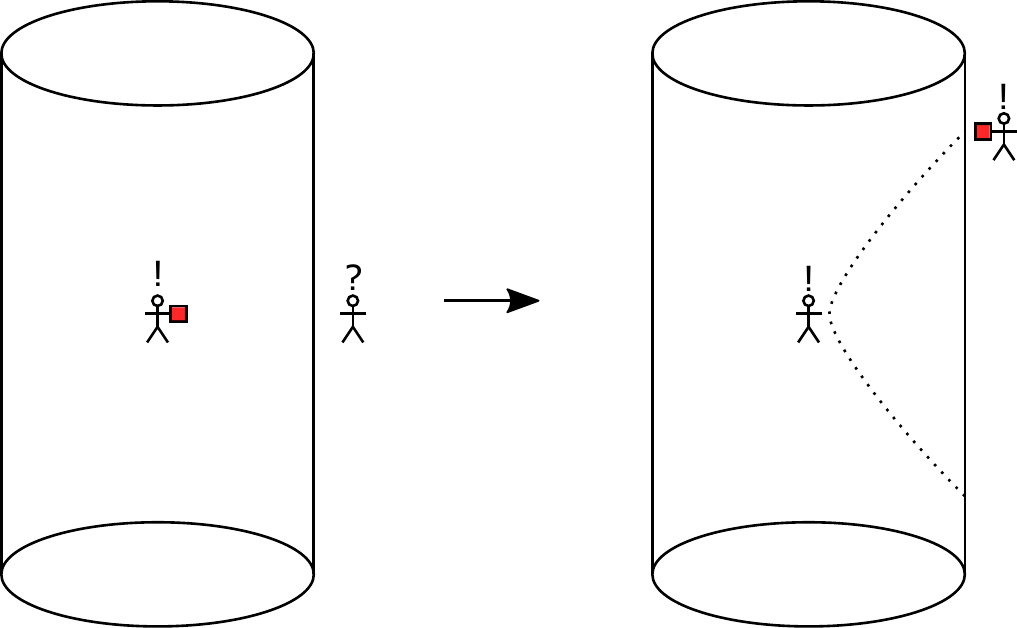}
\caption{Converting a bulk measurement to a boundary measurement in AdS/CFT.  On the left a bulk observer Diana is holding an apparatus which she has just looked at to learn the result of a measurement, but a boundary observer Wendy thinks that Diana and the apparatus are together in a coherent superposition and no definite outcome has been realized.  On the right Wendy uses local operators in the fundamental description to sends a ``space mission'' into the bulk which retrieves Diana's apparatus and swaps it out of the system, after which Wendy can look at it to learn the measurement outcome.}\label{measurementfig}
\efig
As a warm-up we first consider a bulk measurement in AdS/CFT which is not behind a black hole horizon (see \cite{Heemskerk:2012mn} for an earlier discussion of this topic).  The situation in the effective description is shown in the left diagram of figure \ref{measurementfig}.  A bulk observer, let's call her Diana, has an apparatus $a$ which she uses to measure an observable $X$ on some system $s$.  As discussed in the previous section, this results in an effective-description state
\be
\sum_i C_i |i\ran_s|x_i\ran_a|x_i\ran_d.
\ee
Since we are excluding measurements behind horizons there is an (approximately) isometric encoding map $V$, and passing this state through the encoding we obtain a state in the fundamental description.  According to an external observer in the fundamental description, whom we'll call Wendy, the full state of system is
\be\label{dWstate}
\sum_i C_iV\Big(|i\ran_s|x_i\ran_a|x_i\ran_d\Big)|0\ran_W.  
\ee
This state is analogous to state \ref{state3} in the previous subsection, and as in that case we here say that according to Diana a definite measurement outcome has been obtained but according to Wendy the system is still in superposition.  Indeed since we are assuming that Wendy has complete control over the fundamental degrees of freedom, it is in her power to evolve the state \eqref{dWstate} back to the pre-measurement state
\be
\sum_i C_iV\Big(|i\ran_s|0\ran_a|0\ran_d\Big)|0\ran_W,
\ee
where it is clear that no outcome has happened. On the other hand since $V$ is an approximate isometry, given any observable $O$ in the effective description we have the canonical reconstruction $\wt{O}=VOV^\dagger$ in the fundamental description.  Moreover by \eqref{stateaction} and \eqref{codematrix}, a measurement of $\wt{O}$ by Wendy gives the same probabilities and projects onto the same states as a measurement of $O$ by Diana up to exponentially small corrections.  It is therefore common to conflate these two measurements in discussion of reconstruction, as is done in \cite{Heemskerk:2012mn}. Indeed there is a simple way to convert a measurement by Diana into a measurement by Wendy \cite{Bousso:2012mh}: Wendy can create some sort of ``rocket'' near the boundary using local operators, which she programs to fly into the bulk, grab the apparatus $a$, and take it back to the boundary.  She then uses local operators to swap it out into a boundary apparatus $A$, which she can then look at and thus realize a definite measurement outcome (see the right diagram of figure \ref{measurementfig} for an illustration).   The state after this operation is
\be
\sum_i C_i V\Big(|i\ran_s|x_i\ran_d\Big)|x_i\ran_A|x_i\ran_W,
\ee
so, just as in state \eqref{state4} in the previous subsection, Wendy is in full agreement with Diana on the measurement outcome and for both the phase information in the superposition is lost beyond recall.

\subsection{Interior measurements in the fundamental description}
We now consider measurements behind a black hole horizon.  There are two new issues which arise:\footnote{Versions of these issues can arise already in the absence of black holes if we only allow the external observer to have access to a subset of the fundamental degrees of freedom.  For example in AdS/CFT, they can appear whenever the entanglement wedge differs from the causal wedge \cite{Hayden:2018khn,Akers:2019wxj,Akers:2020pmf}.  In those situations however they can be thought of as arising from imposing an artificial restriction on the external observer, and thus do not call into question whether or not the fundamental description has access to the information at all.}
\bi
\item[(1)] Acting on the full set of sub-exponential states behind the horizon, the holographic map $V$ is no longer an isometry and so by theorem \ref{nogo5} we cannot in a state-independent way reconstruct effective-description sub-exponential observables in the fundamental description.
\item[(2)] There is no rocket which can travel to the interior and return, so there is no causal procedure in the effective description which can convert an interior measurement to a measurement in the fundamental description. 
\ei
These two observations are compatible with each other, and indeed (2) can be interpreted as a consequence of (1). We therefore have to be more limited in what we hope to accomplish.  More concretely we need to allow some non-linearity in the fundamental description of measurement, and if we want to convert interior measurements to boundary measurements then we also need to allow for some violation of causality in the effective description.  We now develop a theory of measurement along these lines. 
 
It is useful to first emphasize a particular consequence of the preservation of overlaps for sub-exponential states:
\begin{lemma}\label{invlem}
Let $V$ be defined as in \eqref{Vdef}, and let $|\psi\ran$ and $|\phi\ran$ be sub-exponential states in the sense of section \ref{subexpsec}.  Then for any $0<\gamma<1/2$, with high probability we have
\be\label{invineq}
\Big|\big|\big|V\otimes I_{LR}(|\psi\ran-|\phi\ran)\big|\big|-\big|\big||\psi\ran-|\phi\ran\big|\big|\Big|\leq 2(18)^{1/4}|B|^{-\gamma/2}
\ee
\end{lemma}
\begin{proof}
We first observe that from the triangle inequality and \eqref{subexpconc}, with high probability we have
\be
\Big|\big|\big|V\otimes I_{LR}(|\psi\ran-|\phi\ran)\big|\big|^2-\big|\big||\psi\ran-|\phi\ran\big|\big|^2\Big|\leq 4\sqrt{18}|B|^{-\gamma}.
\ee
\eqref{invineq} then follows from the fact that for any $a,b>0$ we have $|a-b|\leq \sqrt{|a^2-b^2|}$ (without loss of generality we can assume $a\geq b$, in which case we have $a-b\leq a+b$ and thus $(a-b)^2\leq (a-b)(a+b)$).
\end{proof}
What this lemma says is that the encoded images of two sub-exponential states are close if and only if the sub-exponential states themselves are close. In other words although $V\otimes I_{LR}$ is highly non-isometric, it is approximately invertible on the set of sub-exponential states.  Thus any state in the fundamental description which is the image of a sub-exponential state has a unique interpretation in the effective description.  This restricted inverse however cannot be interpreted as a linear map, since neither the set of sub-exponential states nor their encoded images form a linear space.\footnote{It is true however that any superposition with sub-exponential coefficients of a sub-exponential number of sub-exponential states is itself sub-exponential, see lemma \ref{suplemapp}.}  

Let's first consider what an interior measurement looks like in the effective description.   We implement this by splitting the infalling modes $\ell$ into two parts - an apparatus $a$ and its complement $\hat{\ell}$ (an interior observer $d$ who looks at the apparatus can just be included as part of $a$).  The measurement of a sub-exponential observable $X$ which is supported on $\mathcal{H}_{\hat{\ell}} \otimes \Hr\otimes \HR$ is implemented by a unitary operator $U^{\{X\}}$ which acts as
\be
U^{\{X\}}|i\ran_{L\hat{\ell}rR}|0\ran_a=|i\ran_{L\hat{\ell}rR}|x_i\ran_a.  
\ee

Now say that some exterior observer Wendy is given a state $|\wt{\psi}\ran_{LBR}$ in the fundamental description such that\footnote{To avoid clutter, for the remainder of this section we will use the symbol $\approx$ to mean ``equal up to terms that are exponentially small in $\log|B|$.''} 
\be
|\wt{\psi}\ran_{LBR}\approx(V\otimes I_{LR})|\psi\ran_{L\hat{\ell}rR}|0\ran_a
\ee
for some sub-exponential state $|\psi\ran_{L\hat{\ell}rR}$.  Lemma \ref{invlem} ensures that there will only be one such state, so the effective description interpretation of $|\wt{\psi}\ran_{LBR}$ is unambiguous.  Since $U^{\{X\}}$ is sub-exponential, by theorem \ref{dcthm} and equation \eqref{Ldcbound} we see that there exists a state-specific reconstruction $U_{BR}^{\{X\}}$ such that
\be
U_{BR}^{\{X\}}|\wt{\psi}\ran\approx (V\otimes I_{LR}) U^{\{X\}}|\psi\ran_{L\hat{\ell}rR}|0\ran_a.
\ee  
Therefore Wendy is able (using her quantum computer) to act with a unitary in the fundamental description which  implements the measurement of $X$ in the effective description up to exponentially small ambiguities.  On the other hand this procedure is not a measurement from Wendy's point of view: there is no apparatus in the fundamental description, and $BR$ remains in a coherent superposition so no definite measurement outcome has been realized.  Wendy can if she likes compute the measurement probabilities and post-measurement states which an effective-description observer would use to describe this measurement (she just has to use the invertibility promised by lemma \ref{invlem} and then compute them in the effective description), but she assigns them no direct meaning in terms of her own experiences.  

In the previous subsection we saw that Wendy could convert a  measurement in the effective description to a measurement in her own description, for example using the rocket algorithm shown in figure \ref{measurementfig}, but any such algorithm for a measurement behind the horizon would necessarily violate the causal structure of the effective description.  This does not mean it is impossible: Wendy's operations are not constrained by the validity of the effective description.  For example in AdS/CFT it is well within her power to act with an operator which instantly creates a particle deep in the bulk.  This kind of operation however usually cannot be given a coherent interpretation by an observer who is living in the bulk: such operations typically affect their memories or place them in configurations which do not arise from any sensible history. Nonetheless they can still be useful to consider, and we now explain how they can give a physical interpretation in the fundamental description to the outcomes of interior measurements.  

The basic idea is to change our interpretation of the apparatus in the effective description: instead of viewing it as part of $\ell$, we instead view it as an addition to $R$.  In other words the apparatus stays outside of the black hole.  We will now refer to it as $A$ to reflect this change.  The measurement unitary in the effective description acts as
\be
U^{\{X\}}|i\ran_{L\ell r R}|0\ran_A=|i\ran_{L\ell rR}|x_i\ran_A.
\ee
Physically this seems absurd, as it allows a measurement inside a black hole to produce a record which stays outside, but we are not claiming that any effective-description observer could actually implement this unitary.  Now let's say that in the fundamental description we have a state $|\wt{\psi}\ran_{LBR}$ such that
\be\label{psiencode}
|\wt{\psi}\ran_{LBR}\approx(V\otimes I_{LR})|\psi\ran
\ee
for some sub-exponential state $|\psi\ran\in \Hall$.  Lemma \ref{invlem} again ensures that $|\psi\ran$ is unique. Since $U^{\{X\}}$ is sub-exponential, by theorem \ref{dcthm} and equation \eqref{Ldcbound} (now viewing $A$ as an addendum to $R$) we see that there exists a state-specific reconstruction $U_{BRA}^{\{X\}}$ such that
\begin{align}\nonumber
U_{BRA}^{\{X\}}|\wt{\psi}\ran_{LBR}|0\ran_A&\approx (V\otimes I_{LRA})U^{\{X\}}|\psi\ran|0\ran_A\\\label{formalmeas}
&=\sum_i \lan i|\psi\ran(V\otimes I_{LRA})|i\ran|x_i\ran_A.
\end{align}
The advantage of this approach, despite its less clear physical interpretation, is that \eqref{formalmeas} now clearly resembles a projective measurement in the fundamental description.  On the other hand since \eqref{formalmeas} is only approximate, it isn't exactly a projective measurement.  We can however interpret it as a generalized measurement: defining
\be
M_x\equiv \lan x|_A U_{BRA}^{\{X\}}|0\ran_A,
\ee
we have the measurement probabilities
\be
p_x=\tr(|\wt{\psi\ran}\lan \wt{\psi}|M_xM_x^\dagger)
\ee
and the post-measurement state
\be
\rho_{LBRA}^{\{x\}}=\frac{1}{p_x}(M_x|\wt{\psi}\ran\lan \wt{\psi}|M_x^\dagger).
\ee
We then have
\begin{align}\nonumber
p_x &\approx \tr(P_x|\psi\ran\lan \psi|)\\
\rho_{LBRA}^{\{x\}}&\approx \frac{1}{p_x}(V\otimes I_{LRA})P_x |\psi\ran\lan \psi|P_x\otimes |x\ran\lan x|_A(V^\dagger\otimes I_{LRA}),
\end{align}
where $P_x$ is the projection onto the $x$-eigenspace of $X$ in $\ell rR$. Therefore this generalized measurement has outcome probabilities which agree to exponential precision with those obtained using the standard rules in the effective description, and it leads to a post-measurement state which is exponentially close to the encoded version of the standard post-measurement state in the effective description.  We show in appendix \ref{complexityapp} (lemma \ref{projlem}) that in the effective description a post-measurement state is sub-exponential if the pre-measurement state and the observable both are, and so we are free to use it as the starting point for additional measurements.  In this way we are able to account for the full measurement theory of a (sub-exponentially-bounded) interior observer doing measurements in the fundamental description only.  The rules differ from those of ordinary quantum mechanics, in particular because we are restricted to sub-exponential states and observables and we have to already know the initial state in order to determine which generalized measurement to use, but the results agree to exponential accuracy with ordinary quantum mechanics as used by an interior observer in the effective description.

\section{A dynamical model}\label{dynamicsec}

So far we have considered features of the holographic map at a fixed time, with dynamics only implicitly present through the dimensions of the various subsystems and the choice of state.  More properly, both the effective description and the fundamental description should evolve under their own dynamics, and these dynamics should be compatible with the holographic map in the sense that it does not matter whether we evolve and then encode or encode and then evolve.  Mathematically this is the statement that the holographic map should be equivariant with respect to time evolution.  In this section we present a dynamical version of our model, which has unitary time evolution in both the effective and fundamental descriptions and for which the (non-isometric) holographic map is indeed equivariant.   Our model is designed to incorporate the basic dynamics which goes into Hawking's information problem: matter in some adjustable quantum state collapses to form a black hole, which then gradually evaporates into radiation via unitary dynamics.  The goal is to understand how these statements are compatible with the existence of the interior in the effective description.

\subsection{Fundamental dynamics}\label{fundyn}

\begin{figure}[t]
    \centering
    \includegraphics[width=6cm]{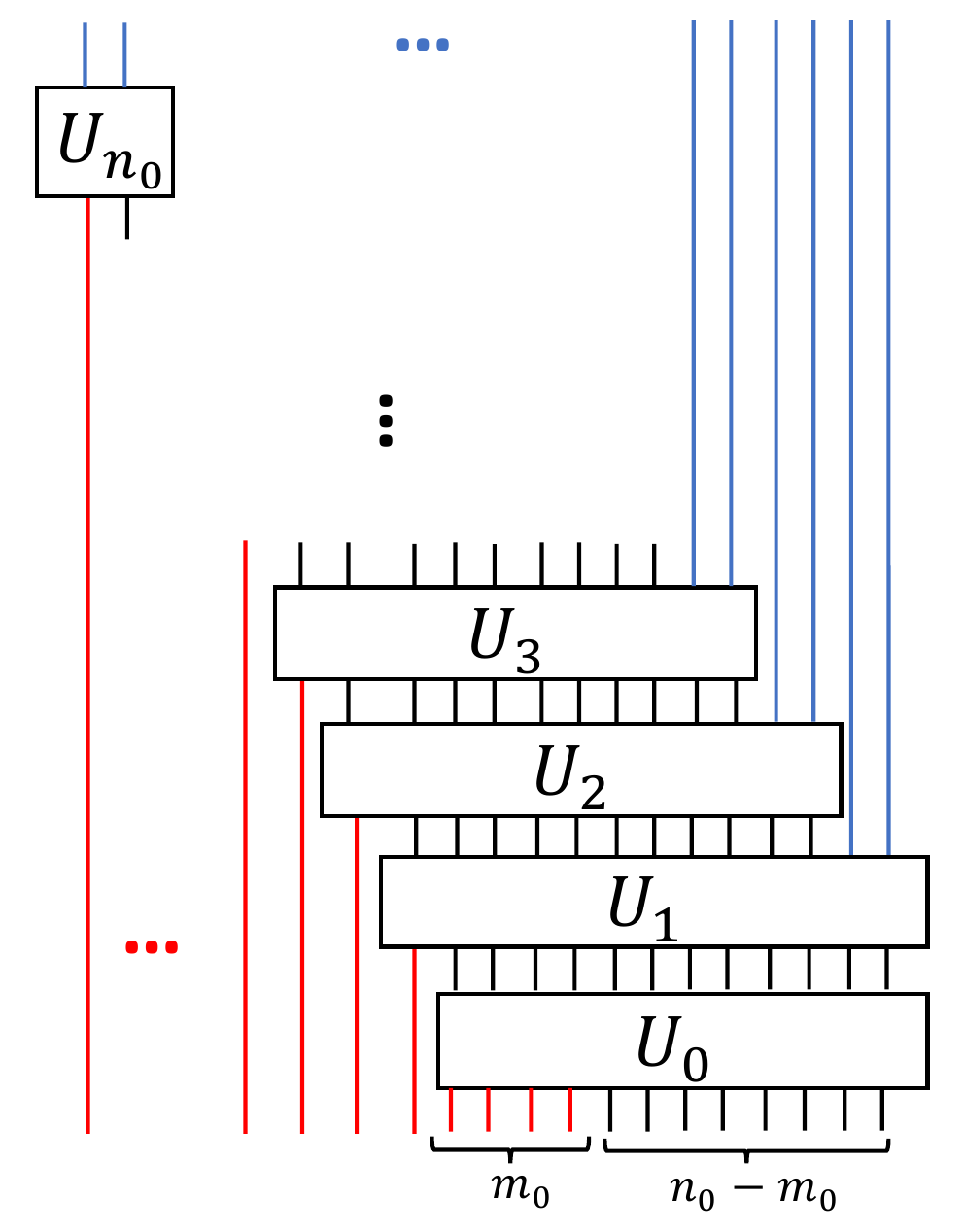}
    \caption{The full time-evolution of an evaporating black hole in the fundamental picture. The red lines represent  qudits that fall into the black hole from the bath, while the blue lines represent qudits that come out of the black hole as radiation. After $n_0$ time steps only radiation remains.  Note that the outgoing qudits appear in the order they would on the nice slice in figure \ref{evapbhmapfig}, while the ingoing qudits are reversed (qudits to the right fell in earlier).}
    \label{fig:bh_dyn}
\end{figure}

We first describe the dynamics in the fundamental picture. At time $t=-1$ we have $m_0$ qudits of collapsing matter, which are converted by an isometry $V_{0}$ to a black hole at time $t=0$ consisting of $n_0$ qudits (for a ``fast collapse'' we should have $m_0\ll n_0$).  We can always write $V_0=U_0|\psi_0\ran_f$ , where $U_0$ is unitary and $|\psi_0\ran_f$ is some fixed state of $n_0-m_0$ qudits.  The black hole is surrounded by a bath consisting of a countably-infinite number of ``ingoing'' qudits $R_{in}$ and countably-infinite number of ``outgoing'' qudits $R_{out}$.  We assume that each qudit has Hilbert space dimension $q$. Starting at $t=0$ we then evolve under the following dynamics: in each time-step 
\begin{enumerate} 
\item One ingoing qudit falls into the black hole. 
\item The black hole qudits interact according to some unitary evolution $U_{t+1}$.
\item Two outgoing qudits come out of the black hole. 
\end{enumerate}
Hence at time $t\geq 0$, the number of ingoing qudits that have fallen into the black hole is $t+m_0$, the number of outgoing qudits that have emerged is $2t$, and the remaining number of qudits in the black hole is $n_0 - t$. The black hole has evaporated completely after $n_0$ steps. See figure \ref{fig:bh_dyn} for an illustration.  The full Hilbert space of the fundamental theory is therefore 
\be
\sH_{\mathrm{fund}}=\sH_{R_{in}}\otimes \sH_{R_{out}} \otimes \sH_{BH},
\ee
where the black hole Hilbert space $\sH_{BH}$ is a direct sum over a different sizes of black hole:
\be \label{eq:funddecomp}
\sH_{BH} \cong \oplus_t \sH_{{B}^{(t)}}~,
\ee
with 
\be\label{Bsize}
\log_q |B^{(t)}| = n_0 - t.
\ee
We will refer to the subset of $R_{out}$ which was produced by the black hole after $t$ time steps as $R^{(t)}$, in which case we have 
\be
\log_q |R^{(t)}| = 2t.
\ee

The fact that two qudits escape the black hole as radiation for each qudit that falls into the black hole is a choice of the model. In reality, an arbitrary amount of information can be thrown into the black hole, but the generalized second law gives a minimum energy cost to doing so. In particular, the black hole can only lose energy and evaporate if fewer qubits fall in than escape in the radiation. The ratio of ingoing qubits to radiation qubits in our model was chosen to be as simple as possible subject to this constraint.

We consider two possible forms for the unitary evolution $U_t$. In the simpler model, which we'll call the block random unitary (BRU) model, each $U_t$ is chosen randomly from the Haar ensemble of unitaries acting on the appropriate number of qudits.  This model has the advantage that computations are relatively simple, but it has also has a substantial disadvantage: it is ``too scrambling'' in the sense that it scrambles perturbations much faster than we'd expect for a dynamics generated by a $k$-local Hamiltonian.  One way to fix this would be to at each time take $\log q\sim \log(n_0-t)$, which coarse-grains the system sufficiently that each time step now lasts for a scrambling time and thus instantaneous scrambling isn't so bad, but this decreases the resolution with which we can model the effective description dynamics.  

We therefore introduce a more realistic version of the model, which we'll call the random pairwise interaction (RPI)  model, as follows:  at each time step we randomly pair up all of the qudits, and then apply an independent Haar-random interaction to each pair.\footnote{One qudit will be left over if $(n_0 - t + 1)$ is odd, in which case we leave it untouched.} More formally, we can take $U_t$ to consist of a product of Haar-random interactions between neighboring pairs of qudits conjugated by a random permutation of the qudits.  The RPI model has the nice features that a) it only involves two-local interactions at each timestep and b) it is `fast scrambling'' in the sense that quantum information which falls into the black hole after the page time is re-emitted after an $O(\log n)$ time \cite{Sekino:2008he,Lashkari:2011yi,brown2012scrambling}.  This latter property is known as the Hayden-Preskill decoding criterion, and is strongly believed to hold for black holes \cite{Hayden:2007cs,Shenker:2013pqa}.  Most of the remainder of this section does not depend on which of these models we use, but it does affect the calculations in appendix \ref{dyn_app}.

\subsection{Effective dynamics}
We now describe the dynamics of the effective picture.  The full effective description Hilbert space consists of the bath qudits together with a direct sum over interior left- and right- moving sectors associated to different times, or equivalently to different lengths of the black hole interior: 
\be \label{eq:effdecomp}
\sH_\mathrm{eff} \cong \sH_{R_{in}}\otimes \sH_{R_{out}}\otimes\left(\oplus_t \sH_{{\ell}^{(t)}} \otimes \sH_{r^{(t)}}\right)~.
\ee
At each time-step we then have the following evolution:
\bi
\item[1.] An ingoing qudit falls in, joining $\ell$.
\item[2.] Two ``Hawking'' pairs of qudits are added to the system, each in the maximally-entangled state $\frac{1}{\sqrt{q}}\sum_{i=0}^{q-1}|i\ran|i\ran$ and each with one qudit joining $r^{(t)}$ and the other joining $R_{out}$.
\ei
For the initial state we take $\ell^{(0)}$ to be $m_0$ matter qudits in whatever state which went into the collapse isometry $V_0$ in the fundamental description, and we take $r^{(0)}$ to be empty.   As a function of time we then have
\begin{align}\nonumber
\log_q|\ell^{(t)}|&=m_0+t\\
\log_q |r^{(t)}|&=2t.\label{lrsize}
\end{align}
These dynamics are clearly unitary in the sense of preserving inner products as we advance in time, but the process of adding qudits to the system in step 2 is somewhat unusual: it requires a ready supply of qudits in a known state which can be brought in as needed.  If we modify the state in some way at time $t$, say by making a measurement or perturbing it from the outside,  we can continue to evolve it forward using steps 1 and 2, but we may not be able to evolve it backwards since the resulting state would not be a product with the two qudits we'd want to remove.  This however is a feature rather than a bug: it models the fact that in effective field theory in curved spacetime we can have a situation where modes are ``coming down'' from above the cutoff, and this necessarily requires additional degrees of freedom to be introduced.  If we modify the state and try to evolve it backwards, then we likely will meet some sort of conventional breakdown of effective field theory due to large energy densities and/or curvatures \cite{tHooft:1984kcu,Kiem:1995iy}.

\subsection{Defining the holographic map}
\bfig
\includegraphics[height=8cm]{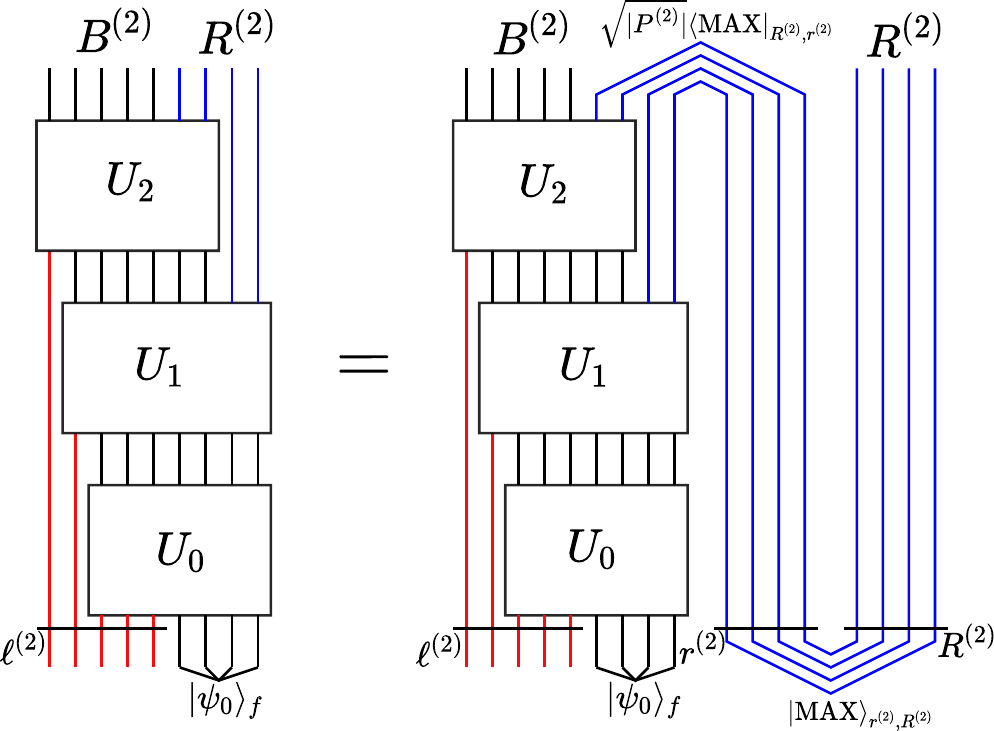}
\caption{Re-interpreting the evolution in the fundamental description as a non-isometric map from the effective description.  Here $n_0=7$, $m_0=3$, and $t=2$.  The factor of $\sqrt{|P^{(2)}|}=|r^{(2)}|$ is necessary to ensure that the maximally-mixed states contract to give the identity map on $R^{(2)}$.  The diagram on the right is precisely the action of a non-isometric holographic map $V_2\otimes I_{R^{(2)}}$ on an effective description state which is arbitrary on $\ell^{(2)}$ and maximally mixed on $r^{(2)}R^{(2)}$, with $V_2:\sH_{\ell^{(2)}}\otimes \sH_{r^{(2)}}\to\sH_{B^{(2)}}$ given by $V_2=\sqrt{|P^{(2)}|}\lan \mathrm{MAX}|_{R^{(2)},r^{(2)}}U_2U_1U_0|\psi_0\ran_f$.  Note the clear resemblance to figure \ref{Vdeffig}.}\label{dynevapfig}
\efig
Let's now understand the relation between the fundamental and effective descriptions.  The basic idea is that we can rewrite the fundamental evolution as a non-isometric map acting on the the effective description by ``bending around'' the outgoing radiation using post-selection.\footnote{This operation is strongly reminiscent of the ``black hole final state'' proposal of \cite{Horowitz:2003he}, but it isn't the same.  For one thing there is no holography in the final state proposal, and for another one does the post-selection there ``at the end of the dynamics'' while here it appears in the holographic map at a fixed time.  We discuss the relationship between the two approaches in more detail in section \ref{discussion}.}  The is illustrated in figure \ref{dynevapfig}.  In general we define the holographic map $V_t:\sH_{\ell^{(t)}}\otimes \sH_{r^{(t)}}\to \sH_{B^{(t)}}$ recursively by starting with
\be
V_0=U_0|\psi_0\ran_f
\ee
and then at each step acting with $U_t$ and then post-selecting on the maximally-mixed state for the two new radiation qudits: see figure \ref{Vrecfig} for an illustration.
\bfig
\includegraphics[height=6cm]{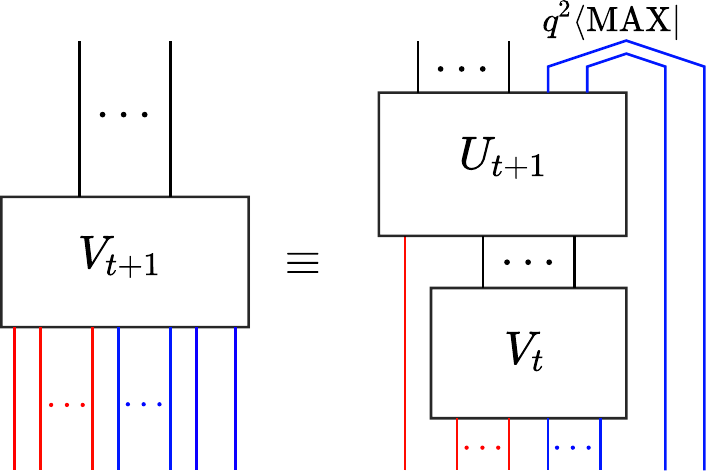}
\caption{The recursive step in defining the holographic map $V_t:\sH_{\ell^{(t)}}\otimes \sH_{r^{(t)}}\to \sH_{B^{(t)}}$. To obtain $V_{t+1}$ we act on $V_t$ with $U_{t+1}$ and then post-select the two new outgoing qudits in $R^{(t+1)}$ down to become the new qudits in $r^{(t+1)}$ in the effective description.}\label{Vrecfig}
\efig
We thus have
\be\label{Vdef2}
V_t=\sqrt{|P^{(t)}|}\lan \mathrm{MAX}|_{R^{(t)},r^{(t)}}U_tU_{t-1}\ldots U_1U_0|\psi_0\ran_f,
\ee
which we observe is a special case of \eqref{Vdef} with
\be
\sH_{P^{(t)}}=\sH_{R^{(t)}}\otimes \sH_{r^{(t)}}.
\ee
We define the full encoding map $V:\oplus_t \sH_{{\ell}^{(t)}} \otimes \sH_{r^{(t)}}\to \sH_{BH}$ by 
\be\label{eq:directsumV}
V\equiv \oplus_t V_t.
\ee

\bfig
\includegraphics[height=6cm]{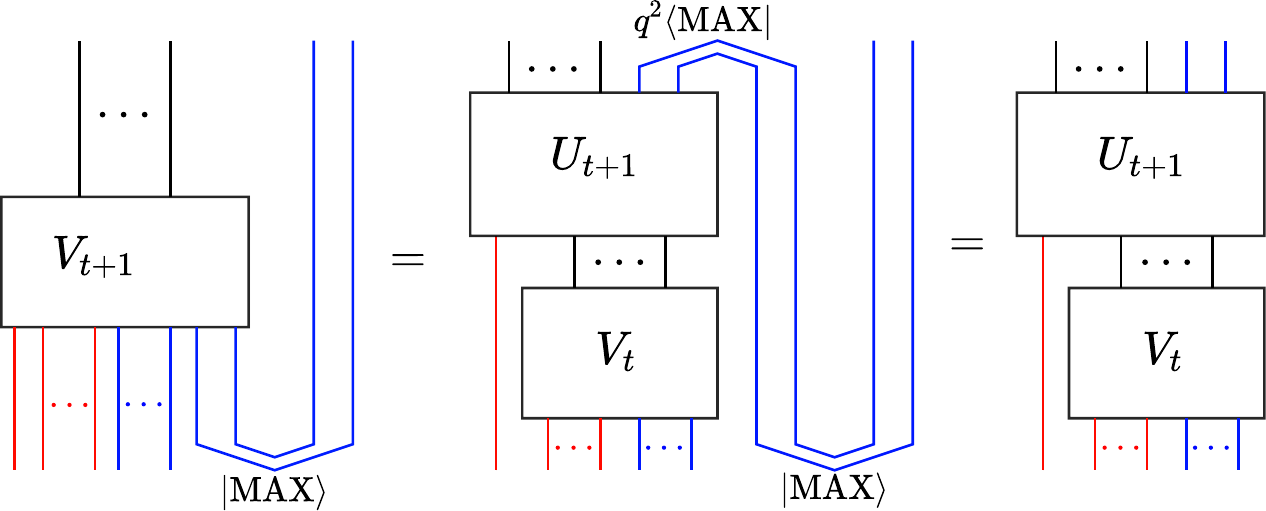}
\caption{Equivariance of the holographic map under time evolution: acting with the effective-description dynamics at time $t$ and then encoding with $V_{t+1}$ is the same as encoding with $V_t$ at time $t$ and then acting with the fundamental-description dynamics.}\label{Veqfig}
\efig
The equivariance of $V\otimes I_{R_{in}R_{out}}$ with respect to time evolution is an immediate consequence of this recursive construction of $V_t$, see figure \ref{Veqfig} for the proof.  We emphasize that this proof works for general inputs into $V_t$; we do not need to restrict to maximal entanglement between $r^{(t)}$ and $R^{(t)}$ even though we used it in figure \ref{dynevapfig} to motivate the definition of $V_t$.  For example at any particular time we can do a sub-exponential measurement or act from the outside with a sub-exponential unitary and then continue to evolve and equivariance will be maintained.  

\subsection{Entropy and the Page curve}
Having introduced our dynamical model, we now describe the results of some computations in the model.  The method of computation is the introduction of an equivalent spin model, as proposed in \cite{Hayden:2016cfa}.  The details are somewhat involved, so we relegate the details to appendix \ref{dyn_app} and just present the results here.  

We first discuss to what extent $V$ is an isometry.  This reduces to the question of when $V_t$ is an isometry.\footnote{The static model introduced in section \ref{modelsec} only describes effective theory states with a single wormhole length (i.e. at a single time), so its holographic map $V$ is analogous to $V_t$ rather than the full direct sum $V$ defined in \eqref{eq:directsumV}.}  $V_t$ certainly cannot be an isometry when $|B^{(t)}|<|\ell^{(t)}||r^{(t)}|$, which from \eqref{Bsize} and \eqref{lrsize} is equivalent to $t>\frac{n_0-m_0}{4}$.  We saw in the static model (i.e. equation \eqref{isomcond}) that this bound is approximately saturated, in the sense that the holographic map was approximately isometric as soon as $|\ell||r|\ll |B|$.  In appendix \ref{dyn_app} we show that the same is true in the dynamical model: $V_t$ is approximately isometric for $t\ll\frac{n_0-m_0}{4}$. This timescale can be interpreted as the ``Page time'' of the model (we'd have gotten $n_0/2$ if we hadn't allowed for ingoing modes).

We can also study the Page curve of the model.  From the discussion in the previous subsections, we expect the QES formula to continue to hold at times when $V_t$ is no longer an isometry.  If we begin the system (in the fundamental description) in a state
\be
|\psi\ran_{L\ell^{(t)} f}=|\chin\ran_{L\ell^{(t)}}|\psi_0\ran_f
\ee
of the pre-collapse matter fields, the ingoing radiation which will have fallen in by time $t$, and a reference system $L$, and then evolve using the dynamics from subsection \ref{fundyn}, then at time $t$ the fundamental description state is
\be
|\Psi\ran_{LB^{(t)}R^{(t)}}=(V_t\otimes I_{LR^{(t)}})\left(|\chin\ran_{L\ell^{(t)}}|\psi_0\ran_f\otimes |\mathrm{MAX}\ran_{r^{(t)},R^{(t)}}\right)
\ee
In appendix \ref{dyn_app} we show that indeed we have  
\begin{align} \nonumber
S_2(\Psi_{R^{(t)}})&= \min\left[\,2 t \log q\, ,\,  (n_0 - t) \log q + S_2\left(\chi^{\{in\}}_{\ell^{(t)}}\right)\,  \right]\\ 
S_2(\Psi_{B^{(t)}}) &= \min\left[\,(n_0-t) \log q\,,\,2t \log q+ S_2\left(\chi^{\{in\}}_{\ell^{(t)}}\right)\,\right].\label{dynrenyi}
\end{align} 
The two quantities that we minimize between in each of these expressions are precisely the (Renyi) generalized entropies of the different quantum extremal surfaces in figure \ref{evapbhislandfig}. For example, when the entanglement wedge of $R$ consists of just $R$,  the generalized entropy of the QES is the ``effective-description entropy'' of $R$, $2t\log q$. When the entanglement wedge of $R$ also contains the ``island''  $\ell r$, then we instead have an ``area term'' $(n_0-t)\log q$ and an ``effective-description entropy term'' $_2(\chi^{\{in\}}_{\ell})$. As explained in Appendix \ref{dyn_app}, the area term is proportional to the number of legs in a minimal cut through the tensor network for $V_t$, and can therefore be interpreted more explicitly as an area term in this model.  We have not computed the analogous von Neumann entropies, but due to the ``fixed area'' nature of the model we expect \eqref{dynrenyi} to also hold with the Renyi entropies replaced by von Neumann entropies.

Finally we can consider reconstruction of interior operators onto $B$ and $R$.  In appendix \ref{dyn_app} we show that in the BRU model an operator on some particular qudit $\ell_i$ can be reconstructed on $B$ when $S_2(\Psi_{R^{(t)}})\approx 2 t \log q$ and on $R$ when $S_2(\Psi_{B^{(t)}})\approx (n_0-t) \log q$, just as entanglement wedge reconstruction would suggest.  In the RPI model an even more precise result is given: these conclusions hold \textit{provided} that $\ell_i$ fell in at least a time $\log (n_0-t)$ before the present, as is required by the Hayden/Preskill scrambling argument.  The BRU model cannot see this time delay since it scrambles too fast.

\section{Coarse-graining and interior complexity}\label{coarsesec}
So far in this paper our only use of complexity has been as a limitation on validity of the effective description: effective-description observers can only expect sensible answers to sub-exponential questions.  It is also natural to consider what can be learned by measuring only observables of sub-exponential complexity in the fundamental description.   This question can be rephrased in terms of coarse-graining: what information do we throw out by limiting ourselves to sub-exponential observables in the fundamental description?  In AdS/CFT attempts were made to study coarse-graining in the dual CFT in various guises over the years (see e.g.\cite{Hubeny:2012wa,Kelly:2013aja,Engelhardt:2017wgc}), with the current understanding primarily based on the ``simple entropy'' approach of Engelhardt and Wall  \cite{Engelhardt:2017aux,Engelhardt:2018kcs} and its converse the ``python's lunch'' proposal of \cite{Brown:2019rox,Engelhardt:2021mue,Engelhardt:2021qjs}.  In this section we explain how these ideas are realized in our model(s), and we also use them to give an interpretation in the fundamental description of Hawking's ``wrong'' calculation of the entropy of the radiation: the radiation entropy which Hawking computed is really the simple entropy rather than the von Neumann entropy, while it is the latter that follows the Page curve as we have seen.  

\subsection{Simple entropy and the outermost wedge}\label{simpentsec}
The key idea of \cite{Engelhardt:2017aux,Engelhardt:2018kcs} is that coarse-graining in the fundamental description throws away information behind horizons.  They refined this idea into a precise formula for coarse-grained entropy in the fundamental description in terms of geometry in the effective description: the \textit{simple entropy formula}.  The coarse-graining protocol used in \cite{Engelhardt:2017aux,Engelhardt:2018kcs} is a special case of a more general procedure of Jaynes \cite{Jaynes:1957zza, Jaynes:1957zz}, and is defined in the following way.  Let $A$ be a quantum system, and $\Psi_A$ a (possibly mixed) state on $\sH_A$.  We then define the set $\mathcal{S}_{\Psi_A}$ of states $\Phi_A$ on $\sH_A$ such that for all sub-exponential observables $O_A$ on $A$ we have
\be
\tr(O_A\Phi_A)\approx \tr(O_A\Psi_A),
\ee
with $\approx$ defined as in section \ref{measurementsec} to mean ``equal up to an error which is exponentially small in $\log |A|$''.\footnote{In practice we are interested in a slightly modified definition, where instead of $\log |A|$ we have $\log |B|$ where $B$ are the black hole degrees of freedom in the fundamental description regardless of whether or not they coincide with $A$.  Since we are only really interested in situations where $|A|$ is not exponentially bigger or smaller than $|B|$ however, this distinction will not matter.}  The simple entropy of $\Psi_A$ is then defined by 
\be\label{eq:SsimpleA}
S^{{\rm simple}}[\Psi_A]\equiv \sup_{\Phi_A\in \mathcal{S}_{\Psi_A}} S[\Phi_A].
\ee
The simple entropy is thus the maximal fine-grained entropy that is compatible with what a subexponential observer can extract from (many copies of) $\Psi_A$.  The proposal of \cite{Engelhardt:2017aux,Engelhardt:2018kcs, Bousso:2019dxk} is that the simple entropy of $\Psi_A$ should be computed by the generalized entropy of the ``outermost'' QES for $A$, regardless of whether it is minimal:
\be \label{eq:SsimpdualA}
S^{\rm simple}[\Psi_{A}]= S_{\rm gen}[X^{\rm outer}_A]. 
\ee
Here $X^{\rm outer}_A$ is whichever QES homologous to $A$ has the property that its outer wedge, which in this case is called the \textit{outermost wedge}, is contained in the outer wedge of any other QES homologous to $A$.\footnote{See our notation section for definitions of terms like ``homologous to $A$'' and ``outer wedge''.  Existence of the outermost QES -- i.e., one which is contained in the outer wedge of every QES homologous to $A$ -- was proved in~\cite{Engelhardt:2021mue, EngPenTA}. See also \cite{Engelhardt:2018kcs} for many more features of this definition.}  The argument of~\cite{Engelhardt:2017aux,Engelhardt:2018kcs} for~\eqref{eq:SsimpdualA} proceeds by explicitly constructing a state $\Phi_{A}$ in which $X^{\rm outer}_A$ is a minimal QES, and then applies the usual QES formula to compute its entropy in the fundamental description.  We will see in the next subsection that the outermost wedge of $A$ is essentially the region that can be accessed using sub-exponential observables on $A$.\footnote{Note that this does not generally coincide with the region outside of the event horizon: the outermost wedge typically includes some spacetime behind the event horizon.}

\bfig
\includegraphics[height=3cm]{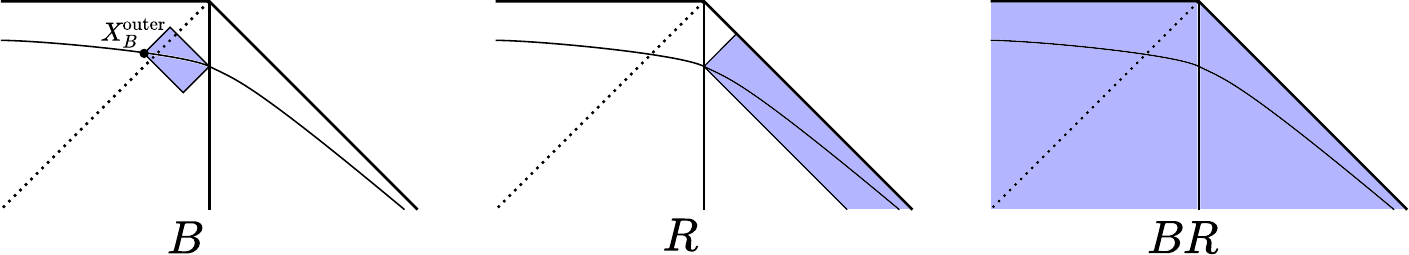}
\caption{Outermost wedges of various subsets of the fundamental degrees of freedom for an evaporating black hole.}\label{outerfig}
\efig
We now apply these ideas to an evaporating black hole, which we take to have been created by a fast collapse of a state of sub-exponential complexity and then evolved for a time which is at least its scrambling time and at most sub-exponential in its final entropy.  In figure \ref{outerfig} we show the outermost wedges for the black hole $B$, its radiation $R$, and their union $BR$.  Here $X^{{\rm outer}}_B$ is the same QES used to compute the Page curve at late times in \cite{Almheiri:2019psf,Penington:2019npb} (see the right diagram of figure \ref{evapbhislandfig}), while $X^{{\rm outer}}_R$ and $X^{{\rm outer}}_{BR}$ are the both the empty set.  For $X^{{\rm outer}}_{BR}$ the empty set has no competition and therefore wins by default, while for $B$ and $R$ the empty set and $X^{{\rm outer}}_B$ are competitors, and we choose whichever has the smaller outer wedge.  For comparison we note that at early times the minimal QES for both $B$ and $R$ is the empty set while at late times the minimal QES for both is $X^{{\rm outer}}_B$.  The outermost wedge has no such transition: as far as coarse-grained entropy is concerned there is nothing special about the Page time.  We can use these outermost wedges to compute the simple entropy using \eqref{eq:SsimpdualA}; in the notation of our models we have
\begin{align}\nonumber
S^{{\rm simple}}[\Psi_B]&=\log |B|\\\nonumber
S^{{\rm simple}}[\Psi_R]&=S(\chot_R)\\
S^{{\rm simple}}[\Psi_{BR}]&=0 .\label{simplents}
\end{align}
As one might expect, the coarse-grained entropy of the black hole is given by the Bekenstein-Hawking formula and the coarse-grained entropy of the radiation is given by Hawking's calculation.  The third result may seem more surprising, but recall that we have only evolved our black hole for a sub-exponential amount of time, so an observer in the fundamental description with access to both $B$ and $R$ can simply evolve it backwards to recover the sub-exponential pure state from which it was made.  

Let's see to what extent we can reproduce the simple entropies \eqref{simplents} in our models.  Unlike the QES-based approach of \cite{Engelhardt:2017aux,Engelhardt:2018kcs}, we will here proceed by directly computing the simple entropies in the fundamental description using the definition \eqref{eq:SsimpleA}.  We first consider the static model of section \ref{modelsec} whose encoding map is defined by \eqref{Vdef}, and we will compute the simple entropies for the  encoded version $\Psi_{LBR}$, defined by \eqref{hawkingencode}, of the Hawking state \eqref{HawkS}.  The key observation is that, due to measure concentration, it is almost surely true that for all sub-exponential observables $O_{BR}$ on $BR$ we have
\be\label{avO}
\tr(O_{BR}\Psi_{BR})\approx \tr(O_{BR}\Phi_{BR}),
\ee
where $\Phi_{BR}$ is the ``average'' state
\be
\Phi_{BR}\equiv \int dU \Psi_{BR}=\frac{I_B}{|B|}\otimes \chot_R.
\ee
That \eqref{avO} holds on average is obvious; our goal is to show that the fluctuations are small and thus that it holds also in each typical instance of $U$.  Indeed we have
\begin{align}\nonumber
\int dU \Big(\tr\big(O_{BR}(\Psi_{BR}-\Phi_{BR})\big)\Big)^2&=\frac{|P|^2|B|^2}{|P|^2|B|^2-1}\Bigg[e^{-S_2\left(\chot_R\right)}\tr\left(O_{BR}\Phi_{BR}\right)^2+\frac{\left(\tr\left(O_{BR}\Phi_{BR}\right)\right)^2}{|P|^2|B|^2}\\\nonumber
&\qquad-\frac{1}{|P||B|}\Bigg(e^{-S_2\left(\chot_R\right)}\tr_R\left(\tr_B \left(O_{BR}\right)\chot_R\right)^2\\\nonumber
&\qquad+\tr_B\left(\tr_R(O_{BR}(I_B\otimes \chot_R))\right)^2\Bigg)\Bigg]\\
&\leq \frac{4}{|B|}||O||_\infty^2 e^{-S_2\left(\chin_L\right)-S_2\left(\chot_R\right)},\label{BRbound}
\end{align}
where in the inequality we have used that for any state $\rho$ and operator $O$ we have
\be
0\leq\tr(\rho O \rho O)\leq ||\rho||_2^2||O||_\infty^2.
\ee
The lower bound follows from $\tr(\rho O \rho O)=||\sqrt{\rho}O\sqrt{\rho}||_2^2$, while the upper bound follows from H\"older's inequality \eqref{Holder} with $p=q=2$ and also \eqref{triplebound}.  By Jensen's inequality we then also have
\be\label{simplebound}
\int dU |\tr\big(O_{BR}(\Psi_{BR}-\Phi_{BR})\big)|\leq \frac{2}{\sqrt{|B|}}||O||_\infty e^{-\left(S_2\left(\chin_L\right)+S_2\left(\chot_R\right)\right)/2}
\ee
Thus we see that any particular observable $O_{BR}$ is quite likely to have an expectation value in $\Psi_{BR}$ which is quite close to its average value.  We can then use a measure concentration argument similar to that of appendix \ref{fixdcapp}  to show that this statement is likely to hold for all sub-exponential $O_{BR}$.\footnote{The analogue of $F(U)$ in appendix \ref{fixdcapp} here is $G(U)\equiv \frac{|\tr\big(O_{BR}(\Psi_{BR}-\Phi_{BR})\big)|}{||(V\otimes I_{LR})|\psi_{{\rm Hawk}}\ran||+1}$, and the analogue of \ref{Fbound} is $G(U)\leq ||O||_\infty(||(V\otimes I_{LR})|\psi_{{\rm Hawk}}\ran||+1)$.  The Lipschitz constant for $G(U)$ is of order $||O||_\infty \sqrt{|P|}$.}  

We can now compute the various simple entropies in the static model.  For $B$ things are straightforward: $\Phi_B$ is already the maximally mixed state, so we must have
\be
S^{{\rm Simple}}[\Psi_B]=\log |B|,
\ee
just as in \eqref{simplents}.  If we take $|\chot\ran_R$ to be maximally mixed then we can also reproduce the second line of \eqref{simplents} by the same argument.  More generally we can use the following lemma:
\begin{lemma}\label{factorlemma}
Let $\sH_A$ be a Hilbert space which tensor factorizes into a sub-exponential (in $\log |A|$) number $n$ of factors $\sH_{A_i}$ of $O(1)$ (in $\log |A|$) size, and let $\Psi_A$ be a state on $\sH_A$ with the property that there exists $\Phi_A\in \mathcal{S}_{\Psi_A}$ such that
\be\label{eq:factorizationassumptionA}
\Phi_A\approx \bigotimes_{i=1}^n \Phi_{A_i},
\ee
with ``$\approx$'' here meaning that the trace norm of the difference of the two states is exponentially small in $\log |A|$.  Then we have
\be
S^{{\rm Simple}}[\Psi_A]\approx S[\Phi_A].
\ee
\end{lemma}
\begin{proof}
It is an immediate consequence of $\Phi_A\in \mathcal{S}_{\Psi_A}$ that 
\be\label{lowersimpA}
S^{{\rm Simple}}[\Psi_A]\geq S[\Phi_A].
\ee
Our goal is thus to prove the opposite inequality 
\be\label{uppersimpA}
S^{{\rm Simple}}[\Psi_A] \lesssim S[\Phi_A],
\ee
meaning that $S^{{\rm Simple}}[\Psi_A]\leq S[\Phi_A]+\epsilon$ with $\epsilon$ exponentially small in $\log |A|$.  

By the definition of $S^{\rm simple}[\Psi_{A}]$, for any $\epsilon>0$ there exists a state $\wt{\Phi}_{A}$ in $\mathcal{S}_{\Psi_A}$ such that $S^{\rm simple}[\Psi_{A}]<S[\wt{\Phi}_{A}]+\epsilon$. Choosing $\epsilon$ to be exponentially small in $\log |A|$, we can find $\wt{\Phi}_A\in \mathcal{S}_{\Psi_A}$ such that
\be\label{phidefA}
S^{\rm simple}[\Psi_{A}]\lesssim S[\wt{\Phi}_{A}].
\ee
By iterative applications of subadditivity of von Neumann entropy we have
\begin{align}\nonumber
S[\wt{\Phi}_{A}]& \leq S[\wt{\Phi}_{A_{1}A_{2}\cdots A_{n-1}}] + S[\wt{\Phi}_{A_{n}}]\\\nonumber
& \leq S[\wt{\Phi}_{A_{1}A_{2}\cdots A_{n-2}}] + S[\wt{\Phi}_{A_{n-1}}] +S[\wt{\Phi}_{A_{n}}]\\\nonumber
& \vdots \\
&\leq \sum\limits_{i=1}^{n} S[\wt{\Phi}_{A_{i}}].\label{subaddA}
\end{align}
Note that we have \textit{not} assumed that $\wt{\Phi}_{A}$ factorizes (in contrast with our assumption on $\Phi_A$). 
By definition of $\mathcal{S}_{\Psi_A}$ (and the fact that $\Phi_A,\wt{\Phi}_A\in \mathcal{S}_{\Psi_A}$), for any subexponential observable $O$ we have 
\be
\tr (O\Psi_{A})\approx\tr (O\Phi_{A})\approx \tr (O\wt{\Phi}_A) ,
\ee
where again $\approx$ means equality up to terms that are exponentially suppressed in $\log |A|$.  If we restrict our attention to observables $O_{i}$ that act nontrivially only on one of the $A_{i}$ (or equivalently to observables that factorize across the $A_{i}$), then:
\be \label{eq:subsystemsimpleA}
\tr (O_{i}\wt{\Phi}_{A_{i}})\approx \tr (O_{i}\Phi_{A_{i}}).
\ee
Because $|A_{i}|$ is by assumption of order $(\log |A|)^{0}$,  an observable $O_{i}$ that acts only on a single $A_{i}$ cannot have complexity which is exponential in $\log|A|$.  Thus~\eqref{eq:subsystemsimpleA} is true for any factorizing $O_{i}$ without the need for an a priori restriction to subexponential operators, i.e. \eqref{eq:subsystemsimpleA} holds for any observable acting on each individual factor. 

We can now make use of H\"older's inequality \eqref{Holder},
which tells us that
\be
\max\limits_{O_i} \left |\tr \left (O_{i}(\wt{\Phi}_{A_{i}}-\Phi_{A_{i}})  \right) \right|=\Vert \wt{\Phi}_{A_{i}} -\Phi_{A_{i}}\Vert_{1},
\ee
with the maximum taken over all operators $O_i$ with $\Vert O_i\Vert_{\infty}=1$. Any operator is a linear combination of two observables, so by \eqref{eq:subsystemsimpleA} the RHS vanishes up to exponentially suppressed corrections. We thus find that the trace distance from $\wt{\Phi}_{A_{i}}$ to $\Phi_{A_{i}}$ is exponentially suppressed in $\log|A|$. We can then use Fannes' inequality 
\be
\left | S\left[ \wt{\Phi}_{A_{i}}\right]- S\left[\Phi_{A_{i}}\right]\right|\leq \Vert \wt{\Phi}_{A_{i}} -\Phi_{A_{i}}\Vert_{1} \log |A_i| - \Vert \wt{\Phi}_{A_{i}} -\Phi_{A_{i}}\Vert_{1}\log \Vert\wt{\Phi}_{A_{i}} -\Phi_{A_{i}}\Vert_{1} 
\ee
to conclude that
\be
S\left [ \wt{\Phi}_{A_{i}}\right]\approx S\left[ \Phi_{A_{i}}\right].
\ee
Combining this with \eqref{phidefA} and \eqref{subaddA} we then have
\be
S^{\rm simple}[\Psi_{A}]\lesssim S[\wt{\Phi}_{A}]\leq \sum\limits_{i=1}^{n} S[\wt{\Phi}_{A_{i}}] \approx\sum\limits_{i=1}^{n} S[\Phi_{A_{i}}] \approx S[\Phi_A],
\ee
where in the last step we used our factorization assumption \eqref{eq:factorizationassumptionA}, and also Fannes' inequality, and in both approximate equalities we have used that $n$ is sub-exponential.  \eqref{uppersimpA} is thus established, which together with \eqref{lowersimpA} completes the proof.
\end{proof}
Turning now to the simple entropy of the radiation $R$, without any additional assumptions about the nature of $\chot_R$ this lemma is of no use.  Recalling however that the Hawking state is often approximated in terms of exterior outgoing modes and their interior partners entangled across the horizon, it is natural to assume that $\chot_R$ indeed has a product structure as in \eqref{eq:factorizationassumptionA}.  Indeed this product structure is manifest in our dynamical model from section \ref{dynamicsec}.  With this assumption, we then  see that $\Phi_{BR}$ satisfies the assumptions of lemma \ref{factorlemma} and thus we have
\begin{align}\nonumber
S^{{\rm Simple}}[\Psi_{R}]&=S\left[\chot_R\right]\\
S^{{\rm Simple}}[\Psi_{BR}]&=\log|B|+S\left[\chot_R\right].
\end{align}
The first of these agrees with \eqref{simplents}, and thus gives confirmation that in our static model Hawking's calculation of the radiation entropy is really computing the simple entropy in the fundamental description.  The second however is \textit{not} compatible with \eqref{simplents}: the latter gave $S^{{\rm Simple}}[\Psi_{BR}]=0$ while here we got something large.  

Why did the static model give the wrong answer for $S^{{\rm Simple}}[\Psi_{BR}]$?  Recall that below \eqref{simplents}, we observed that the reason this should vanish is because we can evolve the black hole backwards in time to see how it was made.  Inspecting figure \ref{Vdeffig} however there are two reasons why we cannot do this in the static model:
\bi
\item[(1)] The unitary $U$ was choosen at random from the Haar ensemble, and thus is almost surely exponentially complex.
\item[(2)] Even if $U$ were not exponentially complex, the post-selection onto $\lan 0|_P$ means that we cannot just apply $U^\dagger$ to run the system backwards.  
\ei
Problem (1) can be avoided by taking $U$ to be a unitary $k$-design (see appendix \ref{Uapp}) instead of a Haar random unitary, but problem (2) is not so easily avoided.  Indeed to deal with problem (2) we need to introduce additional structure to the model.  Fortunately however we have already done so: either of the dynamical models we introduced in the previous section indeed has vanishing $S^{{\rm Simple}}[\Psi_{BR}]$.  This is clear from figure \ref{fig:bh_dyn}: in the fundamental description evolution there is no post-selection, so we can simply apply the dynamical unitaries $U_m$ in reverse to extract the initial state.\footnote{To do this in the BRU model we also need to take the $U_m$ to be $k$-designs instead of Haar random in order to solve problem (1).} On the other hand, since these unitaries act on both $B$ and $R$ we cannot do this backward evolution if we have access only to $B$ or only to $R$; there is thus no similar obstruction to $\Psi_B$ and $\Psi_R$ having nonzero simple entropies.

Showing that $\Psi_B$ and $\Psi_R$ do not need to have small simple entropies is not the same as showing that their simple entropies are given by \eqref{simplents}.  In the static model we were able to prove this using measure concentration, but once we replace random unitaries by $k$-designs we lose the ability to make such arguments since there is no $k$-design version of lemma \eqref{Meckes}.  Have we thus sacrificed $S^{{\rm Simple}}[\Psi_B]$ and $S^{{\rm Simple}}[\Psi_R]$ to recover $S^{{\rm Simple}}[\Psi_{BR}]$?  Unfortunately we are not able to answer this question with a decisive ``no'', essentially because showing things like this is in the same difficulty class as proving that ${\rm P}\neq{\rm NP}$, but we will now argue that the answer is a confident ``probably not'' (just as it is for asking if ${\rm P}={\rm NP}$).  We first observe that the question of whether or not information about $\ell$ is accessible by sub-exponential measurements in $B$ or $R$ (but not both) is equivalent to the question of whether or not, upon maximally entangling the inputs $\ell$ to the dynamical model with a reference system $L$, a sub-exponential observer in $B$ or $R$ (but not both) can extract something which is correlated with $L$.  Up to notation this is precisely the task which was considered in \cite{Harlow:2013tf} in the context of the black hole firewall arguments, and there it was argued that standard assumptions about the difficulty of a quantum complexity class called ``quantum statistical zero knowledge'' imply that no such correlation can be extracted in less than exponential time.  A nicer argument for the same, based on the existence of injective one-way functions, was soon after given by Aaronson, as reported in section 8.1 of \cite{Harlow:2014yka}.  The existence of one-way functions is the foundation of modern cryptography \cite{arora2009computational}, and it is widely expected to hold.  Any failure of \eqref{simplents} in our dynamical models must thus rely on some further structure beyond simply the fact that the dynamical unitaries have sub-exponential size, and we see no reason for such structure to exist.

\subsection{Complexity of distinguishing interior states}
We now study more directly the question of what information about the black hole interior can be accessed using sub-exponential observables on a set $A$ of degrees of freedom in the fundamental description.  Following \cite{Engelhardt:2017aux,Engelhardt:2018kcs} we can guess that the answer to this question changes qualitatively depending on whether or not the interior lies in the outermost wedge of $A$ (see also \cite{Almheiri:2021jwq} for an alternative perspective reaching a similar conclusion).  More concretely, defining
\be
\Psi_{BR}(W_{\ell r})\equiv (V\otimes I_{R})(W_{\ell r}\otimes I_{R})|\psi_{{\rm Hawk}}\ran\lan\psi_{{\rm Hawk}}|(W_{\ell r}^\dagger \otimes I_{R})(V^\dagger \otimes I_{R}),
\ee
where $|\psi_{{\rm Hawk}}\ran$ is the Hawking state \eqref{HawkS} but with $L$ taken to be trivial, i.e. with with $|\chin\ran_{L\ell}=|\chin\ran_\ell$, we will study to what extent sub-exponential observables on $A$ can distinguish the states $\Psi_A(W_{\ell r})$ and $\Psi_{A}(W_{\ell r}')$, with $W_{\ell r}$ and $W_{\ell r}'$ both sub-exponential, in the sense that
\be\label{distinguish}
\left|\tr\left[O_A(\Psi_A(W_{\ell r})-\Psi_A(W_{\ell r}'))\right]\right|>\delta
\ee
for some $\delta$ which is not exponentially small and some sub-exponential observable $O_A$ obeying $||O_A||_\infty\leq 1$.\footnote{\label{distrec}The relationship of this question to state-specific reconstruction is the following: if there is a projection $\Pi_A$ which can distinguish $|\pm\ran\equiv (V\otimes I_{R})\left(|\Psi(W_{\ell r})\ran\pm e^{-i\phi}|\Psi(W_{\ell r}')\ran\right)$ to accuracy $\epsilon$ in the sense that $\lan +|\Pi_A|+\ran/\lan +|+\ran\geq 1-\epsilon$ and $\lan -|\Pi_A|-\ran/\lan-|-\ran\leq \epsilon$, where $e^{i\phi}$ is defined by $\lan\Psi(W_{\ell r})|\Psi(W_{\ell r}')\ran =e^{i\phi}\cos \theta$, then the unitary $W_A=e^{i\phi}\left(2\Pi_A-I_A\right)$ will map $|\Psi(W_{\ell r})\ran$ to $|\Psi(W_{\ell r}')\ran$ with accuracy of order $\epsilon$.  In \eqref{distinguish} however we are only asking if we can distinguish the two states at all, so $\epsilon$ need not be small.  On the other hand, in practice if we can distinguish them at all then we can usually distinguish them well.}  We always assume that the states $W_{\ell r}|\psi_{{\rm Hawk}}\ran$ and $W_{\ell r}'|\psi_{{\rm Hawk}}\ran$ are easily distinguishable in the effective description, and in particular that $W_{\ell r}\neq W_{\ell r}'$.

We first consider the case where $W_{\ell r}$ and $W_{\ell r}'$ act only on $\ell$, in which case things are simple.  Indeed these states differ from the Hawking state only by the quantum state of the infalling matter, and at least so long as they do not change the mass by a large amount this will not create substantial backreaction.  The geometry is thus essentially the same as for $|\psi_{{\rm Hawk}}\ran$, and the outermost wedges are as shown in figure \ref{outerfig}.  We thus can guess that we can sub-exponentially distinguish the two states on $BR$ but not on $B$ or $R$ separately.  The former is indeed is true: on $BR$ we can evolve the system backwards in time, which causes all infalling matter to evolve back out of the black hole into $R$ where it can easily be measured.  The latter is also plausible, for the same complexity-theoretic reasons given at the end of the previous subsection.  

\bfig
\includegraphics[height=6cm]{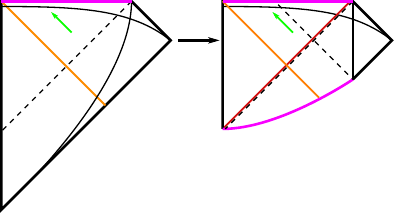}
\caption{Back-reaction on the global geometry of an evaporating black hole caused by disturbing a right-moving interior mode.  The orange line is the collapsing shell which formed the black hole, the thin vertical black line is the join between the radiation system and the black hole region, the red line is the shockwave which is created by disturbing the mode, the horizontal thin black line is the Cauchy slice on which we disturb the mode, and the purple lines are singularities.  The dashed lines are future/past horizons.  Note in particular the appearance of a past singularity in the back-reacted geometry, which didn't exist for the black hole formed from collapse.  The green arrow is an infalling mode which in the unperturbed geometry can be evolved back out to the radiation, but which evolves back to the past singularity in the back-reacted geometry.}\label{pastsingfig}
\efig
The situation is more complicated for general $W_{\ell r}$ and $W_{\ell r}'$.  The reason is that for black holes which were formed at least a scrambling time in the past, disturbing even a single right-moving interior mode and then evolving backwards in time leads to a back-reaction which changes the global geometry in a major way \cite{tHooft:1984kcu,Kiem:1995iy}.  In \cite{Engelhardt:2021qjs} it was explained how this effect can lead to a large change in what information is sub-exponentially accessible.  In the context of an evaporating black hole the relevant back-reaction is shown in figure \ref{pastsingfig}: disturbing a right-moving mode on the indicated Cauchy slice creates a shockwave, shaded red in the figure, which evolves backwards to create a new past singularity which wasn't there in the original geometry.  This is important for sub-exponential reconstruction for two reasons:
\bi
\item[(1)] Any left-moving mode that fell into the black hole at least a scrambling time before we acted with $W_{\ell r}$ will now evolve back to the past singularity rather than leaving the black hole and entering radiation, which means that we can no longer use backward evolution on $BR$ to diagnose whether or not $W_{\ell r}$ acted nontrivially on this mode.  
\item[(2)] As explained in \cite{Engelhardt:2021qjs}, the presence of this shockwave leads to the creation of a new QES homologous to $BR$, which lies in the vicinity of the sphere where the future and past horizons cross in the left diagram of figure \ref{pastsingfig}.  This QES is a new candidate for the outermost QES homologous to $BR$, and it brings about a large change in the outermost wedge of $BR$ (see figure \ref{newouterfig}). 
\ei
These changes do not have much effect on the difficulty of distinguishing $\Psi_{BR}(W_{\ell r})$ and $\Psi_{BR}(W_{\ell r}')$ on $B$ or $R$ separately, since those were already exponentially hard for interior unitaries with support only on $\ell$, but they are essentially dual ways of seeing that accessing interior information using sub-exponential observables on $BR$ is more difficult in general than it is for unitaries with support only on $\ell$.
\bfig
\includegraphics[height=4cm]{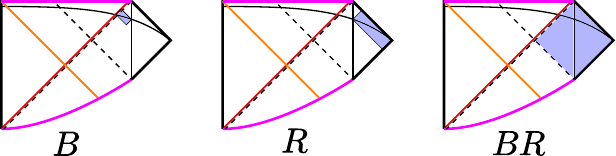}
\caption{Outermost wedges for $B$, $R$, and $BR$ in the back-reacted geometry, shaded blue.  Note that the outermost QES for $B$ is \textit{not} the same as the outermost QES for $BR$; the former is the QES used in \cite{Almheiri:2019psf,Penington:2019npb} to compute the Page curve, while the latter is the new QES constructed in \cite{Engelhardt:2021qjs}.}\label{newouterfig}
\efig

Let's now consider how hard it is to distinguish $\Psi_{BR}(W_{\ell r})$ and $\Psi_{BR}(W_{\ell r}')$ in our models.  We first consider the static model.  By the triangle inequality and \eqref{simplebound} we have
\be
\int dU|\tr\left[O_{BR}(\Psi_{BR}(W_{\ell r})-\Psi_{BR}(W_{\ell r}'))\right]|\leq \frac{4}{\sqrt{|B|}}||O||_\infty e^{-S_2\left(\chot_R\right)/2},
\ee
where $O_{BR}$ is any fixed observable on $BR$ and we have used that $\chin_\ell$ is pure.  Here we have made use of the fact that in the calculation \eqref{BRbound} we can replace $\Psi_{BR}$ by $\Psi_{BR}(W_{\ell r})$ without changing the result, since in $\int dU|\tr\left[O_{BR}(\Psi_{BR}(W_{\ell r})-\Phi_{BR}\right]|^2$ the operators $U$ and $W_{\ell r}$ only appear together and thus $W_{\ell r}$ can be removed by the change of variables $U'=UW_{\ell r}$.  We may then again use measure concentration to extend this conclusion to all sub-exponential $O_{BR}$. Thus we see that $\Psi_{BR}(W_{\ell r})$ and $\Psi_{BR}(W_{\ell r}')$ cannot be distinguished by any sub-exponential observable.  This of course implies the same for distinguishing $\Psi_{B}(W_{\ell r})$ from $\Psi_{B}(W_{\ell r}')$ and $\Psi_{R}(W_{\ell r})$ from $\Psi_{R}(W_{\ell r}')$.  For generic $W_{\ell r}$ and $W_{\ell r}'$ these conclusions are consistent with what one would expect from figure \ref{newouterfig}, but they are not correct when these unitaries have support only on $\ell$ since then we can sub-exponentially distinguish the states on $BR$ by evolving backwards in time.  We will also see in a moment that the situation on $BR$ for general $W_{\ell r}$ is more refined than this calculation would suggest.  

\bfig
\includegraphics[height=7cm]{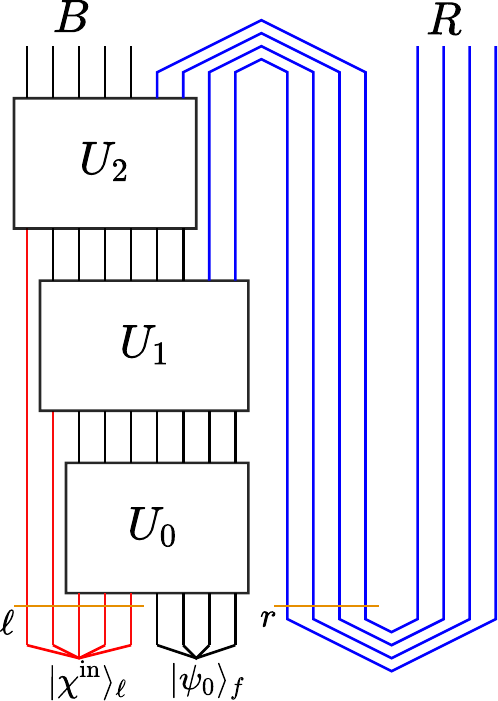}
\caption{Acting with an interior unitary $W_{\ell r}$ in the dynamical model.  The unitary acts on the two orange lines.}\label{Wdynfig}
\efig
We now consider the dynamical model of section \ref{dynamicsec}, focusing on the BRU version of the model with the dynamical unitaries $U_m$ taken to be $k$-designs.  We will primarily consider distinguishing $\Psi_{BR}(W_{\ell r})$ and $\Psi_{BR}(W_{\ell r}')$, since distinguishing the states on $B$ or $R$ separately works in the same way as for the static model.  The basic situation is illustrated in figure \ref{Wdynfig}.  If $W_{\ell r}$ and $W_{\ell r}'$ are supported only on $\ell$, then we can ``straighten out'' the blue lines to remove all post-selection, which is the inverse of the operation shown in figure \ref{dynevapfig}, and then apply the inverses of the dynamical unitaries in reverse order to directly extract the infalling matter and thus easily distinguish which unitary we acted with.  This ``straightening out'' operation however is disrupted the moment we allow $W_{\ell r}$ and $W_{\ell r}'$ to act on $r$, which calls into question whether or not the states can be still sub-exponentially distinguished.  This can be thought of as the manifestation in the fundamental description of the back-reaction phenomenon shown in figures \ref{pastsingfig}, \ref{newouterfig}.  On the other hand, if $W_{\ell r}$ and $W_{\ell r}'$  only affect a small number of the right-moving modes $r$, then we can still straighten out \textit{most} of the blue lines in figure \ref{Wdynfig}; one may thus wonder if that this still makes the states $\Psi_{BR}(W_{\ell r})$ and $\Psi_{BR}(W_{\ell r}')$ easier to distinguish than they would be on $B$ or $R$.  We now argue that it does, although only in a way whose complexity remains exponential in the number of modes which are affected by $W_{\ell r}$ and $W_{\ell r}'$.  

The basic idea, following \cite{Brown:2019rox}, is to use a multi-state version of Grover's algorithm, which we formalize as lemma \ref{groverlem2}, to ``undo'' the post-selection in figure \ref{Wdynfig} and thus extract the initial state.  More concretely, we first do a relabeling (which applies \textit{only} in this paragraph) where we straighten out any blue lines in figure \ref{Wdynfig} which are not affected by $W_{\ell r}$ and then view them as part of $B$.  The symbols $r$ and $R$ then refer only to the legs which remain bent, and the symbol $P$ refers to remaining copy of the new $rR$ which is still post-selected on.  We now construct a unitary $U_{BRP}$ such that
\begin{align}\nonumber
U_{BRP}^\dagger \Big(V W_{\ell r}|\psi_{{\rm Hawk}}\ran\otimes|\mathrm{max}\ran_P\Big)&\approx UW_{\ell r} |\psi_{{\rm Hawk}}\ran|\psi_0\ran_f\\
U_{BRP}^\dagger \Big(V W_{\ell r}'|\psi_{{\rm Hawk}}\ran\otimes|\mathrm{max}\ran_P\Big)&\approx UW_{\ell r}' |\psi_{{\rm Hawk}}\ran|\psi_0\ran_f.\label{groverproj}
\end{align}
Here $U=U_m U_{m-1}\ldots U_0$, and we will see that the complexity of $U_{BRP}$ is of order $\sqrt{|P|}$.  We can then distinguish $\Psi_{BR}(W_{\ell r})$ and $\Psi_{BR}(W_{\ell r}')$ by measuring 
\be
O_{BR}\equiv \lan{\rm max}|_PU_{BRP}O_{\ell r R}U^\dagger U_{BRP}^\dagger|{\rm max}\ran_P,
\ee
where $O_{\ell rR}$ is any observable which distinguishes $W_{\ell r}|\psi_{\rm Hawk}\ran$ and $W_{\ell r}'|\psi_{\rm Hawk}\ran$.  To construct $U_{BRP}$, it is convenient to define $\Pi_P\equiv |{\rm max}\ran\lan{\rm max}|_P$ and introduce an orthonormal basis $|\psi_i\ran$, with $i=1,2$, for the span of $UW_{\ell r} |\psi_{{\rm Hawk}}\ran|\psi_0\ran_f$ and $UW_{\ell r}' |\psi_{{\rm Hawk}}\ran|\psi_0\ran_f$. We then have
\be
|\psi_i\ran=\sqrt{\lan \psi_i |\Pi_P|\psi_i\ran}\frac{\Pi_P|\psi_i\ran}{\sqrt{\lan \psi_i |\Pi_P|\psi_i\ran}}+\sqrt{\lan \psi_i|(1-\Pi_P)|\psi_i\ran}\frac{(1-\Pi_P)|\psi_i\ran}{\sqrt{\lan \psi_i|(1-\Pi_P)|\psi_i\ran}}.
\ee
Noting that $W_{\ell r}|\psi_{\rm Hawk}\ran$ and $W_{\ell r}'|\psi_{\rm Hawk}\ran$ are both sub-exponential states, by lemma \ref{suplemapp} and the (conjectural) analogue of \eqref{subexpconc} for the dynamical model we see that
\be
\lan\psi_i|\Pi_P|\psi_j\ran\approx \frac{1}{|P|}\delta_{ij}.
\ee
Therefore we have
\be
|\psi_i\ran\approx \sin \theta|\phi_i\ran+\cos \theta |\phi_i^\perp\ran,
\ee
with 
\begin{align}\nonumber
\sin \theta &\equiv \frac{1}{\sqrt{|P|}}\\\nonumber
|\phi_i\ran&\equiv \sqrt{|P|}\Pi_P|\psi_i\ran\\
|\phi_i^\perp\ran&\equiv \sqrt{\frac{|P|}{|P|-1}}(1-\Pi_P)|\psi_i\ran.
\end{align}
Noting that the $|\phi_i\ran$ are approximately orthonormal, that $\lan \phi_i|\phi_j^\perp\ran=0$ for all $i,j$, and that the (sub-exponential) unitary 
\be
U_\phi\equiv I-2|\phi_1\ran\lan \phi_i|-2|\phi_2\ran\lan \phi_2|
\ee
acts as
\begin{align}\nonumber
U_\phi|\phi_i\ran&\approx-|\phi_i\ran\\
U_\phi|\phi_i^\perp\ran&\approx|\phi_i^\perp\ran,
\end{align}
we may then invoke lemma \ref{groverlem2} to construct a unitary $U_{BRP}$, whose complexity is of order $\sqrt{|P|}$ (unless that is small compared to the complexity of $|\phi_i\ran$ and $|\psi_i\ran$, in which case it is set by those), and which acts as\footnote{The various approximations here indicate errors in the vector norm which are of order $|B|^{-\gamma}$ for some $0<\gamma<\frac{1}{2}$.  The algorithm in \ref{groverlem2} relies on this approximation $O(\sqrt{|P|})$ times, so the algorithm will work up to accuracy which is of order $\sqrt{|P|}|B|^{-\gamma}$.  Recall however that we have already straightened any legs not affected by $W_{\ell r}$ or $W_{\ell r}'$, so $\sqrt{|P|}$ is only exponential in the number of affected qudits in $R$.}
\be
U_{BRP}|\psi_i\ran\approx |\phi_i\ran.
\ee
Translating this back into a statement about $W_{\ell r}$ and $W_{\ell r}'$ then gives \eqref{groverproj}.  

Thus we see that we can distinguish $\Psi_{BR}(W_{\ell r})$ and $\Psi_{BR}(W_{\ell r}')$ by a measurement whose complexity is exponential in number of modes in $r$ which are affected by $W_{\ell r}$ and $W_{\ell r}'$.  If this is an $O(1)$ fraction of the modes, as it will be for generic sub-exponential $W_{\ell r}$ and $W_{\ell r}'$, we then have that the states cannot be distinguished in sub-exponential (in $\log |B|$) time.  On the other hand if $W_{\ell r}$ and $W_{\ell r}'$ are simpler, and in particular if they only affect an $O(1)$ number of modes in $r$, then by using this algorithm we can distinguish the states in a time which remains sub-exponential in $\log |B|$.  The connection between the outermost wedge and sub-exponential distinguishability is thus in general more subtle than might be hoped.  On the other hand there \textit{is} a more subtle proposal for how these are related: the ``Python's lunch'' conjecture of \cite{Brown:2019rox,Engelhardt:2021mue,Engelhardt:2021qjs}.  This gives a concrete geometric proposal in the effective description for how difficult it should in the fundamental description to to distinguish pairs of states, and we now turn to arguing that for the dynamical model our results are in all cases consistent with this proposal.   

\subsection{Python's lunch conjecture}
In order for there to be a plausible geometric separation between states that can be distinguished using sub-exponential observables on $A$ and states that can be distinguished only using exponential observables on $A$, there need to be at least two QESs homologous to $A$: the minimal QES $X_A^{\rm min}$ and outermost QES $X_A^{\rm outer}$.  It was argued in \cite{Brown:2019rox} that generically this implies the existence of yet a third QES $X_A^{\rm bulge}$, whose generalized entropy is at least as big as both $X_A^{\rm min}$ and $X_A^{\rm outer}$, and whose outer wedge contains the outermost wedge of $A$ but is contained in its entanglement wedge.  In the event that there is more than one such surface, the correct one is identified using a ``maximinimax'' construction as follows: we first consider a spatial slice $\Sigma$ bookended by $X_{A}^{\rm min}$ and $X_{A}^{\rm outer}$ and a foliation of $\Sigma$ by codimension-two surfaces homologous to $A$, and within that foliation we find the surface of maximal $S_{\rm gen}$.  We then pick among all such foliations the one whose max surface has the smallest $S_{\rm gen}$.  Finally we maximize over all possible choices of $\Sigma$, resulting in a surface which can be shown to be a QES obeying the above properties: the ``bulge'' surface $X_A^{\rm bulge}$.  The authors of \cite{Brown:2019rox} then made the following proposal: the complexity of reconstructing observables which lie in the domain of dependence of the surface $\Sigma$ above, which is referred to as the ``Python's lunch'', is of order
\be\label{python}
\mathcal{C}\sim e^{\left(S_{gen}[X_A^{\rm bulge}]-S_{gen}[X_A^{\rm outer}]\right)/2}.
\ee
This is called the Python's lunch formula.  Since generically reconstruction should not be much harder than distinguishability (see footnote \ref{distrec}), it is natural to suppose that this formula should also control the complexity of distinguishing states which differ by the action of a unitary $W_{\ell r}$ which is supported only in the Python's lunch.  

\bfig
\includegraphics[height=3.5cm]{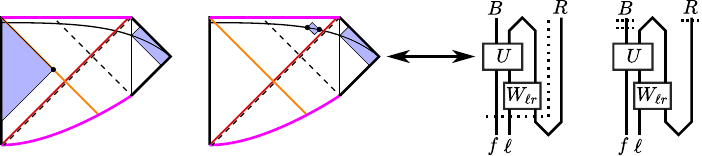}
\caption{Two possible bulges for $B$ and $R$ for the evaporating black hole perturbed by a generic $W_{\ell r}$, indicated by the black dots.  The shaded regions are their outer wedges assuming we are looking to construct $X^{\rm bulge}_R$, the causal complements of these are their outer wedges as candidates for $X^{\rm bulge}_B$.  We have also shown the analogous cuts through the quantum circuit for our dynamical model.}\label{bulges1}
\efig

The motivation for the Python's Lunch formula comes from quantum circuits. In this context, the relevant QESs -- $X_{A}^{\rm outer}$, $X_{A}^{\rm bulge}$, and $X_{A}^{\rm min}$ -- are cuts through the quantum circuit that prepares the state of $A$ in the fundamental description. To be more precise, both are locally minimal cuts, and $X_A^{\rm min}$ is also globally minimal.  The existence of two local minima -- $X_{A}^{\rm min}$ and $X_{A}^{\rm outer}$ -- implies that for any foliation of the network between $X_{A}^{\rm min}$ and $X_{A}^{\rm outer}$, there is a maximal cut $X_A^{\rm bulge}$.  The existence of this bulge gives an obstruction to interpreting the network as an isometry from the legs piercing $X_A^{\rm min}$ to the legs piercing $X_A^{\rm outer}$, as whenever size of the foliating surface decreases this means that post-selection is happening, and to undo post-selection using the multi-state Grover amplification protocol of the last section is expensive.  Indeed the amount of post-selection in the circuit is roughly captured by the constriction from $X_A^{\rm bulge}$ to $X_{A}^{\rm outer}$. This however can differ for different foliations of the circuit, so to find the most efficient isometry we should pick the foliation that minimizes this maximum.  Interpreting now this circuit as lying in the spatial slice $\Sigma$ as is usual in the study of tensor networks these statements translate directly into the Python's lunch formula.\footnote{There are some subtleties with this motivation in the presence of multiple lunches, which need not lie on the same Cauchy slice. This will be discussed in more detail in~\cite{EngPenTA}.}

For an evaporating black hole the locations of $X_B^{\rm bulge}$ and $X_R^{\rm bulge}$ were studied in \cite{Brown:2019rox}.  At early times the outermost and entanglement wedges of $R$ coincide and so there is no bulge for $R$ but there is a bulge for $B$, while at late times there is no bulge for $B$ but there is one for $R$.   In either case there are two bulge candidates, which are shown in figure \ref{bulges1}.  The generalized entropies of these can be estimated by looking at the associated cuts through the quantum circuit of our dynamical model; roughly speaking they are given by $\log|B_0|+S[\chot_R]$ for the bulge with one connected component and $2\log |B|+S[\chot_R]$ for the bulge with two connected components (see \cite{Brown:2019rox} for more details on this calculation).  The correct bulge will be whichever of these has smaller generalized entropy since we are minimizing over foliations.  The key point however is that at early times we have $S_{gen}[X_B^{\rm outer}]\approx \log |B|$, which soon becomes small compared to the generalized entropy of either bulge, and so prior to the Page time \eqref{python} predicts that it will be exponentially hard (in $\log |B|$) to access information about the black hole interior using observables on $B$.  Similarly after the Page time we have $S_{gen}[X_R^{\rm outer}]\approx S[\chot_R]$, so \eqref{python} again predicts that it will be exponentially hard to access information about the black hole interior using observables on $R$.  This is just what we concluded in the previous subsection, and indeed by using the relationship between circuit cuts and bulges shown in figure \ref{bulges1} the authors of \cite{Brown:2019rox} were able to match the complexity predicted by \eqref{python} to that which can be obtained by the multi-state Grover procedure described in the previous subsection (this was essentially their motivation for proposing \eqref{python} in the first place).  

We can also consider the reconstruction/distinguishability of observables in the entanglement wedge of $BR$.  Here the relevant calculations were done in \cite{Engelhardt:2021qjs}, so we will again be brief.  The candidate bulge for $BR$ (present only in the back-reacted geometry) lies just inside of the new outermost QES for $BR$ which is shown in figure \ref{newouterfig}.  Their generalized entropies are comparable, with
\be
S_{gen}[X_{BR}^{\rm bulge}]-S_{gen}[X_{BR}^{\rm outer}]\sim N_r, 
\ee
where $N_r$ is the number of modes in $r$ which are affected by $W_{\ell r}$.  \eqref{python} thus precisely gives the same complexity scaling we found from the multi-state Grover procedure in the previous subsection.

\section{Stable AdS black holes and typicality}\label{adsonesec}
\bfig
\includegraphics[height=6cm]{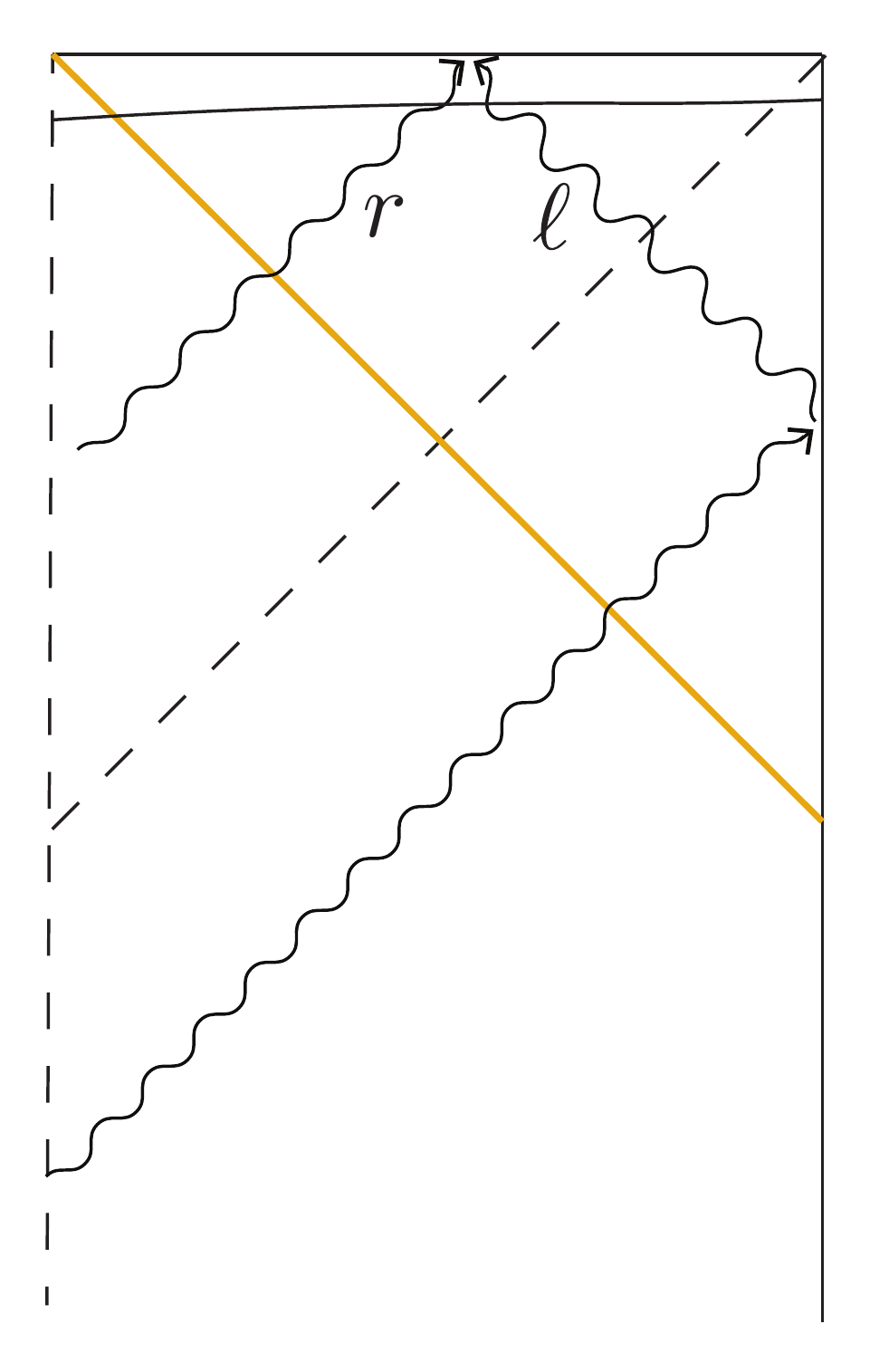}
\caption{Interior modes in a one-sided AdS collapse.  The left boundary of the diagram is the origin of polar coordinates, while the right boundary is the AdS boundary.  The collapsing shell is shown in orange.  Neighboring left and right movers are entangled on the nice slice, since the left movers arise from Hawking quanta which are reflected back into the black hole as shown with the dashed lines.}\label{adsonesidefig}
\efig
So far we have focused on modeling evaporating black holes in flat space.  Large black holes in AdS are also important to understand, coming with their own set of puzzles.  We illustrate the formation and evolution of such a black hole in figure \ref{adsonesidefig}: there is an initial collapse, followed by a long period where the length of the interior geometry (for example on a maximal volume slice) grows with time.  The left and right movers on this slice are entangled as in figure \ref{adsonesidefig}.  Eventually the number of degrees of freedom on this slice exceeds the black hole entropy, leading to a crisis similar to that which happens at the Page time for evaporating black holes.  Inspired by our treatment of the evaporating case, it is natural to expect that this crisis is resolved by allowing the holographic map from the effective description on the nice slice to the fundamental description -- which here is just given by the dual CFT Hilbert space -- to be non-isometric.  In this section we briefly explain how this works, jumping straight to constructing a dynamical model of the problem.  

\bfig
\includegraphics[height=9cm]{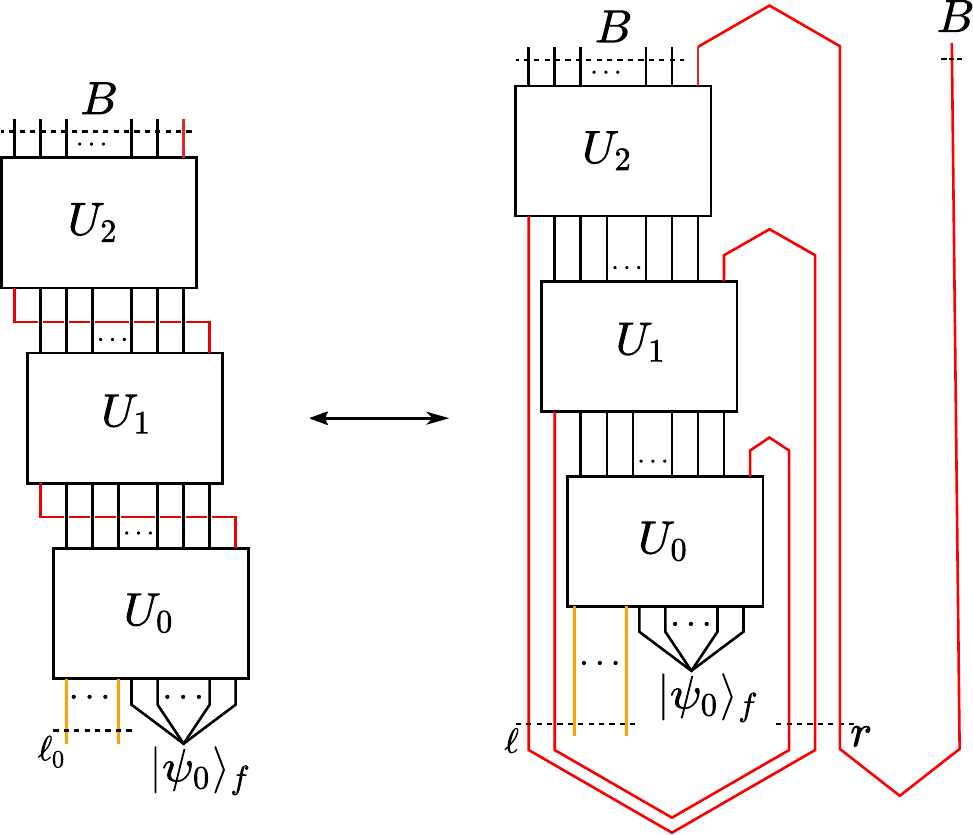}
\caption{A few time steps in a dynamical model of a one-sided AdS black hole formed from a fast collapse.  On the left we have the fundamental CFT dynamics. The initial state of the infalling shell is encoded into a few adjustable qudits $\ell_0$ and then a fixed state $|\psi_0\ran_f$.  It evolves unitarily, with the radiated qudits, shown in red, being reabsorbed by the black hole after a time of order $\ell_{ads}$ (each time step is of order $\ell_{ads}$).  On the right we re-interpret this evolution as defining a non-isometric map from the interior modes $\ell$ and $r$ to the CFT degrees of freedom $B$, with the non-isometry arising from the post-selection on these modes just as in figure \ref{dynevapfig}.  We have allowed the unitary at each time step to be different, but to model the usual situation with no sources in the dual CFT it is natural to take all the $U_i$ to be the same.}\label{dynadsfig}
\efig
Our basic dynamical model is shown in the left diagram of figure \ref{dynadsfig}: we have a set $B$ of fundamental degrees of freedom, which is essentially the dual CFT (up to removing some irrelevant short-distance degrees of freedom), which is evolving under unitary time evolution.  At each moment in time there is a small subset of degrees of freedom, shown in red, which is actually outside of the horizon (or more carefully is in the outermost wedge of $B$).  These are easy to reconstruct in the dual CFT, e.g. by using a variant of the HKLL procedure.  With each time step these exterior degrees of freedom reflect off of the boundary and fall back into the black hole, where they are added at the left side of the black hole degrees of freedom.  The key observation is that we can introduce the interior degrees of freedom $\ell$ and $r$ by bending these lines down to the bottom of the figure using a post-selection and then bending them back up by inserting a maximally entangled state, as in the right diagram.  This should be compared to figure \ref{dynevapfig} above.  We can interpret the bent version of the circuit as defining a non-isometric map $V$ from an entangled set of left and right moving interior modes $\ell$ and $r$, together with the exterior modes, into the fundamental degrees of freedom $B$.  As the evolution continues more and more $\ell$ and $r$ modes are created in this way, corresponding to the growth in length of the nice slice. Eventually there are so many interior modes that $V$ becomes highly non-isometric, which is a kind of Page time for this black hole.  This non-isometric map continues however to preserve the physics of all sub-exponential observables, at least until we run the evolution for a time which is exponential in $\log |B|$. 

As one application of this model, we note that if we consider an effective description state where $r$ and $\ell$ are not maximally entangled with each other, and instead have some appreciable entanglement with an auxiliary system, then the model begins to resemble our dynamical model for the evaporating black hole.  In particular the empty set is no longer homologous to $B$, and there is a new potentially-minimal QES corresponding to the cut across $B$ in figure \ref{dynadsfig}. This exactly matches the calculation of~\cite{Engelhardt:2021qjs}, which found that backreaction resulting from disentangling the outgoing Hawking modes results in a new QES, shown in the left diagram of figure \ref{newouterfig} above, which can be minimal for a sufficiently large code subspace.

As another application, we can use this model to illustrate the difficulty that arises if we try to use a non-isometric code to construct an interior for a ``typical'' black hole microstate.  In order to prepare a typical state, we need to eliminate the ``fixed'' degrees of freedom $f$ in figure \ref{dynadsfig} and include all infalling degrees of freedom in $\ell_0$.  But once we do this there is an immediate problem: defining the full fundamental evolution to be $U$, we can now easily construct null states of the form
\be\label{nullstate}
|\phi^{\rm null}\ran=U^\dagger |\phi\ran_B|0^\perp\ran_P,
\ee
where $|\phi\ran_B$ is arbitrary and $|0^\perp\ran_P$ is any state which is orthogonal to the maximally-entangled state which is post-selected on in the right diagram of figure \ref{dynadsfig}.  If we take the $U_i$ to be $k$-designs, or to be given by the RPI model of section \ref{dynamicsec}, then these null states do \textit{not} have exponential complexity.  In other words if we attempt to introduce an effective description as in figure \ref{dynadsfig} then it breaks down even for sub-exponential observables.  On the other hand this problem disappears rapidly once we restore $f$.  Once $|f|>1$ we can't use the state \ref{nullstate} as input to $V$, but we can still get into trouble by considering the (approximately normalized) input 
\be
|\phi^{\rm small}\ran=\sqrt{|f|}\lan \psi_0|_f U^\dagger|\phi\ran_B|0^\perp\ran_P.
\ee
A short calculation (for simplicity treating the dynamical unitary $U$ as a single random matrix as in \eqref{Vdef}) shows that
\be
\int dU\lan \phi^{\rm small}|V^\dagger V|\phi^{\rm small}\ran=\frac{|P|^2|B|^2}{|P|^2|B|^2-1}\left(1-\frac{1}{|f|}\right),
\ee
so the norm of the state $|\phi^{\rm small}\ran$ is likely to be substantially modified by acting with $V$ unless $|f|$ is exponential in $\log |B|$.  To prepare the state $|\phi^{\rm small}\ran$ starting from $|\phi^{\rm null}\ran$ however we need to use the Grover amplification procedure described in appendix \ref{complexityapp} and used in section \ref{coarsesec}, whose complexity will be exponential $\log |f|$.  Increasing $f$ therefore makes these states both less problematic and more complex.  This is dual to an observation made in \cite{Susskind:2015toa}: even if we assume that typical black holes have singular horizons, as was argued in \cite{Marolf:2013dba} (see also \cite{Harlow:2021dfp} for some recent support), we can restore a smooth horizon for an exponentially long time simply by dropping in a few degrees of freedom and then waiting for a scrambling time.

\section{Multiple black holes}\label{multisec}
One of the most useful laboratories for quantum gravity in AdS has been the two-boundary Schwarzschild wormhole, which is dual to the thermofield double state of the CFT on two disconnected spheres \cite{Maldacena:2001kr}.  This  gives the canonical example of a remarkable phenomenon where black holes which are sufficiently entangled with each other can share an interior region \cite{VanRaamsdonk:2010pw,Maldacena:2013xja}.\footnote{The converse of this statement is not always true; for examples of states with a large amount of entanglement but no shared interior see \cite{Engelhardt:2022qts}.}  We now briefly discuss how this phenomenon arises in the context of our models.  

The obvious guess for how to generalize our static model from section \ref{modelsec} to two black holes is to just define a tensor product holographic map $V\otimes V:\sH_{\ell_1}\otimes \sH_{r_1}\otimes\sH_{\ell_2}\otimes \sH_{r_2}\to \sH_{B_1}\otimes \sH_{B_2}$.  This however cannot be the whole story, since the effective description states in the domain of this map all describe two disconnected spacetimes.  We therefore propose that there is an additional orthogonal sector in the effective description of the form $\sH_{\ell_c}\otimes \sH_{r_c}$ which describes the modes on a nice slice of wormhole geometry connecting the two black holes.  Moreover we take this map to have the form
\be
V_c\equiv\sqrt{|P_c|}\lan 0|_{P_c}U_c|\psi_0\ran_{f_c},
\ee
where $U_c$ is a Haar-random unitary map from $\sH_{\ell_c}\otimes \sH_{r_c}\otimes \sH_{f_c}$ to $\sH_{B_1}\otimes \sH_{B_2}\otimes \sH_{P_c}$.  We then take the full holographic map to act as $V\otimes V$ on disconnected states in the effective description and $V_c$ on connected states in the effective description.  The key point is then that the overlap between the image of a connected state and the image of a disconnected state is likely to be exponentially small:
\begin{align}\nonumber
\int dUdU_c\Big|\lan \phi|_{\ell_c r_c}V_c^\dagger (V\otimes V)|\psi\ran_{\ell_1 r_1,\ell_2 r_2}\Big|^2&=\frac{|P|^2|B|^2}{|P|^2|B|^2-1}\frac{1}{|B|^2}\Big[\left(1-\frac{1}{|P||B|^2}\right)+\frac{\lan \psi|\Sigma_{12}|\psi\ran}{|B|}\left(1-\frac{1}{|P|}\right)\Big]\\
&\leq \frac{4}{|B|^2},
\end{align}
and thus (by Jensen's inquality) 
\be
\int dUdU_c\Big|\lan \phi|_{\ell_c r_c}V_c^\dagger (V\otimes V)|\psi\ran_{\ell_1 r_1,\ell_2 r_2}\Big|\leq\frac{2}{|B|}.
\ee
Here we have taken $|B_1|=|B_2|=|B|$, and $\Sigma_{12}$ is the swap operator that exchanges $\ell_1 r_1$ with $\ell_2 r_2$.  Thus we see that, although connected and disconnected geometries are strictly orthogonal in the effective description, they map to states in the fundamental description where this orthogonality is only approximate.  This is the mechanism that enables the emergence of the connected interior: by taking a superposition of a sufficiently large number of disconnected states we can find ourselves in a connected state.  

In the previous paragraph we showed that any particular disconnected state $|\psi\ran$ and connected state $|\phi\ran$ are likely to have images with an exponentially small overlap.  A natural next step would be to use a measure concentration argument similar to that which went into theorem \ref{setconcthm} to conclude that this approximate orthogonality holds for \textit{all} sub-exponential states. We believe this to be the case, but have not succeeded in constructing the argument.  The reason is that the conclusion is clearly false if we replace $V\otimes V$ by $V\otimes V^*$ (the latter does not preserve the norm of the maximally-entangled state, which is surely sub-exponential), and so the requisite argument must be smart enough to distinguish between $V\otimes V$ and $V\otimes V^*$.

It would be interesting to apply the dynamical evolution of the previous section to the two-sided black holes we consider here.  We expect that the analysis in \cite{Zhao:2020gxq,Haehl:2021prg,Haehl:2022uop} of meetings in the black hole interior could be re-interpreted in this context.

\section{Discussion}\label{discussion}
In this final section we step back and assess our models in the broader context of what was previously known about the black hole information problem.

\subsection{Relation to previous work}
We first discuss how our models relate to previous ideas for understanding the black hole information problem.  Given the size of the literature any such discussion will necessarily be incomplete, so we focus on those proposals about which we have something new to say.

\subsubsection*{Euclidean gravity} 
One of our most powerful tools to date for learning about the quantum properties of black holes is the Euclidean gravity path integral. It provides a straightforward way of doing many calculations in black hole thermodynamics knowing only the low-energy action, which so far have essentially always agreed (within their regime of validity) with those results that can be obtained on more solid ground e.g. from string theory (see for example \cite{Strominger:1996sh,Strominger:1997eq,Witten:1998qj}).  In particular one can compute the entropy and temperature of a black hole \cite{Gibbons:1976ue}, and more generally one can ``derive'' the QES formula of \cite{Engelhardt:2014gca} for von Neumann entropies of subsets of the fundamental degrees of freedom using Euclidean methods \cite{Lewkowycz:2013nqa,Faulkner:2013ana}, \cite{Penington:2019kki,Almheiri:2019qdq}.  On the other hand the justification for the Euclidean gravity path integral has always been somewhat murky.  In particular, as explained in e.g. \cite{Harlow:2020bee}, the saddle points which lead to the most powerful results have topologies which are not consistent with canonical quantization, and thus lead to answers which cannot be interpreted as arising from a quantum theory which treats the integration variables as the true degrees of freedom.  In particular one cannot use Euclidean quantum gravity to construct a Hilbert space of black hole microstates whose dimension is $e^{\mathrm{Area}/4}$. This lack of manifest equivalence with quantum mechanics leads to a number of puzzles, which some authors have tried to interpret in terms of an average over coupling constants and/or a sum over baby universes \cite{Coleman:1988cy,Giddings:1988cx,Saad:2018bqo,Marolf:2020xie}.  Moreover for $d>2$ the Euclidean path integral is somewhat ambiguous due to the non-renormalizability of gravity (with matter), and it is not always clear when this ambiguity is important.    

The role of Euclidean gravity in our models is easy to state: it has none.  We compute all entropies directly from von Neumann's definition $S(\rho)=-\rho \log \rho$ in the fundamental description, in all cases finding results that are consistent with unitarity.  Our results do agree with the QES formula, and thus with the results of Euclidean quantum gravity, but they do not rely on it.  We view this as an advance: we have replaced a mysterious oracle with calculations which are manifestly quantum mechanical.  This does not mean that the oracle now has no value: our models pay a steep price in discarding locality and diffeomorphism/Lorentz invariance, while Euclidean gravity easily incorporates both.  In particular we can only say that $\log |B|$ is analogous to $\mathrm{Area}/4$, we cannot derive it.   It is perhaps useful to also observe more directly how our models reproduce Euclidean calculations.  A first observation is that the modified inner product of \cite{Marolf:2020xie} in the effective description can be interpreted in our language as the average over $U$ of $\lan \psi|V^\dagger V\otimes I_{LR}|\phi\ran$.  The map $V$ however contains more information than is contained in $V^\dagger V$, for example we can multiply $V$ by an arbitrary unitary on the left without changing $V^\dagger V$.  We can also observe that some of the terms in e.g. \eqref{nrenyiexp} and \eqref{s2terms} are clearly playing the same role as ``replica wormholes'' do in the Euclidean calculations of these quantities, for example the second term on the right-hand side of \eqref{s2terms} is precisely the contribution of the ``double trumpet'' wormhole.  From this point of view our models can be viewed as giving additional justification for the Euclidean prescription, since the ``topologies'' one sums over there arise here on general grounds in genuine quantum calculations. 

\subsubsection*{Black hole final state proposal}
\bfig
\includegraphics[height=6cm]{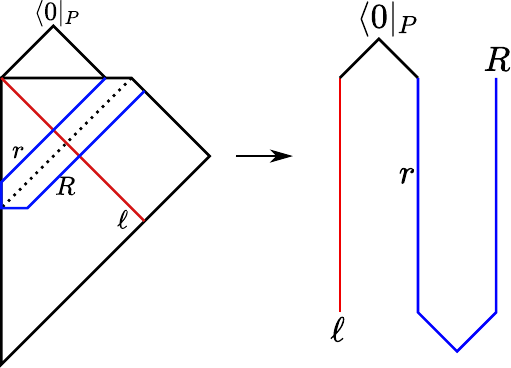}
\caption{The black hole final state proposal: left-moving and right-moving interior modes are jointly post-selected on at the singularity, leading to a circuit which unitarily maps the state of the collapsing matter to the Hawking radiation.  Note the similarity to figure \ref{Vdeffig} and especially figure \ref{dynevapfig}.}\label{finalstatefig}
\efig
In \cite{Horowitz:2003he} Horowitz and Maldacena gave a novel proposal for resolving the black hole information problem by using the post-selected quantum mechanics of \cite{aharonov1964time}.  This is a modification of quantum mechanics where a post-selection onto some ``final state'' is included in the calculations of all measurement probabilities.  Given an initial state $\rho_i$, a final state $\rho_f$, and a sequence $\Pi_{a_1}, \Pi_{a_2},\ldots$ of projective measurements, the proposal of \cite{aharonov1964time} is that the probability of obtaining a sequence of measurement outcomes $a_1,a_2,\ldots$ is given by
\be\label{finalprob}
P(a_1,a_2,\ldots)=\frac{\tr\left(\rho_f \Pi_{a_n}\ldots \Pi_{a_1}\rho_i \Pi_{a_1}\ldots \Pi_{a_n}\right)}{\sum_{b_1,b_2,\ldots}\tr\left(\rho_f \Pi_{b_n}\ldots \Pi_{b_1}\rho_i \Pi_{b_1}\ldots \Pi_{b_n}\right)}.
\ee
This formula has several alarming features when $\rho_f$ is not maximally mixed, the most severe of which is that the probabilities for the outcomes of earlier measurements usually depend on which measurements happen later: causality is violated! This may seem like a heavy price to pay, but the observation of \cite{Horowitz:2003he} was that if we can tolerate it then we can also account for unitary black hole evaporation purely within gravitational effective field theory.  The basic idea is shown in figure \ref{finalstatefig}: by an appropriate post-selection on left-moving and right-moving interior modes we can transfer the information in the former out to the Hawking radiation. See \cite{Gottesman:2003up,Lloyd:2013bza,Bousso:2013uka} for some further discussion of the final state proposal.

Figure \ref{finalstatefig} is very similar to figures \ref{Vdeffig} and \ref{dynevapfig} above, and it in particular it becomes almost identical to figure \ref{dynevapfig} if we wait until the black hole has completely evaporated to ``bend around'' the radiation legs as in figure \ref{dynevapfig}.  The mathematical mechanism for information conservation is thus the same in our models as it is for the final state proposal: post-selection on the interior modes teleports the infalling information out into the Hawking radiation.  On the other hand the physical interpretation is quite different: 
\bi
\item[(1)] In the final state proposal the post-selection comes from a modification of quantum mechanics in the effective description, and no fundamental description is ever introduced.  In our models by contrast, the post-selection arises in the holographic map from the effective description to the fundamental description.
\item[(2)] In the final-state proposal the post-selection is viewed as happening ``at the singularity'', after everything else takes place in the interior.  In our dynamical models by contrast we have a holographic map $V_t$ at each intermediate time, and each of these involves post-selection.
\item[(3)] The measurement theory of the final-state proposal is based on \eqref{finalprob} and includes manifest acausality, while the measurement theory for us in the fundamental description obeys the usual rules of quantum mechanics and in particular is perfectly causal.  This is true even once we translate interior measurements to the fundamental description as in section \ref{measurementsec}: for measurements at time $t$ we do the translation using the holographic map $V_t$, and then do conventional measurements without any further post-selection.  The approximate invertibility of $V_t$ on sub-exponential states ensures a unique interpretation for each (sub-exponential) operation which is performed.
\ei

\subsubsection*{Papadodimas-Raju proposal} 
Another well-known approach to the information problem is the ``state-dependent interior operator'' proposal of Papadodimas and Raju \cite{Papadodimas:2012aq,Papadodimas:2013jku,Papadodimas:2013wnh}.  This proposal has primarily been developed for non-evaporating black holes in $AdS$, so we'll describe it in that context.  The idea is that within the Hilbert space of black hole microstates there is a special set of ``equilibrium states'', which is defined by requiring that in these states the correlation functions of $O(1)$ numbers of low-trace operators in the dual CFT look thermal up to $O(1/N)$ corrections. On top of each equilibrium state one then builds a ``small Hilbert space'' by taking the span of all states obtained by acting on the equilibrium state with an $O(1)$ number of low-trace operators.  One then builds a set of ``mirror operators'' which act within this Hilbert space in the same way as operators on the left exterior of the Hartle-Hawking state would up to $O(1/N)$ corrections.  This then allows the quantum mechanics of the two-sided black hole (including its interior) to be ``simulated'' in the Hilbert space of a one-sided black hole, at least up to $O(1/N)$ corrections.  The mirror operators constructed in this way clearly depend on the choice of equilibrium state, and since the set of equilibrium states does not form a subspace (and in fact likely spans the Hilbert space), these operators are necessarily non-linear.  This proposal was discussed more in \cite{Harlow:2014yoa,Marolf:2015dia}, where in particular it was emphasized that it suffers from an ambiguity arising from the fact that there can be mirror operators which map equilibrium states to other equilibrium states, causing a nontrivial overlap between their small Hilbert spaces, which means that the same rules can give different physical interpretations to the same state in Hilbert space.   

In our models the state-specific reconstructions of interior operators are similarly non-linear.  There is a key difference however: in the PR proposal one begins with a generic state in the fundamental description (the dual CFT) and then tries to construct for it an EFT interpretation, while we begin with the full set of states in the effective description and then construct a linear holographic map $V\otimes I_{LR}$ which maps all of them into the fundamental description.  For us the map $V \otimes I_{LR}$ is fixed once and for all: it is \textit{not} state-dependent.  Non-linearity appears only when we try to invert this map on the set of sub-exponential states.  This is a much tighter structure than the PR proposal, and this leads to several advantages.  Perhaps the most important is that from lemma \ref{invlem}, we learned that our holographic map $V\otimes I_{LR}$ is approximately invertible on sub-exponential states.  Thus there is no state in the fundamental description which is the image of a sub-exponential state but does not have a unique physical interpretation.  We therefore do not have any ambiguity analogous to that of the PR proposal.  Another advantage for us is that equilibrium is not a defining part of our rules, and so we are able to consider non-equilibrium situations such as an evaporating black hole.  Further advantages for us are that the PR proposal works only for observables of $O(1)$ complexity and to leading order in perturbation theory, while our models work for observables of sub-exponential complexity and to all orders in perturbation theory.  Perhaps the closest thing in our models to the PR proposal the discussion of section \ref{sec:subspacerecon}, where we introduced state-independent reconstructions that can work on a large subspace somewhat analogous to the ``small Hilbert space'' of PR. The resemblance however is fairly superficial: there is no role for ``equilibrium states'' in our discussion, and no matter what subspace we choose the map $V\otimes I_{LR}$ remains the same.

\subsubsection*{``Increasing complexity'' criterion for horizon smoothness}
In \cite{Susskind:2015toa} it was proposed that states in the fundamental description whose complexity is increasing should have smooth horizons, while those where it is not may not.  We have not studied this conjecture in detail in our models, but it is quite plausible in them: sub-exponential states are naturally those whose complexity is likely to increase, and as discussed in section \ref{adsonesec} our approach for constructing the interior fails if we try to allow generic states (which have maximal complexity) in the effective description.

\subsubsection*{Previous attempts at an encoded interior}
Here we briefly comment on a few earlier attempts to understand the holographic encoding of the black hole interior.

The first attempt at using quantum error correction to study the black hole interior was given in \cite{Verlinde:2012cy}.  There they showed that interior reconstructions could be given provided that the black hole state was only supported on a subspace of its microstates which was much smaller than the full dimensionality.  This can be thought of as a special case of our discussion in section \ref{sec:subspacerecon}, where their subspace is the image of an appropriately chosen $\sH_C$ through the approximate isometry $\hat{V}$.  Handling old black holes using a coding framework however requires a non-isometric code. 

In \cite{Kourkoulou:2017zaj} a special set of one-sided black hole microstates of the SYK model was proposed, which is obtained by acting with a complete projection $\lan B_s|$ on the left side of the thermofield double. Here $s$ indicates a string of $N/2$ bits, where $N$ is the number of SYK Majorana fermions, and $|B_s\ran$ are a set of eigenstates of the spins obtained by pairing up the Majorana fermions. The geometric description of these states is as a single-sided black hole with an ``end of the world brane'' inside which is labeled by $s$.  In  \cite{Almheiri:2018xdw} the question of reconstruction of interior operators on these states was considered.  In particular it was argued that for each interior operator $O$ there exists a reconstruction $O_R^s$ such that
\be
O_R^s\lan B_s|\beta\ran\approx \lan B_s|O|\beta\ran,
\ee
where $|\beta\ran$ is the thermofield-double state and the $s$ on $O_R^s$ indicates that the reconstruction depends on $s$ and is thus state-specific.  If we try to interpret $\lan B_s|$ as some kind of non-isometric holographic map then this is fairly similar to the reconstruction we discuss here, except that $O|\beta\ran$ describes the action of the observable on the thermofield-double state instead of an end of the world brane state.  This however is more of a defect of notation rather than of principle, since in \cite{Almheiri:2018xdw} only the fundamental description was considered.  If we re-interpret $\lan B_s|O|\beta\ran$ as $VO|\beta,s\ran$ where $|\beta,s\ran$ is a state of the end of the world brane theory in the effective description, and then define $V|\beta,s\ran=\lan B_s|\beta\ran$, then this reconstruction becomes similar to ours:
\be
O_R^sV|\beta,s\ran\approx VO|\beta,s\ran.
\ee
As emphasized in \cite{Almheiri:2018xdw}, the states $V|\beta,s\ran$ are overcomplete.  In other words, the encoding map $V$ cannot be an isometry.  

In \cite{Akers:2019nfi} a simple model using multi-boundary AdS wormholes was proposed by some of us as an analogue of an evaporating black hole. In particular the Page curve for this model follows from the simpler RT formula rather than the full QES formula, essentially because the entanglement of the Hawking pairs was replaced by geometry.  This allows for easy calculations and illustrates the role of islands, but geometrizing all of the bulk entanglement hides the non-isometric nature of the code which is needed to reconstruct interior observables. 

In \cite{Kim:2020cds} an interesting connection was pointed out between pseudorandomness and quantum error correction.  A quantum state $\rho$ is said to be pseudorandom if all sub-exponential observables have the same expectation values in $\rho$ as they would in the maximally-mixed state up to exponentially small errors.  A slight generalization of what they proved is the following:
\begin{thm} (Kim/Preskill/Tang) \label{KPTthm}
Let $V:\sH_A\to \sH_B\otimes \sH_C$ be an isometry, and define
\be
|\Psi\ran_{BC\ol{A}}=(V\otimes I_{\ol{A}})|\mathrm{MAX}\ran_{A,\ol{A}}.
\ee
Furthermore let $\Psi_{C\ol{A}}$ be ``pseudo-product'' in the sense that there is an $\alpha>0$ and a product state $\sigma_{C\ol{A}}=\sigma_C\otimes \sigma_{\ol{A}}$ such that for any sub-exponential (in $\log |B|$) observable $O_{C\ol{A}}$ we have
\be
\Big|\tr\left[O_{C\ol{A}}\left(\Psi_{C\ol{A}}-\sigma_{C}\otimes \sigma_{\ol{A}}\right)\right]\Big|\leq |B|^{-\alpha}.
\ee
Then for any sub-exponential (in $\log |B|$) unitary $U_{OC}$,
if we define the state
\be
\wt{\Psi}_{OBC\ol{A}}\equiv (U_{OC}\otimes I_{B\ol{A}})(\Psi_{BC\ol{A}}\otimes |0\ran\lan 0|_O)(U_{OC}^\dagger\otimes I_{B\ol{A}})
\ee
then we have the decoupling bound
\be
||\wt{\Psi}_{O\ol{A}}-\wt{\Psi}_{O}\otimes\wt{\Psi}_{\ol{A}}||_1\leq 3|O||A||B|^{-\alpha}.
\ee
\end{thm}
What this theorem says (assuming $|B|\gg |O|, |A|$) is that as long as there is no sub-exponentially accessible correlation between $C$ and $\ol{A}$, then no adversary who has access only to $C$ and can only implement sub-exponential operations can learn anything about the initial state of $A$. By standard decoupling arguments this means that after the adversary has acted with $U$ this information is still really in $BC$.\footnote{If the dimensionality of $A$ is large then one has to be careful in applying the decoupling criterion. If we define a channel $\mathcal{C}$ from states on $A$ to states on $BC$ which works by acting with $(U_{OC}\otimes I_B)(V\otimes I_{O})|0\ran_O$ and then tracing out $O$, then the relevant decoupling theorem tells us that 
\be
\min_{\mathcal{R}}\max_{\rho_A}\left(1-F(\mathcal{R}\cdot \mathcal{C}(\rho_A),\rho_A)\right)\leq \frac{1}{2}|A| \, \Big|\Big|\wt{\Psi}_{O\ol{A}}-\wt{\Psi}_{O}\otimes\wt{\Psi}_{\ol{A}}\Big|\Big|_1,
\ee
where $F$ is the fidelity and $\mathcal{R}$ are channels from $BC$ to $A$.  The extra factor of $|A|$ arises because decoupling of the encoded  maximally-mixed state only implies correctability on average, while we are really interested in worst-case correctability.  Thus for theorem \ref{KPTthm} to imply recovery with fidelity at least $1-\epsilon$, we need $\frac{3}{2}|O||A|^2|B|^{-\alpha}\leq \epsilon$.}  
One way to apply this theorem to our setup is by defining $B\to B$, $C\to R$, $A\to \ell$, and $\ol{A}\to L$.  The theorem therefore applies if the encoded state $\Psi_{LR}$ is ``pseudo-product''.  As we've discussed above (see \eqref{Phidef}), in the static model with a Haar random unitary all sub-exponential observables have expectation values which are exponentially close to what they would be in the coarse-grained state
\be
\Phi_{LBR}=\chin_L\otimes \frac{I_B}{|B|}\otimes \chot_R.
\ee
This certainly factorizes on $R$ and $L$, and so $\Psi_{LR}$ is indeed pseudo-product.  In our more realistic models the expectation values of sub-exponential observables on $LBR$ will \textit{not} agree in $\Psi_{LBR}$ and $\Phi_{LBR}$, since if we have access to both $B$ and $R$ we can evolve the dynamics backward and learn about how the black hole was made.  If we have access only to $\Psi_{LR}$ however, then by similar arguments to those at the end of subsection \ref{simpentsec} we expect all sub-exponential observables to agree in $\Phi_{LR}$ and $\Psi_{LR}$, and thus $\Psi_{LR}$ to be pseudo-product.  The theorem then tells us that any sub-exponentially bounded observer who has access only to $R$ cannot learn anything about the black hole microstate, which is a quite reasonable conclusion.  On the other hand the reconstructions of interior observables, called ``ghost operators'' in \cite{Kim:2020cds}, will in general depend on the choice of unitary $U_{OR}$ applied by the observer, which is another illustration of the state-dependence of interior reconstruction. 

Recently \cite{Balasubramanian:2022fiy} included some general discussion of the idea that there could be a non-isometric code mapping the black hole interior into the fundamental description.  Technically their main result was a study of the robust nature of encoding the interior into the radiation emphasized by \cite{Kim:2020cds} in the context of the  ``west coast'' toy model of the evaporating black hole, featuring JT gravity with a (flavored) end of the world brane~\cite{Penington:2019kki}.  More concretely, using Euclidean quantum gravity they showed that the conclusions of theorem \ref{KPTthm} hold in that model with the subsystem identifications suggested in the previous paragraph.

\subsubsection*{Null states as gauge equivalences}

A common idea, advocated for example in~\cite{Marolf:2020xie, anous2020density, langhoff2020ensemble, blommaert2022alpha} is that
null states in the effective description arise from some kind of gauge symmetry.  If so, then it is natural to define a physical Hilbert space modulo addition by a null state and this space is then in one-to-one correspondence with the Hilbert space of the fundamental description.  In this interpretation, whenever $V|\psi\ran=V|\phi\ran$ then $|\psi\ran$ and $|\phi\ran$ are equivalent descriptions of the same state. The question then becomes understanding how to predict the experience of an observer in the effective description when there exist many distinct, but equivalent, effective
descriptions of the same fundamental state.

We do not require, or necessarily advocate, this picture.  Our paper is potentially consistent with it, but only with the understanding that the gauge symmetry must not include any sub-exponential transformations in the effective description.  Indeed we saw from lemma \ref{invlem} that every image of a sub-exponential state has a unique physical interpretation, and so in each equivalence class of the putative gauge symmetry there can be at most one sub-exponential state.  For such classes, we claim that the sub-exponential description is the \textit{only} correct description of what is going on.  It is therefore not so clear what value is added by such a gauge interpretation for the null states of $V$, although one may exist.

\subsubsection*{Related ideas in information theory}
The question of how to pack more logical qubits into fewer physical qubits has previously been discussed in the quantum information literature. In \cite{chao2017overlapping}, the authors studied the questions of how many ``overlapping qubits'' can be fit into a finite-dimensional Hilbert space. 
Imagine we have access to $n$ sets of Pauli operators $\{X_i, Z_i\}$ acting on a Hilbert space $\mathcal{H}$, each satisfying $\lVert \{X_i, Z_i \} \rVert = 0$, $\tr[X_i] = \tr[Z_i] = 0$ and $X_i^2 = Z_i^2 = 1$. We also require that each set approximately (but not necessarily exactly) commutes with all the other sets of operators; in other words, we have $\lVert [S_i, T_j] \rVert \le \varepsilon$ for $i \neq j$ and $S,T \in \{X,Z\}$. How small can $\dim \mathcal{H}$ be, as a function of $\varepsilon$ and $n$? The answer, found in \cite{chao2017overlapping}, is that for any fixed $\varepsilon > 0$, $\dim \mathcal{H}$ only needs to grow polynomially with $n$; the small overlaps allow exponentially more qubits to fit in a fixed Hilbert space $\mathcal{H}$ than would otherwise be possible. The spirit of this result is clearly somewhat similar to those in this paper, and indeed their proof strategy, based on the Johnson-Lindenstrauss lemma, is closely related to techniques that we use here. However the technical statements are different, and in particular the constructions in \cite{chao2017overlapping} do not (and cannot) feature a nonisometric encoding of states in a larger Hilbert space.\footnote{There is, for example, no state in the Hilbert space $\mathcal{H}$ that is simultaneously an (approximate) +1 eigenstate of all the operators $Z_i$.}
The authors of \cite{chao2017overlapping} also prove some no-go theorems about compressing qubits that are closely related to our theorem \ref{nogo5} -- see especially Corollary 4.4 -- although again the precise constructions being ruled out differ from the ones we consider here.

An earlier closely related idea is quantum fingerprinting \cite{buhrman2001quantum}, where  the authors construct a non-isometric code that maps each of $2^n$ orthogonal states to its “quantum fingerprint,” a state in the Hilbert space of $\log(n)$ qubits, such that the inner-product of the fingerprints approximately equals that of the original states. They then show that these fingerprints are helpful for an “equality testing” quantum task. In this task, Alice and Bob each want to communicate one of $2^n$ possible messages to a referee, by encoding that message into as few qubits as possible. The referee should accept with probability at least $1-\delta$ if Alice and Bob sent the same message and accept with probability less than $\delta$ if they did not. Classically, Alice and Bob can accomplish this task for exponentially small $\delta$ by sending just $\log(n)$ classical bits -- so long as they already share a random key. Quantum mechanically, as our non-isometric codes demonstrate, this task is possible without any pre-shared randomness. Alice and Bob can encode their $2^n$ states using \eqref{Vdef}, send the output to referee, and the referee can test the overlap of the state efficiently. It is important to note that the nonisometric codes used in this paper are not identical to those in \cite{buhrman2001quantum}: the postselection in our model cannot be implemented efficiently, whereas a main point of \cite{buhrman2001quantum} is to construct a scheme for efficiently producing fingerprints, thereby providing an efficient protocol for the equality testing task.
On the flip side, their scheme would not work for our purposes, because it does not guarantee that all sub-exponential input states (rather than merely a single basis of states) have their inner-product preserved.

\subsection{Frequently asked questions}
In this subsection we attempt to give concise answers to questions which we expect many readers will have.   
\bi
\item \textit{Is the black hole information problem solved?}  In the introduction we identified three essential ingredients of Hawking's paradox:
\bi
\item[(1)] A finite black hole entropy with a state-counting interpretation.
\item[(2)] A unitary black hole S-matrix.
\item[(3)] A black hole interior which is described to a good approximation by gravitational effective field theory, including the entanglement between the outgoing interior and exterior modes $r$ and $R$.
\ei
We have seen that analogues of all three of these exist in our dynamical model.  Thus, as a point of principle, we have seen that it is possible for them to co-exist, at least to the extent that we have modeled them here.  We can therefore say that within our models the black hole information problem is resolved.  On the other hand there are many aspects of gravity which our models have not captured.  These include Lorentz invariance, diffeomorphism invariance, interactions between left- and right- movers, and even the dimensionality of spacetime.  We believe these to be inessential for Hawking's paradox.  We are thus optimistic that the basic framework of a non-isometric code protected by computational complexity should also be implementable in theories of quantum gravity which do have all these features, with the tension between (1)-(3) resolved on similar lines.  In particular in the AdS/CFT correspondence, we find it quite plausible that the non-isometric encoding map $V$ we use as an input to our models is more or less uniquely determined by compatibility with standard AdS/CFT dictionary near the boundary together with an assumption that black holes formed by quick collapse should have a smooth interior.\footnote{We emphasize that we are allowed to assume that a horizon of a quick-collapse black hole which has not been recently disturbed from the outside is smooth unless this leads to a contradiction, and in our models it does not.}  We leave the details of all this to future work.

\item \textit{Where does Hawking's argument go wrong?} He did not realize that point (3) in the previous question, the validity of effective field theory in the vicinity of a black hole, only holds for sub-exponential observables.  Once we make this restriction then there is no obvious contradiction, and indeed our dynamical model makes it clear that there is no inconsistency between (1)-(3). 
\item \textit{Is there non-locality?}  The dynamical rules of both the fundamental and effective descriptions respect locality, with the former being clear from figure \ref{fig:bh_dyn} and the latter being clear from the discussion below \eqref{eq:effdecomp}.  On the other hand the relationship between the two via the holographic map $V$ involves a very non-local post-selection, which teleports information from deep inside the black hole out into the radiation and annihilates a large number of highly-entangled null states.  Thus the manifest locality in the effective description is really only an emergent notion, approximately valid for some questions but not others.  One can interpret this as evidence that there is a large amount of non-locality in whatever is the full non-perturbative theory of quantum gravity, i.e. non-pertubative string theory, since whatever this is it would have to be equivalent to the fundamental description via something like AdS/CFT.
\item \textit{What radiation entropy is consistent with Hawking's calculation?} The radiation entropy which follows from Hawking's picture is $S(\chot_R)$, which is the entropy in the effective description.  This increases throughout the evaporation.  In the fundamental description we can interpret it as the simple entropy $S^{\rm simple}(\Psi_R)$, as explained in section \ref{simpentsec}.
\item \textit{If Hawking's calculation was wrong, why are we allowed to use his state in the QES formula?} Because Hawking's picture is correct in the effective description, which is mapped into the fundamental description by a non-isometric quantum code.  The QES formula is a general feature of quantum codes \cite{Harlow:2016vwg,Akers:2021fut}, as we have confirmed by explicit calculation in our models. 
\item \textit{What happens in the interior if you make a complete projective measurement of sub-exponential observables on the Hawking radiation, i.e. measuring all the qubits in the $Z$ basis?} Nothing.\footnote{Here we mean ``nothing that can be detected locally''.  Since the interior is entangled with the exterior in the effective description, the results of interior measurements are correlated with the results of exterior measurements in the usual EPR way.}  This is a sub-exponential operation, and thus does not disrupt the validity of the effective description.  The same is true if the radiation interacts with a dust cloud.
\item \textit{What happens in the interior if after the Page time you act on the radiation with the reconstruction of an interior creation operator?} You create a particle in the interior.  This operation is exponentially complex, and so cannot be described solely within the effective description, but the resulting state is a sub-exponential state (we could have just created the particle directly using an interior operator), whose unique interpretation \`a la lemma \ref{invlem} is that there is a particle in the interior.  
\item \textit{Is there fundamental averaging?} No.  We use randomness for doing calculations, but in a fixed theory there is a fixed map $V$ that doesn't change.
\item \textit{Is the vacuum frozen \cite{Bousso:2013ifa}?}  No.  There are many sub-exponential states in the effective description which have structure at the horizon, and since the encoding map $V$ is fixed there is no re-definition that removes this structure.   
\item \textit{Are there firewalls for old black holes formed by quick collapse \cite{Almheiri:2012rt}?}  No.   A black hole created by a quick collapse evolves through its Page time without disrupting the validity of the effective description.  It is true that null states begin to appear at the Page time, but they cannot be detected by sub-exponential operations and jumping in to the black hole and seeing if you make it is surely sub-exponential.
\item \textit{If we collapse the radiation into a second black hole in the thermofield double state with the first black hole, is there a shared interior \cite{VanRaamsdonk:2013sza}?} Yes.  Just as in our answer to acting with the reconstruction of an interior creation operator, this is an exponentially complex operation which cannot be understood in the effective description.  Nonetheless it results in a sub-exponential state in the effective description of the two-black hole system which is in the connected sector (see section \ref{multisec}), so that is how it should be interpreted.
\item \textit{Renyi entropies can be computed via the expectation values of (sub-exponential) permutation observables (such as the swap operator for the second Renyi entropy), so why doesn't the effective description compute them correctly?} Because in situations where the Renyi entropies differ between the fundamental and effective descriptions, these differences appear as exponentially small contributions to the expectation values of these observables and the effective description is only accurate up to exponentially small corrections. 
\item \textit{How do we avoid Mathur's ``small corrections'' theorem \cite{Mathur:2009hf}?} Mathur's theorem (including its ``effective'' version \cite{Guo:2021blh}) assumes that the state of the radiation in the fundamental description is the same as the state of the radiation in the effective description.  In our language, it assumes that $V$ is an approximate isometry.  After the Page time however it isn't. \item \textit{When can the holographic encoding map be non-isometric?} We don't have a complete answer to this question, but it is reasonable to expect that whenever the minimal QES is `novel' -- i.e. when it exists when quantum corrections are included but disappears at strictly infinite $N$ -- the encoding map should be non-isometric. We can't prove this, but it is consistent with the expectation that a novel QES becomes minimal whenever the effective description has a larger Hilbert space dimension than the fundamental description does. A related argument was recently made in~\cite{Bousso:2022tdb}. If this is indeed borne out, then under the assumption of some version of quantum cosmic censorship (and non-compact time slices), the encoding can only be non-isometric when there is a black hole (or cosmological singularity) present in the effective description.
\item \textit{Are there firewalls (i.e no emergent interior) in typical states \cite{Marolf:2013dba}?} We don't know.  We've seen that our models can explain most of the key features of quantum black holes without requiring the answer to be ``no'', and we've moreover seen that our construction of the interior fails if we try to include typical states in the effective description  (in particular see the discussion at the end of section \ref{adsonesec}).  It is therefore natural to guess that the answer is ``yes''.  Some solace however can be taken in the fact that our models reproduce the observation of \cite{Susskind:2015toa} that such states are very unstable: any small perturbation restores a smooth horizon after a scrambling time, and this smoothness then lasts for an exponentially long time.
\item \textit{Is it ok that quantum mechanics is only approximate for the infalling observer?}  We think so. Due to the presence of the singularity there is a fundamental limit on how accurate of an experiment an infaller can do, so they do not need a quantum theory which computes arbitrarily precise probabilities \cite{Harlow:2010my}.  
\item \textit{What does this all mean for cosmology?}  We don't know.
\ei

\paragraph{Acknowledgments:}
We thank Scott Aaronson, Ahmed Almheiri, Raphael Bousso, Wissam Chemissany, Patrick Hayden, Hong Liu, Juan Maldacena, Don Marolf, Henry Maxfield, Yasunori Nomura, and Herman Verlinde for useful discussions.  CA is supported by the Simons Foundation as an ``It from Qubit'' fellow, the Air Force Office of Scientific Research under the award number FA9550-19-1-0360, the US Department of Energy under grant DE-SC0012567, the John Templeton Foundation and the Gordon and Betty Moore Foundation via the Black Hole Initiative, and the National Science Foundation under grant no. PHY-2011905.  NE is supported in part by NSF grant no. PHY-2011905, by the U.S. Department of Energy Early Career Award DE-SC0021886, by the John Templeton Foundation and the Gordon and Betty Moore Foundation via the Black Hole Initiative, and by funds from the MIT department of physics. DH is supported by the Simons Foundation as a member of the ``It from Qubit'' collaboration, the Sloan Foundation as a Sloan Fellow, the Packard Foundation as a Packard Fellow, the Air Force Office of Scientific Research under the award number FA9550-19-1-0360, the US Department of Energy under grants DE-SC0012567 and DE-SC0020360, and the MIT department of physics. GP is supported by the UC Berkeley Physics Department, the Simons Foundation through the ``It from Qubit'' program, the Department of Energy via the GeoFlow consortium (QuantISED Award DE-SC0019380), and AFOSR award FA9550-22-1-0098.
He also acknowledges support from an IBM Einstein Fellowship at the Institute for Advanced Study.
 SV is supported by US Department of Energy under grants DE-SC0012567 and DE-SC002036.

\appendix
\section{Unitary integrals}\label{Uapp}
In the main text we often want to integrate products of some number of matrix elements of a unitary matrix $U$ and its adjoint $U^\dagger$ over the unitary group $U(N)$ using the Haar measure $dU$ (normalized so that $\int dU=1$).  Such integrals can be evaluated using the left and right invariance of the Haar measure.  In particular since the integral must be invariant under $U'=e^{i\theta}U$ we must have the same number of $U$'s and $U^\dagger$'s, and the Haar invariance and permutation symmetry require that we have
\be\label{unitaryint}
\int dU U_{i_1 j_1}\ldots U_{i_n j_n} U^\dagger_{j'_1 i'_1}\ldots U^\dagger_{j'_ni'_n}=\sum_{\sigma,\tau \in S_n}\delta_{i_1i'_{\sigma(1)}}\ldots \delta_{i_ni'_{\sigma(n)}}\delta_{j_1j'_{\tau(1)}}\ldots \delta_{i_ni'_{\tau(n)}}\mathrm{Wg}(\sigma \tau^{-1},N).
\ee
Here $S_n$ is the permutation group on $n$ elements, and $\mathrm{Wg}$ is some real-valued function of a permutation $\pi\in S_n$ and the unitary dimensionality $N$ which is called the Weingarten function.  The permutation symmetry of the integral ensures that $\mathrm{Wg}(\pi,N)$ depends only on the conjugacy class of $\pi$, and we will see below that for fixed $\pi$ it is a rational function of $N$.  

For small values of $n$ the Weingarten function can be determined by contracting some of the indices and using that $UU^\dagger=1$.  In this manner we immediately find
\be
\mathrm{Wg}((1),N)=\frac{1}{N}
\ee
and
\begin{align}\nonumber
\mathrm{Wg}((1)(2),N)&=\frac{1}{N^2-1}\\
\mathrm{Wg}((1,2),N)&=-\frac{1}{N(N^2-1)},
\end{align}
where we use the usual disjoint cycle description of permutations and we only need to give the value of the function for a single representative of each conjugacy class.  More explicitly, we have
\begin{align}\nonumber
\int dU\, U_{i_1 j_1}U^\dagger_{j'_1 i'_1}&=\frac{1}{N}\delta_{i_1 i'_1}\delta_{j_1 j'_1}\\\nonumber
\int dU\, U_{i_1 j_1} U_{i_2 j_2}U^\dagger_{j'_1 i'_1}U^\dagger_{j'_2 i'_2}&=\frac{1}{N^2-1}\Bigg[\delta_{i_1i'_1}\delta_{i_2i'_2}\delta_{j_1 j'_1}\delta_{j_2 j'_2}+\delta_{i_1i'_2}\delta_{i_2i'_1}\delta_{j_1 j'_2}\delta_{j_2 j'_1}\\
&\qquad\qquad-\frac{1}{N}\Big(\delta_{i_1i'_1}\delta_{i_2i'_2}\delta_{j_1 j'_2}\delta_{j_2 j'_1}+\delta_{i_1i'_2}\delta_{i_2i'_1}\delta_{j_1 j'_1}\delta_{j_2 j'_2}\Big)\Bigg]\label{Uresults}
\end{align}

This ad hoc method becomes quite tedious for higher $n$, so it is useful to obtain some kind of general formula for $\mathrm{Wg}(\pi,N)$ that holds for any $n$.  This can be achieved by using some facts about the representation theory of the permutation and unitary groups, and in particular their relationship via Schur-Weyl duality \cite{collins2003moments,collins2006integration}. We will not develop this theory in detail, but we do need to review a few things to state the resulting formula.  In general the conjugacy classes of the permutation group $S_n$ are determined by its cycle structure: all permutations whose disjoint cycles have the same lengths are in the same conjugacy class.  For example in $S_4$, $(1,2)(3,4)$ and $(1,3)(2,4)$ are in the same conjugacy class but $(1,2,3)(4)$ is not.  Thus conjugacy classes of $S_n$ are in one-to-one correspondence with partitions $\lambda$ of $n$.  Moreover these partitions have a natural one-to-one correspondence with the irreducible representations of $S_n$, which proceeds through the theory of Young tableau and symmetrizers \cite{fulton2013representation}.  We can introduce the Schur polynomials 
\be
s_{\alpha,N}(X)\equiv \frac{1}{n!}\sum_{\pi\in S_n}\chi_\alpha(\pi)\tr(X^{\lambda_1(\pi)})\tr(X^{\lambda_2(\pi)})\ldots,
\ee
where $X$ is an $N\times N$ matrix, $\chi_\alpha$ is the character of the irreducible representation $\alpha$ of $S_n$, and $\lambda(\pi)$ is the partition associated to the cycle structure of $\pi$.  For the Weingarten function we are particularly interested in
\be
s_{\alpha,N}(I_N)=\frac{1}{n!}\sum_{\pi\in S_n}\chi_\alpha(\pi)N^{\ell(\lambda(\pi))},
\ee
where $\ell(\lambda)$ is the number of elements of the partition $\lambda$ (sometimes called its length).  The Schur polynomials have the interesting property that although they are constructed using characters of the permutation group $S_n$, through Schur-Weyl duality they are also characters of irreducible representations of $U(N)$.  This enables us to evaluate the integral of the product of two Schur polynomials over the unitary group $U(N)$ using Schur orthogonality, and thus to explicitly compute the Weingarten function.  The result is that \cite{collins2003moments,collins2006integration}
\be\label{Wgresult}
\mathrm{Wg}(\pi,N)=\frac{1}{(n!)^2}\sum_{\alpha}\frac{d_\alpha^2 \chi_{\alpha}(\pi)}{s_{\alpha,N}(I_N)},
\ee
where $d_\alpha$ is the dimensionality of $\alpha$ and the sum is over all irreducible representations.  Thus the calculation of $\mathrm{Wg}(\pi,N)$ is reduced  to an exercise in looking up elements of the character table of $S_n$.  For example doing this for $S_3$ and $S_4$, one finds
\begin{align}\nonumber
\mathrm{Wg}((1)(2)(3),N)&=\frac{N^2-2}{N(N^2-1)(N^2-4)}\\\nonumber
\mathrm{Wg}((1,2)(3),N)&=-\frac{1}{(N^2-1)(N^2-4)}\\
\mathrm{Wg}((1,2,3),N)&=\frac{2}{N(N^2-1)(N^2-4)}\label{tripleWg}
\end{align}
and 
\begin{align}\nonumber
\mathrm{Wg}((1)(2)(3)(4),N)&=\frac{N^4-8N^2+6}{N^2(N^2-1)(N^2-4)(N^2-9)}\\\nonumber
\mathrm{Wg}((1,2)(3)(4),N)&=-\frac{1}{N(N^2-1)(N^2-9)}\\\nonumber
\mathrm{Wg}((1,2)(3,4),N)&=\frac{N^2-6}{N^2(N^2-1)(N^2-4)(N^2-9)}\\\nonumber
\mathrm{Wg}((1,2,3)(4),N)&=\frac{2N^2-3}{N^2(N^2-1)(N^2-4)(N^2-9)}\\
\mathrm{Wg}((1,2,3,4),N)&=-\frac{5}{N(N^2-1)(N^2-4)(N^2-9)}.
\end{align}
The poles in these expressions at small values of $N$ may look alarming, but in fact they cancel in the sum on the right-hand side of \eqref{unitaryint}.  For example with $n=2$ if we take all indices to be equal then we have
\be
\frac{2}{N^2-1}-\frac{2}{N(N^2-1)}=\frac{2}{N(N+1)},
\ee
which is finite when $N=1$.

We can study the asymptotics of $\mathrm{Wg}(\pi,N)$ at fixed $n$ and large $N$.  In this limit the Schur polynomials are dominated by the contribution of the identity (the longest partition), 
\be
s_{\alpha,N}(I_N)=\frac{d_\alpha}{n!}N^n+O(N^{n-1}), 
\ee
so we have
\begin{align}\nonumber
\mathrm{Wg}(\pi,N)&=\frac{1}{n!}\sum_\alpha \frac{d_\alpha \chi_\alpha(\pi)}{N^n}+O(1/N^{n+1})\\
&=\frac{\delta(\pi,e)}{N^n}+O(1/N^{n+1}),\label{Wgasymp}
\end{align}
which is compatible with the above formulas.  Beware however that in applications one often sums over the various indices in \eqref{unitaryint}, and these sums can potentially enhance subleading terms.

Finally we briefly introduce the idea of a unitary $k$-design.  This is a finite set $\mathcal{T}$ of elements of $U(N)$ with the property that we can replace the integral in equation \eqref{unitaryint} by an average over $\mathcal{T}$ and the equality will still hold provided that $n\leq k$ \cite{dankert2009exact}.  One can also introduce approximate unitary $k$-designs by saying that the equality holds for $n\leq k$ in some appropriate approximation.  Unitary $k$-designs, often just called $k$-designs, are useful because examples can be found which consist only of unitaries which are much simpler than a generic Haar random unitary.  For example in \cite{dankert2009exact} a unitary $2$-design of accuracy $\epsilon$ was constructed for the case of $n$ qubits whose elements have circuit depth at most $O\left(\log n \times \log \frac{1}{\epsilon}\right)$.

\section{Measure concentration}\label{measureapp}
In this appendix we give an introduction to the theory of measure concentration.  We could not find any single source which presents everything which is needed, our exposition draws primarily on  \cite{meckes2019random} and \cite{anderson2010introduction}, with \cite{emery1987simple} and \cite{ledoux2001concentration} also being useful (additional historical references include \cite{bakry1985diffusions,gromov1983topological,szarek1990spaces}).  We make no claim to originality, the goals of this appendix are pedagogical.

\subsection{Bakry-Emery theory and logarithmic Sobolev inequalities}
Let $(M,g)$ be a $d$-dimensional Riemannian manifold, and
\be
d\mu_\Phi=e^{-\Phi} \sqrt{g}d^d x
\ee
be a probability measure on $M$.  Here $\Phi$ is any smooth scalar function on $M$, with the restriction that if $M$ is noncompact then $\Phi$ must grow fast enough at infinity to ensure the convergence of $\int_M d\mu_\Phi=1$.  It will be convenient to write the expectation value of any random variable $F:M\to \mathbb{R}$ as
\be
\lan F\ran_\Phi\equiv \int_M F d\mu_\Phi.
\ee
\begin{mydef}
The measure space $(M,\mu_\Phi)$ obeys a \textbf{Bakry-Emery condition} if there exists $C>0$ such that for all smooth functions $f:M\to \mathbb{R}$ we have
\be\label{BEcrit}
\nabla_\alpha \nabla_\beta f \nabla^\alpha \nabla^\beta f+\left(R^{\alpha\beta}+\nabla^\alpha \nabla^\beta \Phi\right)\nabla_\alpha f \nabla_\beta f\geq \frac{1}{C}\nabla_\alpha f \nabla^\alpha f.
\ee
\end{mydef}
This criterion is interesting because it implies a measure concentration result:
\begin{thm}\label{mcthm}
Let $(M,g)$ be a Riemannian manifold with probability measure $d\mu_\Phi=e^{-\Phi}\sqrt{g}d^dx$ such that the Bakry-Emery condition \eqref{BEcrit} is satisfied with some constant $C>0$.  Then for every $\kappa$-Lipschitz function $F:M\to \mathbb{R}$ obeying $\lan F\ran_\Phi<\infty$ we have
\begin{align}\nonumber
\PR\left[F\geq \lan F\ran_\Phi+\epsilon\right]\leq \exp\left(-\frac{\epsilon^2}{2C\kappa^2}\right)\\
\PR\left[F\leq \lan F\ran_\Phi-\epsilon\right]\leq \exp\left(-\frac{\epsilon^2}{2C\kappa^2}\right),\label{onesign1}
\end{align}
and thus
\be
\PR\left[|F-\lan F\ran_\Phi|\geq\epsilon\right]\leq 2\exp\left(-\frac{\epsilon^2}{2C\kappa^2}\right)\label{mceq}.
\ee
\end{thm}
In applications of interest it is often the case that when the dimension of $M$ is large then the argument in the exponents here is large and negative.  For example we can take $(M,g)$ to be $\mathbb{S}^d$ with the unit round metric. The Ricci tensor is $R_{\alpha \beta}=(d-1)g_{\alpha \beta}$, and thus a Bakry-Emery condition holds with $C=\frac{1}{d-1}$.  Levy's lemma \ref{Levy} is thus an immediate consequence of theorem \ref{mcthm}.\footnote{Note that the Lipschitz constant $\kappa$ appearing in theorem \ref{mcthm} is for the geodesic distance on $\mathbb{S}^d$.  In quantum information theory we are often interested instead in the chordal distance, which on the sphere is related to the geodesic distance by $d_c(x,y)\leq d_g(x,y)\leq \frac{\pi}{2} d_c(x,y)$. In particular the first of these inequalities implies that a Lipschitz constant for the chordal distance also gives a Lipschitz constant in terms of geodesic distance, while the second implies that if $\kappa$ is a Lipschitz constant for the geodesic distance then $\frac{\pi}{2}\kappa$ is a Lipschitz constant for the chordal distance.}  Another simple application is to take $(M,g)$ to be $\mathbb{R}^d$ with the flat metric and the Gaussian distribution
\be
e^{-\Phi(x)}=\frac{1}{\sqrt{2\pi \sigma^2}}e^{-\frac{|x|^2}{2\sigma^2}},
\ee
for which we can take $C=\sigma^2$ and thus we have\footnote{The dependence on dimension is less obvious here, but in practice the mean of $F$ often scales with $d$, e.g. $\lan |x|\ran\sim \sqrt{d}$, and the fractional deviations obey $\PR\left[|F-\lan F\ran|/\lan F \ran\geq \epsilon\right]\leq 2 \exp\left(-\frac{\epsilon^2\lan F\ran^2}{2\sigma^2\kappa^2}\right)$.}
\be
\PR\left[|F-\lan F\ran|\geq \epsilon\right]\leq 2 \exp\left(-\frac{\epsilon^2}{2\sigma^2\kappa^2}\right).
\ee
Later in this appendix we will see how theorem \ref{mcthm} can also be applied to the case where $M$ is a compact matrix group: it applies directly if $M=SU(N), SO(N)$, and with some further work the results \eqref{onesign1}, \eqref{mceq} also apply to $M=U(N)$,  with
\be
C=\begin{cases} \frac{2}{N} & SU(N)\\ \frac{4}{N-2} & SO(N)\\ \frac{6}{N} & U(N)\end{cases}.
\ee

The proof of theorem \ref{mcthm} involves the following notion:
\begin{mydef}
Let $(X,d)$ be a metric space and $\mu$ be a Borel probability measure on $X$. $(X,d,\mu)$ is said to obey a \textbf{logarithmic Sobolev inequality} (LSI) if there exists a constant $C>0$ such that for every locally-Lipschitz function $f:X\to \mathbb{R}$ we have
\be\label{LSI}
\left\lan f^2 \log \frac{f^2}{\lan f^2\ran}\right\ran\leq 2C \lan |\nabla f|^2\ran,
\ee
where $|\nabla f|(x)\equiv \limsup_{y\to x}\frac{|f(y)-f(x)|}{d(x,y)}$.
\end{mydef}
We can get some intuition for this definition by rewriting it in terms of relative entropy: 
\be\label{entLSI}
S\left(\frac{f^2}{\lan f^2\ran}d\mu,d\mu\right)\leq 2C\frac{\lan|\nabla f|^2\ran}{\lan f^2\ran}.
\ee
Thus \eqref{LSI} says that the distribution $\frac{f^2}{\lan f^2\ran}d\mu$ will be close to the distribution $d\mu$ unless the average variation of $f$ is large compared to $1/C$.  Since (classical) relative entropy is a measure of the distinguishability of probability distributions, it is hopefully plausible that \eqref{entLSI} implies a measure concentration result. 

To prove theorem \ref{mcthm} we first show that the concentration result \eqref{mceq} indeed holds for any metric measure space $(X,d,\mu)$ obeying a logarithmic Sobolev inequality, and then show that such an inequality is implied by the Bakry-Emery condition.  We can formulate the first part as a lemma:
\begin{lemma}\label{herbstlemma}
(Herbst) Let $(X,d,\mu)$ be a metric probability space obeying an LSI with $C>0$, and let $F:X\to \mathbb{R}$ be $\kappa$-Lipschitz and obey $\lan F\ran<\infty$.  Then we have
\begin{align}\nonumber
\PR\left[F\geq\lan F\ran+\epsilon\right]\leq  \exp\left(-\frac{\epsilon^2}{2C\kappa^2}\right)\\
\PR\left[F\leq \lan F\ran-\epsilon\right]\leq  \exp\left(-\frac{\epsilon^2}{2C\kappa^2}\right)\label{onesign}
\end{align}
and thus
\be
\PR\left[|F-\lan F\ran|\geq \epsilon\right]\leq 2 \exp\left(-\frac{\epsilon^2}{2C\kappa^2}\right).\label{twosign}
\ee
Moreover if $X$ is a smooth manifold and $d$ arises from a Riemannian metric on $X$, then in establishing \eqref{onesign} and \eqref{twosign} it is sufficient to assume that \eqref{LSI} holds for all smooth functions $f$.
\end{lemma}
\begin{proof}
The two lines of \eqref{onesign} are equivalent since we can flip the sign of $F$, so we just need to establish the first line.  By Markov's inequality we have
\be
\PR\left[F-\lan F\ran\geq \epsilon\right]=\PR\left[e^{\lambda(F-\lan F\ran)}\geq e^{\lambda \epsilon}\right]\leq e^{-\lambda \epsilon}\big\lan e^{\lambda(F-\lan F\ran)}\big\ran
\ee
for any $\lambda>0$.  We will show momentarily that
\be\label{herbstbound}
\big\lan e^{\lambda(F-\lan F\ran)}\big\ran\leq e^{\frac{C\lambda^2\kappa^2}{2}}, 
\ee
so choosing $\lambda=\frac{\epsilon}{C\kappa^2}$ to make the exponent as negative as possible we then indeed have \eqref{onesign}.
To establish \eqref{herbstbound}, we apply the LSI \eqref{LSI} to the function\footnote{Here we omit an argument dealing with the possibility that $\lan e^{\lambda F}\ran_\Phi$ is infinite.  The idea is to truncate $F$ by the bounded function
\be\label{truncate}
F_\delta(x)=\begin{cases} \frac{1}{\delta} & F(x)>\frac{1}{\delta}\\ F(x) & |F(x)|<\frac{1}{\delta}\\ -\frac{1}{\delta} & F(x)<-\frac{1}{\delta}\end{cases},
\ee
run the argument for $\delta>0$, and then take take $\delta\to 0$ at the end using Fatou's lemma.}
\be
f=e^{\lambda(F-\lan F\ran)/2}.  
\ee
Defining 
\be
B(\lambda)\equiv \log \lan f^2\ran,
\ee
the LSI implies that
\be
\frac{d}{d\lambda}\left(\frac{B(\lambda)}{\lambda}\right)=\frac{1}{\lambda^2}\frac{1}{\lan f^2\ran}\left\lan f^2 \log \left(\frac{f^2}{\lan f^2\ran}\right)\right\ran\leq \frac{2C}{\lambda^2}\frac{\lan |\nabla f|^2\ran}{\lan f^2\ran}.
\ee
Moreover by the Lipschitz property of $F$ we have
\be
|\nabla f|^2=\frac{\lambda^2}{4}|\nabla F|^2e^{\lambda(F-\lan F\ran)}\leq \frac{\lambda^2\kappa^2}{4}f^2,
\ee
and thus
\be
\frac{d}{d\lambda}\left(\frac{B(\lambda)}{\lambda}\right)\leq \frac{C\kappa^2}{2}.
\ee
Integrating and using that $\lim_{\lambda\to 0}\frac{B(\lambda)}{\lambda}=0$, we thus have
\be
B(\lambda)\leq \frac{C\kappa^2 \lambda^2}{2},
\ee
which directly implies \eqref{herbstbound}.  In the case where $X$ is a smooth manifold and $d$ is induced by a Riemannian metric, we observe that any $\kappa$-Lipschitz function $F:M\to\mathbb{R}$ can be approximated by a smooth function $G_\delta:M\to\mathbb{R}$ with the same Lipschitz constant $\kappa$ such that $|G_\delta(x)-F(x)|\leq \delta$ (for example in each patch we can smear $F$ against a bump function of compact support).  Therefore we can repeat the proof for $G_{\delta}$, note that $\PR[F\geq \lan F\ran+\epsilon]\leq P[G_\delta\geq \lan G_\delta\ran+\epsilon-2\delta]$ for any $\delta>0$, and then take the limit $\delta \to 0$.
\end{proof}

The next step is to show that an LSI is implied by the Bakry-Emery condition \eqref{BEcrit} on a Riemmanian manifold $(M,g)$ with probability measure $\mu_\Phi$.  The proof is based on studying a generalized heat flow on $M$ , which is generated by the differential operator
\be
\LP\equiv \nabla^2-\nabla^\mu \Phi\nabla_\mu.
\ee 
A natural set of functions to consider in this situation is $L^2_\Phi(M)$, which are the set of functions $f:M\to \mathbb{R}$ such that
\be
\lan f^2\ran<\infty.
\ee
We can view $L^2_\Phi(M)$ as a Hilbert space, with inner product
\be
\lan g,f\ran_\Phi=\int gf d\mu_\Phi,
\ee
and the operator $\LP$ is symmetric in the sense that
\be\label{LPsym}
\lan g,\LP f\ran_\Phi=\lan \LP g,f\ran_\Phi
\ee
and also obeys
\be\label{LPzero}
\lan \LP f\ran_\Phi=0
\ee
whenever $f$ and $g$ are in the domain of $\LP$.  We then define the generalized heat flow operation $P_t$ on smooth functions $f\in L^2_\Phi(M)$, with $t\geq 0$, by
\be
P_t f\equiv e^{t\LP}f.
\ee
$P_t$ can then be extended to a linear operator on all of $L_\Phi^2(M)$ by continuity.  $P_t$ has the following useful properties for all $f,g\in L_\Phi^2(M)$:
\bi
\item $P_0 f=f$
\item $P_t P_s f=P_{t+s}f$ 
\item $P_t 1=1$
\item $\lan P_t f\ran_\Phi=\lan f\ran_\Phi$
\item $\lan P_t f,g\ran_\Phi=\lan f,P_t g\ran_\Phi$
\item $f\geq 0\implies P_t f\geq 0$
\ei
Most of these are obvious from the definitions and the symmetry of $\LP$, but the last one requires a bit of thought.  One argument is the following: say that at $t=0$ we have $f\geq 0$.  For all $t\geq 0$ $P_t f$  obeys the generalized heat equation 
\be\label{heat}
\frac{d}{dt}P_t f=\LP P_t f,
\ee
so if for $\epsilon>0$ we define
\be
g_\epsilon(t,x)\equiv P_t f+\epsilon t,
\ee
then for all $t\geq 0$ and $x\in M$  we have
\be\label{gbound}
\frac{d}{dt}g_\epsilon(t,x)>\LP g_\epsilon(t,x)
\ee
and also $g_\epsilon(0,x)\geq 0$.  We now argue that $g_\epsilon(t,x)\geq 0$ for all $t\geq 0$.  Indeed say that at some time we have $g_\epsilon(t,x)<0$. There must be a ``last'' time $t$ before this happens (i.e. the supremum over all times where $g_\epsilon$ is non-negative), and at whichever $x$ where $g_\epsilon$ is about to become negative by continuity we must have $g_\epsilon(t,x)=0.$  Since $g$ is still non-negative we must have $\nabla_\mu g_\epsilon(t,x)=0$ and $\nabla^2 g_\epsilon(t,x)\geq 0$, but then by \eqref{gbound} we must have $\frac{d}{dt}g_\epsilon>0$ and so $g$ is not actually about to become negative. Since this is true for any $\epsilon>0$, taking $\epsilon\to 0$ we see that we must also have $P_t f\geq 0$.   A family of operators $P_t$ obeying the above conditions are sometimes said to form a Markov semigroup.

To proceed further it is useful to now introduce the Bakry-Emery bilinear operators
\begin{align}\nonumber
\Gamma_1(f,g)&\equiv\frac{1}{2}\left(\LP(fg)-f\LP g-g\LP f\right)\\
\Gamma_2(f,g)&\equiv\frac{1}{2}\left(\LP \Gamma_1(f,g)-\Gamma_1(\LP f,g)-\Gamma_1(f,\LP g)\right),\label{BHdefs}
\end{align}
which more explicitly are given by
\begin{align}\nonumber
\Gamma_1(f,g)&\equiv \nabla_\mu f \nabla^\mu g\\
\Gamma_2(f,g)&\equiv \nabla_\alpha\nabla_\beta f \nabla^\alpha \nabla^\beta g+(R^{\alpha\beta}+\nabla^\alpha \nabla^\beta\Phi)\nabla_\alpha f \nabla_\beta g.
\end{align}
In terms of these the Bakry-Emery condition \eqref{BEcrit} can be written as
\be
\Gamma_2(f,f)\geq \frac{1}{C}\Gamma_1(f,f).
\ee
A first illustration of the usefulness of the Bakry-Emery condition is the following: let 
\be
\psi(s)\equiv P_s \Gamma_1(P_{t-s} f, P_{t-s} f).
\ee
We then have
\be
\psi'(s)=2P_s \Gamma_2(P_{t-s}f,P_{t-s}f)\geq \frac{2}{C}P_s\Gamma_1(P_{t-s}f,P_{t-s}f)=\frac{2}{C}\psi(s),
\ee
and thus
\be
\psi(s)\geq \psi(0)e^{2s/C}.
\ee
In particular setting $s=t$ we have
\be\label{BEbound}
\Gamma_1(P_t f,P_t f)\leq P_t\Gamma_1(f,f)e^{-2t/C},
\ee
which implies that if $\Gamma_1(f,f)$ is bounded then $P_t f$ approaches a constant at late times (this uses that $P_t g(x)\leq \sup_{y\in M}g(y)$ for any $g$, which follows from the non-negative-preserving property of $P_t$). 

We now show how the Bakry-Emery condition leads to an LSI.\footnote{By the last part of lemma \ref{herbstlemma} it is enough to establish that an LSI holds for all smooth $f$. We'll give the argument assuming that if $M$ is noncompact then $f$ has whatever convergence properties are needed to justify the following manipulations.  This issue can be dealt with more carefully by first establishing the result for bounded $f$, and then truncating a general $f$ as in \eqref{truncate} and using monotone convergence.  The truncated function will not in general be smooth (it may have ``corners''), but the smoothing nature of the heat kernel makes $P_t f$ smooth for all $t>0$, which is enough for the proof to go through.}  We first define 
\be
h\equiv \frac{f^2}{\lan f^2\ran_\Phi}
\ee
and 
\be
S_f(t)\equiv \lan f^2\ran_\Phi\lan P_th \log (P_th)\ran_\Phi,
\ee
in terms of which the LSI \eqref{LSI} we wish to prove is 
\be
S_f(0)\leq 2C \Gamma_1(f,f).
\ee
As $t\to \infty$ we have $h_t\to 1$, and thus $S_f(\infty)=0$, so we then have
\begin{align}\nonumber
S_f(0)&=-\int_0^\infty dt\frac{dS_f(t)}{dt}\\\nonumber
&=-\lan f^2\ran_\Phi\int_0^\infty dt \lan \LP P_t h\log (P_t h)+\LP P_t h\ran_\Phi\\
&=\lan f^2\ran_\Phi \int_0^\infty dt\lan\Gamma_1(P_t h, \log P_t h)\ran_\Phi,
\end{align}
where in the third equality we have used \eqref{heat}, \eqref{BHdefs}, \eqref{LPsym}, and \eqref{LPzero}.  We now wish to establish an upper bound on $S_f(0)$.  Thus:
\begin{align}\nonumber
\lan \Gamma_1(P_t h,\log (P_t h))\ran_\Phi&=\lan\Gamma_1(h,P_t \log(P_t h))\ran_\Phi\\\nonumber
&\leq \left\lan \sqrt{\Gamma_1(h,h)}\sqrt{\Gamma_1(P_t \log (P_t h),P_t\log (P_t h))}\right\ran_\Phi\\\nonumber
&\leq\sqrt{\left \lan \Gamma_1(h,h)/h\right \ran_\Phi\left \lan h \Gamma_1(P_t \log (P_t h),P_t\log (P_t h))\right\ran_\Phi}\\\nonumber
&\leq e^{-t/C}\sqrt{\left \lan \Gamma_1(h,h)/h\right\ran_\Phi\lan h, P_t \Gamma_1(\log(P_t h),\log(P_t h))\ran_\Phi}\\\nonumber
&=e^{-t/C}\sqrt{\left \lan \Gamma_1(h,h)/h\right\ran_\Phi\lan P_th,  \Gamma_1(\log(P_t h),\log(P_t h))\ran_\Phi}\\
&=e^{-t/C}\sqrt{\left \lan \Gamma_1(h,h)/h\right\ran_\Phi\lan \Gamma_1(P_t h,\log(P_t h))\ran_\Phi},
\end{align}
where we have twice used the Cauchy-Schwarz inequality (once in $\mathbb{R}^d$ and once in $L_\Phi^2(M)$), and also \eqref{BEbound} and the symmetry of $P_t$, and thus
\begin{align}\nonumber
\lan \Gamma_1(P_t h,\log (P_t h))\ran_\Phi&\leq e^{-2t/C}\left \lan \frac{\Gamma_1(h,h)}{h}\right\ran_\Phi\\
&=4e^{-2t/C}\frac{\lan \Gamma_1(f,f)\ran_\Phi}{\lan f^2\ran_\Phi}
\end{align}
Finally we then have
\begin{align}\nonumber
S_f(0)&\leq 4 \int_0^\infty dte^{-2t/C}\lan \Gamma_1(f,f)\ran_\Phi\\
&=2C\lan \Gamma_1(f,f)\ran_\Phi.
\end{align}

\subsection{Measure concentration on simple matrix groups}
We'll now use Bakry-Emery theory to establish measure concentration results on the classical matrix groups $SO(N)$ and $SU(N)$ (a similar argument works for $Sp(N)$ but we won't bother discussing it).  Bakry-Emery theory reduces this to the study of the Ricci curvature on the group manifold.  Indeed since each of these groups is a matrix group, we can endow it with a Riemannian metric by using the Hilbert-Schmidt distance
\be\label{HSd}
d(M_1,M_2)=||M_1-M_2||_2
\ee
to define a Euclidean metric on the set of all matrices and then pulling this back to the group submanifold (here $||X||_2\equiv \sqrt{\tr (X^\dagger X)}$).  This ``Hilbert-Schmidt metric'' is invariant under left and right unitary multiplication,
\be
d(UM_1,UM_2)=d(M_1U,M_2U)=d(M_1,M_2),
\ee
so these operations are isometries and the volume measure in this metric is indeed the Haar measure for any subgroup of the unitary group.  Since any compact Lie group $G$ can be constructed as a subgroup of $U(N)$ for some $N$, this gives an explicit construction of the Haar measure for any such group.  To avoid confusion we emphasize that geodesic distances in the Hilbert-Schmidt metric are \textit{not} given by the distance \eqref{HSd}: geodesics in the induced metric are required to stay within the group manifold, and thus will typically be longer than the Euclidean distance $||M_1-M_2||_2$.  On the other hand since the geodesic distance can't be shorter we have the useful fact that a $\kappa$-Lipschitz function with respect to the distance \eqref{HSd} will also be a $\kappa$-Lipschitz function with respect to the geodesic distance.

For any Lie group $G$ the Lie algebra $\mathfrak{g}$ is defined as the set of vector fields on $G$ which are invariant under push-forward by the left-multiplication maps $L_h:g\mapsto hg$ for any $h\in G$.  These vector fields are completely determined by their values at the identity, and for subgroups of $U(N)$ they are thus in one-to-one correspondence with some subspace of the anti-hermitian operators on $\mathbb{C}^N$.  More concretely, each element in the identity component of $G$ can be written as 
\be
g=e^{X}
\ee
for some antihermitian $X$.  If we chose a basis $T_a$ of anti-hermitian operators obeying
\be
[T_a,T_b]=C^c_{\phantom{c}ab}T_c
\ee 
and coordinates $\lambda^a$ in the neighborhood of the identity such that
\be
g=e^{\lambda^aT_a}, 
\ee
then near the identity the left-invariant vector fields have the form
\be
V^a(\lambda)=V^a(0)+\frac{1}{2}C^a_{\phantom{a}bc}\lambda^b V^c(0)+O(\lambda^2).
\ee
To each anti-hermitian matrix $X=X^a T_a$ we then associate the Lie-algebra element which near the identity is given by
\be
\wt{X}^a(\lambda)=X^a+\frac{1}{2}C^a_{\phantom{a}bc}\lambda^b X^c+O(\lambda^2).
\ee
From these formulas we see that at the identity the commutator of the Lie algebra is related to the commutators of the matrices by
\be
[\wt{X},\wt{Y}](0)=[X,Y],  
\ee
so by left-invariance we then have
\be\label{Liecomm}
\left[\wt{X},\wt{Y}\right]=\wt{[X,Y]}
\ee
everywhere on $G$.  In a neighborhood of the identity we can write the Riemannian metric induced by Hilbert-Schmidt metric as
\begin{align}\nonumber
g_{ab}(\lambda)&=\tr\left[\partial_a\left(e^{-\lambda^c T_c}\right)\partial_b\left(e^{\lambda^d T_d}\right)\right]\\
&=\tr(T_a^\dagger T_b)+O(\lambda^2).
\end{align}
At the identity we thus have
\be\label{Gmetric}
g_{ab}\wt{X}^a \wt{Y}^b=\tr\left(X^\dagger Y\right),
\ee
and by left-invariance this equation actually holds throughout $G$.

We can now study the Ricci tensor on $G$ in this metric.  For any three vector fields $U$, $V$, $W$ on a Riemannian manifold we have
\be
2U^\mu W_\nu \nabla_\mu V^\nu=U(g(V,W))+Vg(U,W)-Wg(U,V)+g([U,V],W)+g([W,U],V)+g([W,V],U).
\ee
If $U,V,W$ are Lie-algebra elements $\wt{X},\wt{Y},\wt{Z}$, then the first three terms on the right-hand side vanish by the left-invariance of \eqref{Gmetric}. Using \eqref{Liecomm} and \eqref{Gmetric} we thus have
\begin{align}\nonumber
\wt{X}^a \wt{Z}_b \nabla_a \wt{Y}^b&=\frac{1}{2}(\tr([X,Y]Z)+\tr([Z,X]Y)+\tr([Z,Y]X))\\\nonumber
&=\frac{1}{2}\tr([X,Y]Z)\\
&=\frac{1}{2}\wt{Z}_b [\wt{X},\wt{Y}]^b,
\end{align}
and therefore
\be
\wt{X}^a \nabla_a \wt{Y}^b=\frac{1}{2}[\wt{X},\wt{Y}]^b.
\ee
By the definition of the Riemann tensor we then have
\begin{align}\nonumber
R^a_{\phantom{a}bcd}\wt{Z}^b \wt{X}^c\wt{Y}^d&=[\wt{X}^c\nabla_c,\wt{Y}^d\nabla_d]\wt{Z}^a-[\wt{X},\wt{Y}]^b\nabla_b \wt{Z}^a\\\nonumber
&=\frac{1}{4}\left([\wt{X},[\wt{Y},\wt{Z}]]-[\wt{Y},[\wt{X},\wt{Z}]]\right)+\frac{1}{2}[\wt{Z},[\wt{X},\wt{Y}]]\\\nonumber
&=\frac{1}{4}[\wt{Z},[\wt{X},\wt{Y}]]\\
&=\frac{1}{4}\wt{[Z,[X,Y]]}.
\end{align}
We are interested in computing $R_{ab}\wt{Z}^a\wt{Z}^b$ for any $\wt{Z}$, which we can do by choosing the $T_a$ to be orthonormal in the sense that
\be
\tr(T_a^\dagger T_b)=\delta_{ab}
\ee
and then computing
\begin{align}\nonumber
R_{ab}\wt{Z}^a\wt{Z}^b&=\sum_a (\wt{T_a})_b, (\wt{T_a})^d R^b_{\phantom{b}cde}\wt{Z}^c\wt{Z}^e\\\nonumber
&=\frac{1}{4}\sum_a\tr\Big(T_a^\dagger[Z,[T_a,Z]]\Big)\\
&=\frac{1}{4}\sum_a \tr\Big(Z[T_a,[T_a,Z]]\Big).\label{RicciG}
\end{align}

Finally we can evaluate the trace in \eqref{RicciG} for various classical groups.  We will just consider $SU(N)$ and $SO(N)$, there is a similar calculation for $SP(N)$.  The generators of $SO(N)$ are a subset of those for $SU(N)$, so if we compute the terms in the sum in \eqref{RicciG} we can also obtain the result for $SO(N)$ by just including fewer terms.  We will see that the following proposition holds:
\begin{prop}\label{algebraprop}
In the Lie algebra of $SU(N)$ we have
\be\label{SUcom}
\sum_a[T_a,[T_a,Z]]=-2NZ,
\ee
while in the Lie algebra of $SO(N)$ we have
\be\label{SOcom}
\sum_a[T_a,[T_a,Z]]=-(N-2)Z.
\ee
\end{prop}
\begin{proof}
To prove \eqref{algebraprop}, we first observe that for any unitary $U\in U(N)$ we have
\be
UT_aU^\dagger=D_a^{\phantom{a}b}(U) T_b,
\ee
where the adjoint representation matrices $D_a^{\phantom{a}b}(U)$ are orthogonal.  Thus the quantity $\sum_a[T_a,[T_a,Z]]$ is invariant under $T_a\to U T_a U^\dagger$.  Moreover we have
\be
U\Big(\sum_a[T_a,[T_a,Z]]\Big)U^\dagger=\sum_a[UT_aU^\dagger,[UT_aU^\dagger,UZU^\dagger]]=\sum_a[T_a,[T_a,UZU^\dagger]],
\ee
so if \eqref{SUcom} or \eqref{SOcom} holds when $Z$ is equal to any particular $T_a$, then it holds for all $Z$ since  matrices of the form $UT_a U^\dagger$ span the Lie algebra.  A natural basis for the Lie algebra of $SU(N)$ is the Gell-Mann matrices:
\begin{align}\nonumber
(S_{\hat{m}\hat{n}})_{mn}&\equiv \frac{i}{\sqrt{2}}\left(\delta_{\hat{m},m}\delta_{\hat{n},n}+\delta_{\hat{m},n}\delta_{\hat{n},m}\right) \qquad \hat{m}<\hat{n}\\\nonumber
(A_{\hat{m}\hat{n}})_{mn}&\equiv \frac{1}{\sqrt{2}}\left(\delta_{\hat{m},m}\delta_{\hat{n},n}-\delta_{\hat{m},n}\delta_{\hat{n},m}\right) \qquad \hat{m}<\hat{n}\\
(D_{\hat{m}})_{mn}&\equiv\frac{i}{\sqrt{\hat{m}(\hat{m}+1)}}\delta_{m,n}\left(\sum_{k=1}^{\hat{m}}\delta_{m,k}-\hat{m}\delta_{m,\hat{m}+1}\right),
\end{align}
and some routine calculation shows that
\begin{align}
[S_{\hat{m}\hat{n}},[S_{\hat{m}\hat{n}},A_{12}]]&=\begin{cases} -2A_{12} & \hat{m}=1,\hat{n}=2\\ -\frac{1}{2}A_{12} & \hat{m}=1,\hat{n}>2\\ -\frac{1}{2}A_{12} & \hat{m}=2,\hat{n}>2\\ 0 &\hat{m}>2\end{cases}\\
[A_{\hat{m}\hat{n}},[A_{\hat{m}\hat{n}},A_{12}]]&=\begin{cases} 0 & \hat{m}=1,\hat{n}=2\\ -\frac{1}{2}A_{12} & \hat{m}=1,\hat{n}>2\\ -\frac{1}{2}A_{12} & \hat{m}=2,\hat{n}>2\\ 0 &\hat{m}>2\end{cases}\\
[D_{\hat{m}},[D_{\hat{m}},A_{12}]]&=\begin{cases}-2A_{12} & \hat{m}=1\\ 0 & \hat{m}>1\end{cases}.
\end{align}
Summing over all generators for $SU(N)$ or the $A_{\hat{m}\hat{n}}$ only for $SO(N)$ then directly gives \eqref{SUcom} and \eqref{SOcom}.
\end{proof}
Applying now this proposition to the curvature, from \eqref{RicciG} we have
\be
R_{ab}\wt{Z}^a\wt{Z}^b=\begin{cases} \frac{N}{2}\wt{Z}^a\wt{Z}_a & G=SU(N)\\
\frac{N-2}{4}\wt{Z}^a\wt{Z}_a & G=SO(N)\end{cases},
\ee
and thus the Bakry-Emery condition \eqref{BEcrit} holds on the group manifolds $SU(N)$ and $SO(N)$ in the Hilbert-Schmidt metric with
\be
C=\begin{cases} \frac{2}{N} & G=SU(N)\\ \frac{4}{N-2} & G=SO(N)\end{cases}.
\ee
The measure concentration results \eqref{onesign1} thus apply to these groups with these coefficients, with the Lipschitz constant $\kappa$ defined using either the geodesic distance or the Hilbert-Schmidt distance \eqref{HSd}.    

\subsection{Measure concentration on the unitary group}
It is also interesting to consider measure concentration on $U(N)$.  This does not obey a Bakry-Emery condition, since if we choose $Z$ to be proportional to the identity then $R_{ab}\wt{Z}^a \wt{Z}^b=0$.  On the other hand it turns out we can still use Bakry-Emery theory to show that Lipschitz functions on $U(N)$ obey the concentration results \eqref{onesign1} with $C=6/N$.  The basic idea is to observe that each element of $U(N)$ can be uniquely written as $e^{i\theta}V$ with $V\in SU(N)$ and $\theta\in [0,2\pi/N)$: we already know that there is measure concentration on $SU(N)$, and at large $N$ the range of $\theta$ is small so any Lipschitz function shouldn't vary much over it.  We just need to sharpen this latter observation into a concentration inequality and then learn how to combine these two concentration results.    

It is useful to first consider the case of $U(1)$.  A Bakry-Emery condition on $U(1)$ would say that any function $f\in L^2(\mathbb{S}^1)$ obeys $(f'')^2\geq \frac{1}{C}(f')^2$ for some $C>0$, but this is clearly false (here we parametrize $\mathbb{S}^1$ using $x\in [0,2\pi)$).  For example the function $f(x)=\sin(x)$  has $f''(0)=0$ but $f'(0)=1$.  On the other hand it is still true that functions on the circle obey a Bakry-Emery condition \textit{on average}:
if we expand $f$ in a Fourier basis
\be
f(x)=\sum_{n\in \mathbb{Z}}a_n e^{inx},
\ee
then we have
\be
\int\frac{dx}{2\pi} f''(x)^2=\sum_n n^4|a_n|^2\geq \sum_n n^2 |a_n|^2=\int \frac{dx}{2\pi}f'(x)^2\label{BEave},
\ee
which is the average of the Bakry-Emery condition with $C=1$.  This turns out to be sufficient to establish an LSI \eqref{LSI} with $C=1$.  We first observe that replacing $f^2 \to f$, with $f\geq 0$, we want to show that 
\be\label{newLSI}
\lan f \log f\ran-\lan f\ran \log \lan f\ran\leq \frac{1}{2}\Big\lan\frac{f'^2}{f}\Big\ran.
\ee
Introducing the heat flow operation $P_t$ as before, we have
\begin{align}\nonumber
\lan f \log f\ran-\lan f\ran \log \lan f\ran&=-\int_0^\infty dt \frac{d}{dt}\lan P_t f\log P_t f\ran\\\nonumber
&=-\int_0^\infty dt \Big\lan (P_t f)''(\log (P_t f)+1)\Big\ran\\\nonumber
&=\int_0^\infty dt \Big\lan (P_t f)'(\log (P_t f)+1)'\Big\ran\\
&=\int_0^\infty dt\Big \lan \frac{(P_t f')^2}{P_t f}\Big\ran.\label{circint}
\end{align}
Moreover we have
\begin{align}\nonumber
\frac{d}{dt}\left(e^{2t}\Big\lan \frac{(P_t f')^2}{P_tf}\Big\ran\right)&=e^{2t}\Big\lan 2\frac{(P_tf')^2}{P_tf}+2\frac{P_tf'P_tf'''}{P_tf}-\frac{(P_tf')^2P_tf''}{(P_tf)^2}\Big\ran\\\nonumber
&=4e^{2t}\Big\lan h'^2+\frac{2h'^2h''}{h}+h'h'''-\frac{h'^4}{h^2}\Big\ran\\\nonumber
&=4e^{2t}\Big\lan h'^2-h''^2-\frac{1}{3}\frac{h'^4}{h^2}\Big\ran\\
&\leq 0
\end{align}
where in the second line we've defined $P_tf=\frac{1}{2}h^2$ and in the last line we've used \eqref{BEave}.  Thus we have
\be
\Big\lan \frac{(P_t f')^2}{P_tf}\Big\ran\leq e^{-2t}\Big\lan \frac{f'^2}{f}\Big\ran
\ee
Using this in \eqref{circint} we recover \eqref{newLSI}.  Since $1<6$, we see that in the case $N=1$ we indeed have an LSI for $U(1)$ with $C=6/N$. 

Before proceeding it will be useful to note that we can use this LSI on $\mathbb{S}^1$ to also derive an LSI for functions on the compact interval $[0,a]$  We first consider the case $a=\pi$.  Any function $f:[0,\pi]\to \mathbb{R}$ can be extended to a function on $\mathbb{S}^1$ by
\be
\wt{f}(x)=\begin{cases} f(x) & 0\leq x \leq \pi\\ f(2\pi-x) & \pi\leq x \leq 2\pi\end{cases},
\ee 
and we then have
\be
\int_0^\pi \frac{dx}{\pi}f^2 \log f^2= \int_0^{2\pi} \frac{dx}{2\pi}\wt{f}^2 \log \wt{f}^2\leq 2\int_0^{2\pi} \frac{dx}{2\pi}|\nabla \wt{f}|^2=\int_0^\pi \frac{dx}{\pi}|\nabla f|^2,
\ee
so on $[0,\pi]$ with the uniform measure we have an LSI with $C=1$.  The general case of uniform measure on $[0,a]$ is converted back to the $[0,\pi]$ case by the change of variables $x'=\pi x/a$, which gives an LSI with
\be\label{Ca}
C=\frac{a^2}{\pi^2}.
\ee
Therefore, as suggested in the beginning of this subsection, from lemma \eqref{herbstlemma}  there is strong measure concentration on $[0,a]$ when $a$ is small and $\kappa$ is fixed.

We'd now like to use this LSI on $[0,a]$ together with the fact that $U(N)=(U(1)\times SU(N))/\mathbb{Z}_N$ to derive an LSI for $U(N)$ with $N>1$.  We first observe that any $U\in U(N)$ can be uniquely written as
\be
U=e^{i\theta}V,
\ee
with $V\in SU(N)$ and $\theta\in [0,2\pi/N)$.  Moreover if we take a uniform distribution on $\theta$ and a Haar distribution on $V$, this induces the Haar distribution on $U$.  To see this, we show that it is left-invariant: for any function $f:U(N)\to \mathbb{R}$ and any $\lambda\in \mathbb{R}$ we have
\begin{align}\nonumber
\PR\Big[f\left(e^{i\theta'}V' e^{i\theta}V\right)>\lambda\Big]&=\PR\Big[f\left(e^{i(\theta'+\theta\,\, \mathrm{mod} \frac{2\pi}{N})} e^{\frac{2\pi ik(\theta,\theta')}{N}}V'V\right)>\lambda\Big]\\
&=\PR\Big[f(e^{i\theta}V)>\lambda\Big],
\end{align}
where $k(\theta,\theta')$ is the integer such that $\theta+\theta'\,\,\mathrm{mod}\frac{2\pi}{N}=\theta+\theta'-\frac{2\pi k(\theta,\theta')}{N}$.

To proceed further we need to spend some time developing a ``tensorization'' property of logarithmic Sobolev inequalities:
\begin{thm}\label{tensorthm}
Let $(X_i,d_i, \mu_i)$ be a collection of $n$ metric probability spaces with Borel measure, and suppose that for each $i$ we have an LSI
\be
\Big\lan f_i^2\log \frac{f^2}{\lan f^2\ran}\Big\ran\leq 2C_i \lan |\nabla f_i|^2\ran
\ee
for all functions $f_i:X_i\to \mathbb{R}$.  Then the product space $X=X_1\times \ldots \times X_n$ in the metric
\be\label{prodmet}
d(\{x_1,\ldots,x_n\},\{y_1,\ldots,y_n\})=\sqrt{\sum_i d_i(x_i,y_i)^2}
\ee
and the product measure $\mu=\mu_1\otimes \ldots \otimes \mu_n$ obeys the inequality
\be
\Big\lan f^2\log \frac{f^2}{\lan f^2\ran}\Big\ran\leq 2C\sum_i \lan |\nabla_i f|^2,
\ee
with 
\be
C=\max_{i}C_i
\ee
and
\be
|\nabla_i f|\equiv \limsup_{y_i\to x_i}\frac{|f(x_1,\ldots, y_i, \ldots, x_n)-f(x_1,\ldots,x_n)|}{d_i(y_i,x_i)}.
\ee
\end{thm}
In particular if the $X_i$ are Riemannian manifolds with geodesic distance $d_i$ and $f:X\to \mathbb{R}$ is smooth then we have 
\be
\sum_i \lan |\nabla_i f|^2=|\nabla f|^2
\ee
and thus the LSI
\be
\Big\lan f^2 \log \frac{f^2}{\lan f^2\ran}\Big\ran\leq 2C\lan |\nabla f|^2\ran
\ee 
holds.
\begin{proof}
To prove theorem \ref{tensorthm}, we first establish that
\be
\Big\lan f^2 \log \frac{f^2}{\lan f^2\ran}\Big\ran=\sup_{\lan e^g\ran\leq 1}\lan f^2g\ran.
\ee
Rescaling $f$ so that $\lan f^2\ran=1$, by setting $g=\log f^2$ we see that the supremum must at least be as big as promised.  To see that it is no bigger, since $\log x\leq x-1$ we have
\be
f^2\log(e^g/f^2)\leq e^g-f^2
\ee
and thus
\be
\lan f^2g\ran\leq \lan f^2\log f^2+e^g-f^2\ran=\lan f^2 \log f^2\ran.
\ee
Now defining
\be
g_i(x_1,\ldots ,x_n)=\log \frac{\int d\mu_1(y_1)\ldots d\mu_{i-1}(y_{i-1}) e^{g(y_1,\ldots,y_{i-1},x_i,\ldots, x_n)}}{\int d\mu_1(y_1)\ldots d\mu_{i}(y_{i}) e^{g(y_1,\ldots,y_{i},x_{i+1},\ldots, x_n)}},
\ee
we have
\be
\int d\mu_i(x_i)e^{g_i(x_1,\ldots, x_n)}=1
\ee
and 
\be
\sum_{i=1}^n g_i(x_1,\ldots, x_n)=\log \frac{e^{g(x_1,\ldots,x_n)}}{\lan e^g\ran}\geq g(x_1,\ldots,x_n), 
\ee
so taking $g=\log f^2$ we have
\begin{align}\nonumber
\lan f^2 \log f^2\ran&=\lan f^2g\ran\\\nonumber
&\leq \sum_{i=1}^n \lan f^2 g_i\ran\\\nonumber
&=\int d\mu_1(x_1)\ldots d\mu_{i-1}(x_{i-1})d\mu_{i+1}(x_{i+1})\ldots d\mu_n(x_n)\lan f^2 g_i\ran_i\\\nonumber
&\leq \int d\mu_1(x_1)\ldots d\mu_{i-1}(x_{i-1})d\mu_{i+1}(x_{i+1})\ldots d\mu_n(x_n) \lan f^2 \log f^2\ran_i\\\nonumber
&\leq 2\sum_i C_i \lan |\nabla_i f|^2\ran\\
&\leq 2C \sum_i \lan |\nabla_i f|^2\ran.
\end{align}
\end{proof}

We'd now like to use tensorization to combine the LSIs on $[0,2\pi/N]$ and $SU(N)$ to derive a concentration inequality for $U(N)$.  There is some freedom in how we do this, but it turns out to be a good idea to define\footnote{More generally we could define $\theta=\frac{\lambda}{\sqrt{N}}t$, but setting $\lambda=\sqrt{2}$ turns out to give the best concentration inequality.}
\be
\theta=\sqrt{\frac{2}{N}}t,
\ee
with $t\in [0,\pi \sqrt{\frac{2}{N}}]$, so that by \eqref{Ca} maps from $[0,\pi \sqrt{\frac{2}{N}}]$ to $\mathbb{R}$ obey an LSI with $C=\frac{2}{N}$ just as we found for $SU(N)$.  Therefore by tensorization maps from $[0,\pi \sqrt{\frac{2}{N}}]\times SU(N)$ to $\mathbb{R}$, including those of the form $f(e^{i\sqrt{\frac{2}{N}}t}V)$, obey an LSI with $C=\frac{2}{N}$.  Therefore by lemma \ref{herbstlemma} they will obey a concentration inequality.  The only task remaining is to understand how a Lipschitz constant for a function $f:U(N)\to \mathbb{R}$ with the Hilbert-Schmidt distance \eqref{HSd} on $U(N)$ is related to a Lipschitz constant for the same function on $[0,\pi \sqrt{\frac{2}{N}}]\times SU(N)$ with the product metric \eqref{prodmet}.  Say that $f$ is $\kappa$ Lipschitz on $U(N)$.  Then we have
\begin{align}\nonumber
|f(e^{i\theta_1}V_1)-f(e^{i\theta_2}V_2)|&\leq \kappa ||e^{i\theta_1}V_1-e^{i\theta_2}V_2||_2\\\nonumber
&\leq \kappa\Big(||V_1-V_2||_2+||e^{i\theta_1}-e^{i\theta_2}||_2\Big)\\\nonumber
&\leq \kappa\Big(||V_1-V_2||_2+\sqrt{N}|e^{i\theta_1}-e^{i\theta_2}|\Big)\\\nonumber
&\leq \kappa\Big(||V_1-V_2||_2+\sqrt{2}|t_1-t_2|\Big)\\\nonumber
&\leq \kappa \sqrt{3}\sqrt{||V_1-V_2||_2^2+|t_1-t_2|^2}\\
&\leq \kappa \sqrt{3}\sqrt{d(V_1,V_2)^2+|t_1-t_2|^2}
\end{align}
so $f$ is $\sqrt{3}\kappa$-Lipschitz on $[0,\pi \sqrt{\frac{2}{N}}]\times SU(N)$ with the metric \ref{prodmet} (here $d(V_1,V_2)$ is the geodesic distance on $SU(N)$ and we have used that $||V_1-V_2||_2\leq d(V_1,V_2)$). Thus by lemma \ref{herbstlemma} we at last have the following:
\begin{thm}(Meckes) For any function $f$ which is $\kappa$ Lipschitz on $U(N)$ with respect to the Hilbert-Schmidt distance \eqref{HSd}, in Haar measure we have the concentration inequalities
\begin{align}\nonumber
\PR\Big[f\geq \lan f\ran+\epsilon\Big]&\leq \exp\left(-\frac{\epsilon^2N}{12\kappa^2}\right)\\\nonumber
\PR\Big[f\leq \lan f\ran-\epsilon\Big]&\leq \exp\left(-\frac{\epsilon^2N}{12\kappa^2}\right)\\
\PR\Big[|f-\lan f\ran |\geq\epsilon\Big]&\leq 2\exp\left(-\frac{\epsilon^2N}{12\kappa^2}\right).
\end{align}
\end{thm}
This is the same as lemma \ref{Meckes} in the main text.

\section{Proof of the deviation bound for overlaps of encoded states}\label{proofapp}
In this appendix we prove theorem \ref{setconcthm}.  We first observe that showing $V\otimes I_{LR}$ approximately preserves inner products is essentially equivalent to showing that it approximately preserves norms:
\begin{lemma}\label{overnorm}
Let $L$ be a linear map and $|\psi_1\ran$, $|\psi_2\ran$ be states such that for $|\psi\ran=|\psi_1\ran,|\psi_2\ran, |\psi_1\ran\pm |\psi_2\ran, \, \mathrm{or}\, |\psi_1\ran\pm i|\psi_2\ran$ we have
\be
1-\epsilon\leq \frac{||L|\psi\ran||}{|||\psi\ran||}\leq 1+ \epsilon
\ee
for some $0\leq\epsilon\leq 1$.  Then 
\be
|\lan \psi_2|L^\dagger L|\psi_1\ran-\lan \psi_2|\psi_1\ran|\leq \sqrt{18}\epsilon.
\ee
\end{lemma}
\begin{proof}
Noting that
\begin{align}\nonumber
4\mathrm{Re}\Big(\lan \psi_2|L^\dagger L|\psi_1\ran\Big)&=||L(|\psi_1\ran+|\psi_2\ran)||^2-||L(|\psi_1\ran-|\psi_2\ran)||^2\\
4\mathrm{Im}\Big(\lan \psi_2|L^\dagger L|\psi_1\ran\Big)&=||L(|\psi_1\ran+i|\psi_2\ran)||^2-||L(|\psi_1\ran-i|\psi_2\ran)||^2,
\end{align}
from the assumed inequalities we have
\begin{align}\nonumber
-(2\epsilon+\epsilon^2)&\leq \mathrm{Re}\left(\lan \psi_2|L^\dagger L|\psi_1\ran-\lan \psi_2|\psi_1\ran\right)\leq 2\epsilon+\epsilon^2\\
-(2\epsilon+\epsilon^2)&\leq \mathrm{Im}\left(\lan \psi_2|L^\dagger L|\psi_1\ran-\lan \psi_2|\psi_1\ran\right)\leq 2\epsilon+\epsilon^2,
\end{align}
so the conclusion follows from $\epsilon^2\leq\epsilon$ since $0\leq\epsilon\leq 1$.
\end{proof}
It will also be useful to establish the following bounds:
\begin{lemma}\label{avbounds}
Let $V$ be defined as in \eqref{Vdef}, and $|\psi\ran\in \HL\otimes\Hl\otimes \Hr \otimes \HR$.  Then
\begin{align}\nonumber
\int dU ||(V\otimes I_{LR})|\psi\ran||&\leq 1\\
\int dU ||(V\otimes I_{LR})|\psi\ran||&\geq 1-\frac{1}{2}\frac{|P|-1}{|P|^2|B|^2-1}\left(|P||B|\tr(\psi_{LR}^2)-1\right).
\end{align}
\end{lemma}
\begin{proof}
The first inequality follows immediately from Jensen's inequality:
\be
\int dU \sqrt{\lan \psi|(V^\dagger V\otimes I_{LR})|\psi\ran}\leq \sqrt{\int dU\lan \psi|(V^\dagger V\otimes I_{LR})|\psi\ran}=1,
\ee
where in the second step we used \eqref{overlapav}.  To establish the second inequality we can observe that for all $x\geq 0$ we have $x\geq \frac{3}{2}x^2-\frac{1}{2}x^4$, and thus
\begin{align}\nonumber
\int \sqrt{\lan \psi|(V^\dagger V\otimes I_{LR})|\psi\ran}&\geq \frac{3}{2}-\frac{1}{2}\frac{|P|^2|B|^2-|P|+|P||B|(|P|-1)\tr \psi_{LR}^2}{|P|^2|B|^2-1}\\
&=1-\frac{1}{2}\frac{|P|-1}{|P|^2|B|^2-1}\left(|P||B|\tr(\psi_{LR}^2)-1\right).
\end{align}
In the first line we have again used \eqref{overlapav} and also the calculation shown in figure \ref{overlapflucfig}.
\end{proof}
To apply lemma \ref{Meckes} we need a Lipschitz constant for 
\be
F(U)\equiv ||(V\otimes I_{LR})|\psi\ran||=\sqrt{|P|}\big|\big|\lan 0|U\otimes I_{LR}|\psi_0\ran|\psi\ran\big|\big|.
\ee
Indeed we have
\begin{align}\nonumber
|F(U_1)-F(U_2)|&\leq \sqrt{|P|}\big|\big|\lan 0|(U_1-U_2)\otimes I_{LR}|\psi_0\ran|\psi\ran\big|\big|\\\nonumber
&=\sqrt{|P|}\sqrt{\lan \psi_0|\lan \psi|(U_1-U_2)^\dagger \otimes I_{LR}|0\ran\lan 0|(U_1-U_2)\otimes I_{LR}|\psi_0\ran|\psi\ran}\\\nonumber
&\leq \sqrt{|P|}\sqrt{\lan \psi_0|\lan \psi|(U_1-U_2)^\dagger(U_1-U_2)\otimes I_{LR}|\psi_0\ran|\psi\ran}\\\nonumber
&\leq \sqrt{|P|}||(U_1-U_2)\otimes I_{LR}||_\infty\\\nonumber
&=\sqrt{|P|}||U_1-U_2||_\infty\\
&\leq \sqrt{|P|}||U_1-U_2||_2,\label{normlip}
\end{align}
where in the first line we've used the triangle inequality for the Hilbert space norm and $||X||_\infty$ indicates the operator norm
\be
||X||_\infty = \sup_{|||\psi\ran||=1} ||X|\psi\ran||.
\ee
Thus $F$ is a Lipschitz function with $\kappa=\sqrt{|P|}$.  We then have the following deviation bound
\begin{thm}\label{normconcthm}
Let $V$ be defined as in \eqref{Vdef} and $|\psi\ran\in \HL\otimes\Hl\otimes \Hr\otimes \HR$.  Then for any $0<\gamma<\frac{1}{2}$ and $|B|\geq 16$, in the Haar distribution on $U$ we have
\be
\PR\Big[\big|||(V\otimes I_{LR})|\psi\ran||-1\big|\geq |B|^{-\gamma}\Big]\leq 2 \exp\left(-\frac{|B|^{1-2\gamma}}{24}\right).\label{normconc}
\ee
\end{thm}
Thus we see that for any particular state $|\psi\ran$, the norm $||(V\otimes I_{LR})|\psi\ran||$ is doubly-exponentially unlikely (in $\log |B|$) to differ from one by more than an exponentially small amount (in $\log|B|$). 
\begin{proof}
By lemmas \eqref{avbounds} and \eqref{Meckes}, for all $\delta>0$ and $\hat{\delta}>0$ we have
\begin{align}\nonumber
\PR\Big[F(U)\geq 1+\delta\Big]&\leq \PR\Big[F(U)\geq \lan F\ran+\delta \Big]\leq e^{-\frac{\delta^2|B|}{12}}\\
\PR\Big[F(U)\leq 1-\zeta-\hat{\delta}\Big]&\leq \PR\Big[F(U)\leq \lan F\ran -\hat{\delta}\Big]\leq e^{-\frac{\hat{\delta}^2|B|}{12}}
\end{align}
with
\be
\zeta\equiv \frac{1}{2}\frac{|P|-1}{|P|^2|B|^2-1}\left(|P||B|\tr(\psi_{LR}^2)-1\right).
\ee
It is convenient to observe here that
\be\label{zetabound}
\zeta\leq \frac{1}{|B|}.
\ee
Taking $\delta=|B|^{-\gamma}$ and $\hat{\delta}=|B|^{-\gamma}-\zeta$ (by \eqref{zetabound} the latter is positive for $\gamma<1$), we have
\begin{align}\nonumber
\PR\Big[F(U)\geq 1+|B|^{-\gamma}\Big]&\leq \exp\left(-\frac{|B|^{1-2\gamma}}{12}\right)\\
\PR\Big[F(U)\leq 1-|B|^{-\gamma}\Big]&\leq \exp\left(-\frac{|B|^{1-2\gamma}}{12}(1-\zeta |B|^\gamma)^2\right)
\end{align}
We can make these symmetric by weakening them a bit.  Observing that from \eqref{zetabound}, $\gamma<1/2$, and $|B|\geq 16$ we have
\be
1-\zeta|B|^\gamma\geq 1-|B|^{\gamma-1}\geq 1-|B|^{-1/2}\geq\frac{3}{4}>\frac{1}{\sqrt{2}},
\ee
we see that
\begin{align}\nonumber
\PR\Big[F(U)\geq 1+|B|^{-\gamma}\Big]&\leq \exp\left(-\frac{|B|^{1-2\gamma}}{24}\right)\\
\PR\Big[F(U)\leq 1-|B|^{-\gamma}\Big]&\leq\exp\left(-\frac{|B|^{1-2\gamma}}{24}\right).
\end{align}
An application of the union bound $\PR[A\cup B]\leq \PR[A]+\PR[B]$ then gives \eqref{normconc}.
\end{proof}
We can combine this theorem with lemma \ref{overnorm} to get a deviation bound on the overlap between pairs of states: 
\begin{thm}\label{overconcthm}
Let $V$ be defined as in \eqref{Vdef}, and $|\psi_1\ran,|\psi_2\ran$ be states in $\HL\otimes \Hl\otimes \Hr\otimes \HR$.  Then for all $|B|\geq 16$ and $0<\gamma<1/2$ we have
\be
\PR\Bigg[\Big|\lan \psi_2|V^\dagger V\otimes I_{LR}|\psi_1\ran-\lan \psi_2|\psi_1\ran\Big|\geq\sqrt{18}|B|^{-\gamma}\Bigg]\leq 12 \exp\left(-\frac{|B|^{1-2\gamma}}{24}\right).
\ee
\end{thm}
\begin{proof}
By lemma \ref{overnorm}, if we have $|\lan \psi_2|V^\dagger V\otimes I_{LR}|\psi_1\ran-\lan\psi_2|\psi_1\ran|>\sqrt{18}|B|^{-\gamma}$ then at least one of the six states $|\psi_1\ran, |\psi_2\ran, \frac{|\psi_1\ran\pm |\psi_2\ran}{|||\psi_1\pm |\psi_2\ran||},\,\mathrm{or}\, \frac{|\psi_1\ran\pm i|\psi_2\ran}{|||\psi_1\pm i|\psi_2\ran||}$ must obey $\Big|||V\otimes I_{LR}|\psi\ran||-1\Big|>|B|^{-\gamma}$.  The result then follows from theorem \ref{normconcthm} and the union bound.
\end{proof}
Theorem \ref{setconcthm} then follows immediately from theorem \ref{overconcthm} and the union bound.

\section[Epsilon-nets and approximate isometry]{$\epsilon$-nets and approximate isometry}\label{appxapp}
In this appendix we show that the condition \ref{isomcond} is sufficient for the holographic encoding $V$  defined by \eqref{Vdef} to likely be an approximate isometry in the sense that
\be
||V^\dagger V-I_{\ell r}||_\infty<\delta
\ee
with $\delta$ of order $|B|^{-\gamma}$ and $\gamma\in (0,1/2)$.  
It is useful to introduce the the following idea:
\begin{mydef}
Let $\sH$ be a Hilbert space.  A collection $S$ of states in $\sH$ is an \textbf{$\mathbf{\epsilon}$-net for $\mathbf{\sH}$} if for any state $|\psi\ran\in \sH$ there exists a state $|\phi\ran\in S$ such that 
\be
|||\psi\ran-|\phi\ran||\leq \epsilon.
\ee
An $\epsilon$-net is called \textbf{minimal} if no element can be removed without ruining this property.
\end{mydef}
Simple counting arguments\footnote{One simple construction of an $\epsilon$-net is to take a maximal set of vectors $v_i$ such that $||v_i-v_j||>\epsilon$ for all $i\neq j$.  That this is an $\epsilon$-net follows from its maximality (a vector that had distance greater than $\epsilon$ from all $v_i$ could be added to the set), and we can get an upper bound on its size by observing that the set of balls of radius $\epsilon/2$ in $R^{2d}$ which are centered at the $v_i$ must be disjoint and must also fit into the ball of radius $1+\epsilon/2$ centered at the origin.  This gives \eqref{epsnetbound}.  We can obtain the lower bound \eqref{epsnetlower} by demanding that the set of balls of radius $\epsilon$ centered on the $v_i$ must cover $\mathbb{S}^{2d-1}$.} show that for any $\sH$ of dimension $d$ there exists a minimal $\epsilon$-net $S$ whose number of states obeys
\be\label{epsnetbound}
N_S\leq \left(1+\frac{2}{\epsilon}\right)^{2d},
\ee
and also that any $\epsilon$-net must have at least
\be\label{epsnetlower}
N_S\geq \sqrt{\pi d}\left(\frac{1}{\epsilon}\right)^{2d}.
\ee
Thus for small $\epsilon$ and large $d$ we can roughly speaking think of an $\epsilon$-net $S$ as having 
\be\label{epsnet}
N_S\sim e^{2d \log \frac{1}{\epsilon}}
\ee
states. Therefore by theorem \ref{setconcthm} any $\epsilon$-net $S$ of states in $\Hl\otimes\Hr$ is likely to obey
\be\label{netisom}
|\lan \psi|(V^\dagger V-I_{\ell r})|\phi\ran|<\sqrt{18}|B|^{-\gamma}
\ee
for all $|\psi\ran,|\phi\ran\in S$ provided that
\be
\frac{|\ell||r|}{|B|}|B|^{2\gamma}\log \frac{1}{\epsilon}\ll 1.
\ee
For $V$ to be an approximate isometry however, we'd like something like \eqref{netisom} to hold for \textit{all} states in $\Hl\otimes\Hr$.  This indeed follows provided that we take
\be
\epsilon=\frac{1}{|P||B|^{\gamma}},
\ee
since then from \eqref{naivebound} for any $|\psi\ran,|\phi\ran\in \Hl\otimes\Hr$ we have
\begin{align}\nonumber
\Big|\lan \psi|(V^\dagger V-I_{\ell r})|\phi\ran\Big|\leq&\Big|\lan \hat{\psi}|(V^\dagger V-I_{\ell r})(|\hat{\phi}\ran\Big|+\Big|\lan \hat{\psi}|(V^\dagger V-I_{\ell r})\left(|\phi\ran-|\hat{\phi}\ran\right)\Big|\\\nonumber
&+\Big|\left(\lan \psi|-\lan \hat{\psi}|\right)(V^\dagger V-I_{\ell r})|\phi\ran\Big|\\
&\leq \sqrt{18}|B|^{-\gamma}+2\epsilon|P|.\label{nettoall}
\end{align}
Here $|\hat{\psi}\ran$ is the element of $S$ which is closest to $|\psi\ran$ and likewise $|\hat{\phi}\ran$ for $|\phi\ran$.  Thus we see that $V$ is likely to be an approximate isometry to exponential precision in $\log |B|$ provided that
\be
\frac{|\ell||r|}{|B|} |B|^{2\gamma}\left(\gamma \log|B|+\log |P|\right)\ll 1,
\ee
which matches \eqref{isomcond}.

\section{Proof of the decoupling theorem}\label{proofapp2}
In this appendix we prove the decoupling theorem \ref{dcthm}.  The proof uses a few lemmas which generalize standard quantum information theory formulas to non-normalized states.  Namely for any vectors $|\phi\ran$, $|\psi\ran$ in a finite-dimensional Hilbert space we have
\begin{align}\nonumber
|||\phi\ran-|\psi\ran||&=\sqrt{\lan \psi|\psi\ran+\lan\phi|\phi\ran-2\mathrm{Re}\lan\psi|\phi\ran}\\\nonumber
|||\phi\ran\lan\phi|-|\psi\ran\lan\psi|||_1&=\sqrt{\left(\lan\psi|\psi\ran+\lan\phi|\phi\ran\right)^2-4|\lan\psi|\phi\ran|^2}\\
|||\phi\ran\lan\phi|-|\psi\ran\lan\psi|||_1&\leq \sqrt{2(\lan \psi|\psi\ran+\lan\phi|\phi\ran)}\,|||\phi\ran-|\psi\ran||\label{normrel}
\end{align}
and for any positive operators $P$ and $Q$ we have
\be
\tr(P)+\tr(Q)-2F(P,Q)\leq ||P-Q||_1\leq \sqrt{(\tr(P)+\tr(Q))^2-4F(P,Q)^2}.
\ee
Here $F(P,Q)\equiv||\sqrt{P}\sqrt{Q}||_1$ is called the fidelity of $P$ and $Q$.  We'll leave the proofs of these lemmas to the imagination of the reader (see the appendices of \cite{Akers:2021fut}, or any QI textbook, for help).  To show that (1)$\implies$(2) we then have
\begin{align}\nonumber
||\tr_B(LW|\psi\ran\lan\psi|W^\dagger L^\dagger)-\tr_B(L|\psi\ran\lan\psi L^\dagger)||_1&\leq||LW|\psi\ran\lan\psi|W^\dagger L^\dagger-W_BL|\psi\ran\lan\psi|L^\dagger W_B^\dagger||_1\\\nonumber
&\leq \sqrt{2(\lan \psi|L^\dagger L|\psi\ran+\lan\psi|W^\dagger L^\dagger L W|\psi\ran)}||W_BL|\psi\ran-LW|\psi\ran||\\
&\leq \sqrt{2(\lan \psi|L^\dagger L|\psi\ran+\lan\psi|W^\dagger L^\dagger L W|\psi\ran)}\epsilon_1,
\end{align}
where the first line uses the monotonicity of the trace norm under partial trace.  To show (2)$\implies$(1) we have
\begin{align}\nonumber
\epsilon_2&\geq ||\tr_B(LW|\psi\ran\lan\psi|W^\dagger L^\dagger)-\tr_B(L|\psi\ran\lan\psi| L^\dagger)||_1\\\nonumber
&\geq\lan \psi|L^\dagger L|\psi\ran+\lan\psi|W^\dagger L^\dagger L W|\psi\ran-2F(\tr_B(LW|\psi\ran\lan\psi|W^\dagger L^\dagger),\tr_B(L|\psi\ran\lan\psi |L^\dagger))\\\nonumber
&=\lan \psi|L^\dagger W_B^\dagger W_B L|\psi\ran+\lan\psi|W^\dagger L^\dagger L W|\psi\ran-2\mathrm{Re}\lan\psi|W^\dagger L^\dagger W_BL|\psi\ran\\
&=||W_BL|\psi\ran-LW|\psi\ran||^2,
\end{align}
where the existence of a unitary $W_B$ on $\mathcal{H}_B$ such that the equality in the third line holds is a consequence of Uhlmann's theorem:
\begin{thm}
(Uhlmann) For any positive-semidefinite operators $P,Q$ on $\mathcal{H}_A$ we have $F(P,Q)=\max_{|\psi\ran,|\phi\ran\in \mathcal{H}_A\otimes \mathcal{H}_{\ol{A}}}|\lan \psi|\phi\ran|$, where $|\ol{A}|$ is greater than or equal to the ranks of both $P$ and $Q$ and we only maximize over $|\phi\ran,|\psi\ran$ such that $P=\tr_{\ol{A}}|\phi\ran\lan\phi|$ and $Q=\tr_{\ol{A}}|\psi\ran\lan\psi|$.
\end{thm}
For a proof see e.g. theorem 3.22 in \cite{watrous_2018}.  To understand the relationship to $W_B$, note that the conditions on $|\phi\ran$ and $|\psi\ran$ just say that $|\phi\ran$ is a purification of $P$ and $|\psi\ran$ is a purification of $Q$.  By the Schmidt decomposition any two purifications of the same positive semi-definite operator differ only by the action of a unitary on the purifying system, and it is this unitary which becomes $W_B$ (the symbol replacements are $A\to \ol{B}$ and $\ol{A}\to B$).  

\section{Deviation bound for decoupling in sub-exponential states}\label{fixdcapp}
In this appendix we use measure concentration to show that it is very likely that for all sub-exponential states $|\psi\ran\in \Hall$ and sub-exponential unitaries $W_{\ell rR}$ the decoupling bound \eqref{decoupling} holds with $\epsilon_2$ being exponentially small in $\log |B|$.  We first need a lemma:
\begin{lemma}
Let $V$ be defined as in \eqref{Vdef}, $W_{\ell r R}$ be a unitary on $\Hl\otimes\Hr\otimes \HR$, and $|\psi\ran$ be a state in $\Hall$.  Moreover define
\be
F(U)\equiv \frac{||\Psi_L(U,W_{\ell r R})-\Psi_L(U,I)||_1}{||(V\otimes I_{LR})W_{\ell r R}|\psi\ran||+||(V\otimes I_{LR})|\psi\ran||},
\ee
where $\Psi_L(U,W_{\ell r R})$ is the reduction to $L$ of the state $\Psi_{LBR}(U,W_{\ell r R})$ defined by \eqref{LBRdef}. Then we have
\be\label{Fbound}
F(U)\leq ||(V\otimes I_{LR})W_{\ell r R}|\psi\ran||+||(V\otimes I_{LR})|\psi\ran||,
\ee
and moreover $F$ is a Lipschitz function of $U$ with Lipschitz constant $6\sqrt{|P|}$:
\be
|F(U_2)-F(U_1)|\leq 6\sqrt{|P|} ||U_2-U_1||_2.\label{Flip}
\ee
\end{lemma}
\begin{proof}
First of all by the triangle inequality and the monoticity of the partial trace we have
\begin{align}\nonumber
||\Psi_L(U,W_{\ell r R})-\Psi_L(U,I)||_1&\leq ||\Psi_L(U,W_{\ell r R})||_1+||\Psi_L(U,I)||_1\\\nonumber
&\leq ||\Psi_{LBR}(U,W_{\ell r R})||_1+||\Psi_{LBR}(U,I)||_1\\\nonumber
&\leq ||(V\otimes I_{LR})W_{\ell r R}|\psi\ran||^2+||(V\otimes I_{LR})|\psi\ran||^2\\
&\leq (||(V\otimes I_{LR})W_{\ell r R}|\psi\ran||+||(V\otimes I_{LR})|\psi\ran||)^2.
\end{align}
This establishes \eqref{Fbound}.

To see that $F$ is a Lipschitz function, we first observe that since we are using the Euclidean Hilbert-Schmidt metric on unitaries it is sufficient to establish a local Lipschitz constant \textit{provided} that we establish it for non-unitary $U$ as well as unitary $U$ (this is because the geodesic from $U_1$ to $U_2$ with length $||U_2-U_1||_2$ will usually pass through non-unitary matrices).  Therefore we only need to establish \eqref{Flip}
to linear order in $||U_2-U_1||_2$.  We can assume without loss of generality that $F(U_2)\geq F(U_1)$, so to linear order in $||U_2-U_1||_2$ we have
\begin{align}\nonumber
|F(U_2)-F(U_1)|&= F(U_2)-F(U_1)\\\nonumber
&\leq \frac{||\Psi_L(U_2,W_{\ell r R})-\Psi_L(U_2,I)||_1-||\Psi_L(U_1,W_{\ell r R})-\Psi_L(U_1,I)||_1}{||(V_1\otimes I_{LR})W_{\ell r R}|\psi\ran||+||(V_1\otimes I_{LR})|\psi\ran||}\\\nonumber
&\phantom{\leq}-\frac{F(U_1)\Big(||(V_2\otimes I_{LR})W_{\ell r R}|\psi\ran||+||(V_2\otimes I_{LR})|\psi\ran||\Big)}{||(V_1\otimes I_{LR})W_{\ell r R}|\psi\ran||+||(V_1\otimes I_{LR})|\psi\ran||}\\
&\phantom{\leq}+\frac{F(U_1)\Big(||(V_1\otimes I_{LR})W_{\ell r R}|\psi\ran||+||(V_1\otimes I_{LR})|\psi\ran||\Big)}{||(V_1\otimes I_{LR})W_{\ell r R}|\psi\ran||+||(V_1\otimes I_{LR})|\psi\ran||}
\end{align}
By \eqref{Fbound} the second and third terms are bounded by $2\sqrt{|P|}||U_2-U_1||_1$ since $||(V\otimes I_{LR})|\psi\ran||$ is a $\sqrt{|P|}$-Lipschitz function of $U$ (see equation  \eqref{normlip}), and we can bound the first term via
\begin{align}\nonumber
&\frac{||\Psi_L(U_2,W_{\ell r R})-\Psi_L(U_2,I)||_1-||\Psi_L(U_1,W_{\ell r R})-\Psi_L(U_1,I)||_1}{||(V_1\otimes I_{LR})W_{\ell r R}|\psi\ran||+||(V_1\otimes I_{LR})|\psi\ran||}\\\nonumber
&\leq \frac{||\Psi_L(U_2,W_{\ell rR})-\Psi_L(U_1,W_{\ell r R})||_1+||\Psi_L(U_2,I)-\Psi_L(U_1,I)||_1}{||(V_1\otimes I_{LR})W_{\ell r R}|\psi\ran||+||(V_1\otimes I_{LR})|\psi\ran||}\\\nonumber
&\leq \frac{||\Psi_{LBR}(U_2,W_{\ell rR})-\Psi_{LBR}(U_1,W_{\ell r R})||_1+||\Psi_{LBR}(U_2,I)-\Psi_{LBR}(U_1,I)||_1}{||(V_1\otimes I_{LR})W_{\ell r R}|\psi\ran||+||(V_1\otimes I_{LR})|\psi\ran||}\\\nonumber
&\leq 2\Bigg(\frac{||(V_1\otimes I_{LR})W_{\ell r R}|\psi\ran||\times||((V_2-V_1)\otimes I_{LR})W_{\ell r R}|\psi\ran||}{||(V_1\otimes I_{LR})W_{\ell r R}|\psi\ran||+||(V_1\otimes I_{LR})|\psi\ran||}\\\nonumber
&\phantom{\leq 2\Bigg(}+\frac{||(V_1\otimes I_{LR})|\psi\ran||\times||((V_2-V_1)\otimes I_{LR})|\psi\ran||}{||(V_1\otimes I_{LR})W_{\ell r R}|\psi\ran||+||(V_1\otimes I_{LR})|\psi\ran||}\Bigg)\\
&\leq 4 \sqrt{|P|}||U_2-U_1||_2,
\end{align}
where the first inequality uses two applications of the triangle inequality for the trace norm, the second inequality uses the monotonicity of the trace norm, the third inequality uses two applications of the third line of \eqref{normrel} and also neglects higher-order terms in $||U_2-U_1||_2$, and the last inequality follows from the same argument as in \eqref{normlip}.  Thus \eqref{Flip} is established.
\end{proof}
It is also useful to observe that
\begin{align}\nonumber
\int dU F(U)&\leq \sqrt{\int dU F(U)^2}\\\nonumber
&\leq \sqrt{\int dU F(U)\Big(||(V\otimes I_{LR})W_{\ell r R}|\psi\ran||+||(V\otimes I_{LR})|\psi\ran||\Big)}\\\nonumber
&=\sqrt{\int dU ||\Psi_L(U,W_{\ell r R})-\Psi_L(U,I)||_1}\\
&\leq \sqrt{2}\frac{|L|^{1/4}}{|B|^{1/4}},
\end{align}
where the first inequality follows from Jensen's inequality, the second from \eqref{Fbound}, and the third from \eqref{Ldcbound}.  With these results in hand we can then immediately apply lemma \ref{Meckes} to obtain the deviation bound
\begin{align}\nonumber
\PR\Big[F(U)\geq \sqrt{2}\frac{|L|^{1/4}}{|B|^{1/4}}+\frac{1}{|B|^\gamma}\Big]&\leq \PR\Big[F(U)\geq \int dU F(U)+\frac{1}{|B|^\gamma}\Big]\\
&\leq e^{-\frac{1}{432}|B|^{1-2\gamma}},
\end{align}
so choosing $0\leq\gamma<\frac{1}{2}$ we see that we are doubly-exponentially unlikely (in $\log |B|$) for $F(U)$ to not be exponentially small.  Moreover the number of sub-exponential states (and the number of sub-exponential unitaries) are both bounded as in \eqref{subexpcount} for any $\alpha>0$, so choosing $\alpha<1-2\gamma$ we see that it is very likely for $F(U)$ to be exponentially small for all sub-exponential states $|\psi\ran$ and sub-exponential unitaries $W_{\ell r R}$.  Moreover by \eqref{subexpconc} it is also very likely that for these states and unitaries the denominator in the definition of $F(U)$ is exponentially close to being two, and thus the numerator (which is what appears in the decoupling bound \eqref{decoupling}) is exponentially small.

\section{Complexity theory lemmas}\label{complexityapp}
In this appendix we establish various complexity-theoretic results which are needed in the main text.  In all results ``sub-exponential'' can be replaced by ``polynomial'' and the statement will still be true.   
\begin{lemma}\label{measeqlem}
Let $S$ be a quantum system, $X$ a bounded observable on $S$, $A$ an apparatus system labeled by the distinct eigenvalues $x_i$ of $X$, and $U^{\{X\}}$ a sub-exponential unitary which implements the measurement
\be
U^{\{X\}}|i\ran_S|0\ran_A=|i\ran_S|x_i\ran_A.
\ee
Then $X$ can be written as a linear combination of two sub-exponential unitaries.
\end{lemma}
\begin{proof}
The two sub-exponential unitaries are 
\be
U_{\pm}\equiv \frac{X}{||X||}\pm i \sqrt{1-\frac{X^2}{||X||^2}},
\ee
where $||X||$ is the operator norm of $X$ (i.e. the supremum of its singular values).  These are easily confirmed to be unitary, and we can observe that
\be
X=\frac{||X||}{2}(U_++U_-).
\ee
Therefore we need to show that $U_\pm$ are sub-exponential. Defining
\be
u_{\pm}(x)\equiv \frac{x}{||X||}\pm i \sqrt{1-\frac{x^2}{||X||^2}},
\ee
we can implement $U_\pm$ as
\be
|i\ran_S\to |i\ran_S|0\ran_A\to |i\ran_S|x_i\ran_A\to u_{\pm}(x_i)|i\ran_S|x_i\ran_A\to u_{\pm}(x_i)|i\ran_S|0\ran_A\to u_{\pm}(x_i)|i\ran_S.
\ee
Here the first step is adjoining the ancilla in the all-zero state, the second step is acting with $U^{\{X\}}$, the third step is acting with the phase $u_\pm(x)$ on the ancilla (which is a simple function to compute numerically given $x$), the fourth step is acting with $(U^{\{X\}})^\dagger$, and the fifth is removing the ancilla.
\end{proof}

\begin{lemma}(Grover amplification) \label{groverlem}
Let $|\phi\ran$ be some state we would like to prepare, and let $|\psi\ran$ be a state such that 
\bi
\item[(1)] We can write
\be
|\psi\ran=\sin \theta|\phi\ran+\cos\theta|\phi_\perp\ran,
\ee
with $\theta\in(0,\pi/2]$ and $\lan \phi|\phi^\perp\ran=0$.
\item[(2)] There exists a unitary $U_\phi$ such that
\begin{align}\nonumber
U_\phi|\phi\ran&=-|\phi\ran\\
U_\phi|\phi_\perp\ran&=|\phi_\perp\ran.
\end{align}
\ei
Then there is a circuit of size $O\left(\frac{1}{\sin \theta}\left(\mathcal{C}(|\psi\ran)+\mathcal{C}(U_\phi)\right)\right)$, where $\mathcal{C}(|\psi\ran)$ is the complexity of the circuit which prepares $|\psi\ran$, which prepares $|\phi\ran$ from the all-zero state $|0\ran$ to exponential precision.  In particular if $\frac{1}{\sin \theta}$, $|\psi\ran$, and $U_\phi$ are sub-exponential then there is  sub-exponential circuit which prepares $|\phi\ran$ to exponential precision.\footnote{This lemma is the core of Grover's search algorithm, where the idea is that $|\phi\ran$ is a uniform superposition over bit strings satisfying some desired property, $|\phi_\perp\ran$ is a uniform superposition over bit strings not satisfying that property, and $|\psi\ran$ is a sum over all bit strings.}
\end{lemma}
\begin{proof}
We first define
\be
U_\psi\equiv \hat{U}^{\{|\psi\ran\lan \psi|\}}=2|\psi\ran\lan\psi|-I,
\ee
whose complexity is $O(\mathcal{C}(|\psi\ran))$ since we can write it as the conjugation of $2|0\ran\lan 0|-I$ by the unitary which prepares $|\psi\ran$ (e.g. for qudits a linear time algorithm for $2|0\ran\lan 0|-I$ is to go through the qudits one by one and flip the overall sign as soon as anything other than a zero is encountered).  The idea is then to observe that both $U_\psi$ and $U_\phi$ implement reflections in the two-dimensional real vector space spanned by $|\phi\ran$ and $|\phi_\perp\ran$, so their product $U_\psi U_\phi$ is a rotation in this space.  Moreover this rotation moves $|\psi\ran$ in the direction of $|\phi\ran$ with rotation angle $2\theta$, so applying the $U_\psi U_\phi$ $n$ times we can implement
\be\label{groverrot}
|\psi\ran\to \sin\left[(2n+1)\theta\right]|\phi\ran+\cos\left[(2n+1)\theta\right]|\phi^\perp\ran.
\ee
In the most interesting regime we have $\theta\ll 1$, in which case by performing this rotation $O(1/\theta)$ times we can rotate $|\psi\ran$ to $|\phi\ran$ within at least a precision which is $O(\theta)$. To get to exponential precision (and also to deal with the case where $\theta$ is $O(1)$) is only slightly harder. The trick is to slightly reduce $\theta$ to
\be
\theta'=\frac{\pi}{2(2n+1)},
\ee
where $n$ is the smallest integer such that $\theta'\leq \theta$, in which case the rotation \eqref{groverrot} prepares $|\phi\ran$ exactly. To do this, we add an ancilla qubit in the state $a \ket{0} + b \ket{1}$, which we can do to exponential precision in sub-exponential (and in fact low-order polynomial) time by the Solovay-Kitaev theorem (see e.g. \cite{nielsen&chuang}). We have
\begin{align}
    \ket{\psi} (a \ket{0} + b \ket{1}) = a \,\sin \theta \ket{\phi}\ket{0} + \ket{\mathrm{other}}~,
\end{align}
where $\ket{\mathrm{other}}$ is orthogonal to $\ket{\phi}\ket{0}$. We now carry out exactly the same procedure as above except with $\ket{\psi}$ replaced by $\ket{\psi} (a \ket{0} + b \ket{1})$ and $\ket{\phi}$ replaced by $\ket{\phi}\ket{0}$. Choosing $a$ so that $a\sin \theta=\sin \theta'$ to exponential precision, we thus can ensure that we hit $\ket{\phi}\ket{0}$ to exponential precision.
\end{proof}

\begin{lemma}Let $|\psi\ran$ and $|\phi\ran$ be sub-exponential states, and let $a,b\in \mathbb{C}$ be such that $||a|\psi\ran+b|\phi\ran||=1$.  Moreover assume that $|a|^2+|b|^2$ is upper bounded by some sub-exponential function.  Then there is a sub-exponential circuit which prepares $a|\phi\ran+b|\psi\ran$ with exponentially small error.\label{suplemapp}
\end{lemma}
\begin{proof}
We first initialize an ancillary qubit in the state
\be
\frac{1}{\sqrt{|a|^2+|b|^2}}\left(a|0\ran+b|1\ran\right),
\ee
which we can again do to exponential precision in polynomial time by the Solovay-Kitaev theorem.  
The unusual normalization here arises since $|\psi\ran$ and $|\phi\ran$ are not necessarily orthogonal; instead we have
\be\label{normalization}
|a|^2+|b|^2=1-2\mathrm{Re}\left(a^*b \lan \psi|\phi\ran\right).
\ee
We then prepare $|\psi\ran$ or $|\phi\ran$ conditioned on the state of this ancillary qubit, followed by a Hadamard gate on the ancillary qubit:
\begin{align}\nonumber
\frac{1}{\sqrt{|a|^2+|b|^2}}\big(a|0\ran+b|1\ran\big)\otimes|0\ldots 0\ran&\to\frac{1}{\sqrt{|a|^2+|b|^2}}\Big(a|0\ran|\psi\ran+b|1\ran|\phi\ran\Big)\\\nonumber
&\to\frac{1}{\sqrt{|a|^2+|b|^2}}\Bigg(a\frac{|0\ran+|1\ran}{\sqrt{2}}|\psi\ran+b\frac{|0\ran-|1\ran}{\sqrt{2}}|\phi\ran\Bigg)\\
&=\frac{1}{\sqrt{2}}\frac{1}{\sqrt{|a|^2+|b|^2}}\Bigg(|0\ran\Big(a|\psi\ran+b|\phi\ran\Big)+|1\ran\Big(a|\psi\ran-b|\phi\ran\Big)\Bigg).\label{prepstate}
\end{align}
The next step is to apply Grover amplification to enhance the first branch of this state.  To do this we need a sub-exponential unitary which acts as $-1$ on the first branch and $+1$ on the second, but this is precisely minus the Pauli $Z$ operator on the ancillary qubit. The Grover angle $\theta$ is given by
\be
\sin \theta=\frac{1}{\sqrt{2}}\frac{1}{\sqrt{|a|^2+|b|^2}},
\ee
which by assumption is bigger than the inverse of some sub-exponential function.\footnote{The assumed sub-exponential bound on $|a|^2+|b|^2$ is only necessary if $|\psi\ran$ and $|\phi\ran$ are exponentially close to being proportional to each other.  Indeed from \eqref{normalization} we have 
\begin{align}\nonumber
|a|^2+|b|^2&=1-2\mathrm{Re}(a^* b\lan \psi |\phi\ran)\\\nonumber
&\leq 1+2|a||b||\lan \psi|\phi\ran|\\
&\leq 1+(|a|^2+|b|^2)|\lan \psi|\phi\ran|,
\end{align}
and thus
\be
|a|^2+|b|^2\leq \frac{1}{1-|\lan \psi|\phi\ran|},
\ee
so the bound is automatic unless $|\lan \psi|\phi\ran|$ is exponentially close to one.  In this situation there could be a ``large cancellation'' in $a|\psi\ran+b|\phi\ran$, so the bound is necessary to ensure we don't amplify an exponentially small remnant which has exponential complexity.  We thank Scott Aaronson for pointing out this possibility.
}
We may thus use Grover amplification with a sub-exponential runtime of order $\frac{1}{\sqrt{|a|^2+|b|^2}}$ and the lemma is proved.\footnote{The fact that we get a run time of order $\frac{1}{\sqrt{|a|^2+|b|^2}}$ instead of the naive $\frac{1}{|a|^2+|b|^2}$ we'd get from just measuring the ancillary qubit and repeating until we get $0$ is the characteristic square-root improvement of Grover's algorithm.}
\end{proof}

\begin{lemma}
Let $X$ be an sub-exponential observable with eigenstates $|i\ran$, let $U^{\{X\}}:|i\ran_S|0\ran_A\mapsto |i\ran_S|x_i\ran_A$ be a sub-exponential measurement unitary for $X$, let $|\psi\ran$ be a sub-exponential state, and let $P_x$ be the projection onto those $|i\ran$ with $x_i=x$.  Then the state $P_x|\psi\ran$ is also sub-exponential provided that $\lan \psi|P_x|\psi\ran$ is not exponentially small. \label{projlem}
\end{lemma}
\begin{proof}
The idea is to write
\be
|\psi\ran=\sqrt{\lan \psi|P_x|\psi\ran} \frac{P_x|\psi\ran}{\sqrt{\lan\psi|P_x|\psi\ran}}+\sqrt{1-\lan \psi|P_x|\psi\ran}\frac{(1-P_x)|\psi\ran}{\sqrt{1-\lan \psi|P_x|\psi\ran}}
\ee
and then apply Grover amplification to the first branch.  To enable this, we need to argue that there is a sub-exponential unitary $U_x$ which acts as $-1$ on the image of $P_x$ and $+1$ on its kernel. Defining $f(x_i)$ to be one if $x_i=x$ and zero if $x_i\neq x$, we can implement $U_x$ as follows:
\begin{align}
|i\ran_S\to |i\ran_S|0\ran_A \to |i\ran_S|x_i\ran_A \to (-1)^{f(x_i)}|i\ran_S|x_i\ran_A\to (-1)^{f(x_i)}|i\ran_S|0\ran_A \to (-1)^{f(x_i)}|i\ran_S.
\end{align}
Here the first step adjoins the ancilla in the all-zero state, the second acts with $U^{\{X\}}$, the third acts with $(-1)^{f(x_i)}$ on the ancilla (which can be done in linear time e.g. by going through the binary expression of $x_i$ and checking if each digit agrees with $x$), the fourth acts with $(U^{\{X\}})^\dagger$, and the fifth removes the ancilla.  Therefore the assumptions of lemma \ref{groverlem} apply, and so $P_x|\psi\ran$ is sub-exponential.
\end{proof}

\begin{lemma} \label{groverlem2} (Multi-state Grover amplification) Let $|\phi_i\ran$ be some set of $N$ orthonormal states we would like to prepare and $|\psi_i\ran$ be some set of $N$ orthonormal states such that the following conditions hold.
\bi
\item[(1)] We have
\be
|\psi_i\ran=\sin \theta |\phi_i\ran+\cos \theta |\phi_i^\perp\ran,
\ee
with $\theta\in (0,\pi/2]$ and $\lan \phi_i|\phi_j^\perp\ran=0$ for all $i,j$.  We emphasize that $\theta$ is independent of $i$.
\item[(2)] There exists a unitary $U_\phi$ such that
\begin{align}\nonumber
    U_\phi|\phi_i\ran&=-|\phi_i\ran\\
    U_\phi|\phi_i^\perp\ran&=|\phi_i^\perp\ran.
\end{align}
\ei
Then there is a circuit $U$ of size $O\left(\frac{1}{\sin \theta}(\mathcal{C}(|\psi_i\ran)+\mathcal{C}(U_\phi))\right)$ which implements
\be
U|\psi_i\ran=|\phi_i\ran
\ee
to exponential precision.  In particular if $\frac{1}{\sin\theta}$, $U_\phi$, and $|\psi_i\ran$ are sub-exponential then so is $U$.
\end{lemma}
\begin{proof}
The proof is the same as for lemma \ref{groverlem}, but with the role of $U_\psi$ now played by
\be
U_\psi\equiv (-1)^{N-1}\big(2|\psi_1\ran\lan\psi_1|-I\big)\ldots\big(2|\psi_N\ran\lan\psi_N|-I\big)
\ee

\end{proof}

\section{State-independent reconstruction no-go}\label{app:rec_no_go}

This appendix argues that there cannot be reconstructions of all sub-exponential operators that work with small error $\epsilon$ on all sub-exponential states unless the kernel of $V$ is small.  The idea is to show that if such reconstructions exist, then they also work on a maximally entangled state of the system with a reference system. 
This in turn implies that the encoding map must have a relatively small kernel.  To avoid inessential complications, we will give the proof assuming that all Hilbert spaces are tensor products of qubits.  
Specifically we prove

\begin{thm}\label{thm:no_sir_simple}
    Let $\mathcal{H}_a, \mathcal{H}_B$, and $\mathcal{H}_A$ be finite-dimensional Hilbert spaces, which we take to be tensor products of qubits, and let $V: \mathcal{H}_a \to \mathcal{H}_B$ be a linear map. Furthermore, let $\log(|a||A|)$ be sub-exponential in $\log|B|$.
    Assume that for every sub-exponential (in $\log |B|$) unitary $W_{aA}$, there exists a unitary $\wt{W}_{BA}$ such that for all sub-exponential (in $\log |B|$) states $\ket{\psi}_{aA}$ we have
    \begin{equation}
	\lVert \wt{W}_{BA} V \ket{\psi} - V W_{aA} \ket{\psi} \rVert \le \epsilon~.\label{Wbound}
    \end{equation}
    Then 
    \begin{equation}\label{eq:B_lower_bound}
        |B| \ge  |a| \left(1 - \frac{4 \epsilon}{1 + \delta} \log_2 \left(|a||A|\right)\right)~,
    \end{equation}
    where
    \be
    \label{deltadef}
    1+\delta \equiv || V|\mathrm{MAX}\ran_{aA, \ol{a}\ol{A}}||,
    \ee
    with
    \be
    |\mathrm{MAX}\ran_{aA,\ol{a}\ol{A}}\equiv \frac{1}{\sqrt{|a||A|}}\sum_i|i\ran_{aA}|i\ran_{\ol{a}\ol{A}}
    \ee
    being the canonical maximally-entangled state in the computational basis $|i\ran_{aA}$ for $aA$ (i.e. the $z$-basis for each qubit).  Moreover \eqref{eq:B_lower_bound} still holds if \eqref{Wbound} is only required to hold when $W_{aA}$ is a single-site Pauli operator.
\end{thm}

\begin{proof}
The proof relies on introducing a certain isometry  $L:\mathcal{H}_a\otimes \mathcal{H}_A\to \mathcal{H}_a\otimes \mathcal{H}_A\otimes\mathcal{H}_t\otimes\mathcal{H}_k$, where $t$ and $k$ are two additional copies of $aA$.  This isometry is defined by
\be
L|\psi\ran_{aA}\equiv |\psi\ran_t|\mathrm{MAX}\ran_{k,aA},
\ee
which we can interpret as tensoring in a maximally-entangled state $|\mathrm{MAX}\ran_{k,t}$ and then acting with the swap operator exchanging $aA$ and $t$.  Noting that the swap operator on two qubits can be written as
\be
S=\frac{1}{2}\sum_{b=0}^3 \sigma_b\otimes \sigma_b,
\ee
where $\sigma_1, \sigma_2, \sigma_3$ are the Pauli matrices and $\sigma_0=I$, we have $L=\otimes_n L_n$, where $n$ labels the qubits in $aA$ and 
\be
L_n\equiv \frac{1}{2}\sum_{b=0}^3\sigma_b^{(t_n)}|\mathrm{MAX}\ran_{t_n,k_n}\otimes \sigma_b^{(n)}.
\ee
Here $\sigma_b^{(n)}$ acts on the $n$th qubit in $aA$ and $\sigma_b^{(t_n)}$ acts on the $n$th qubit in $t$.

We can also define a reconstructed version of $L$,
\be
\wt{L}\equiv \otimes_n \wt{L}_n,
\ee
with 
\be
\wt{L}_n\equiv \frac{1}{2}\sum_{b=0}^3\sigma_b^{(t_n)}|\mathrm{MAX}\ran_{t_n,k_n}\otimes \wt{\sigma}_b^{(n)}.
\ee
Here $\wt{\sigma}_b^{(n)}$ is the unitary reconstruction of $\sigma_b^{(n)}$ promised by the assumptions of the theorem.  Noting that the Pauli matrices obey
\be
\sigma_a\sigma_{b}=\delta_{a,b}+i\epsilon_{abc}\sigma_c,
\ee
we see that each $\wt{L}_n$ is itself an isometry:
\be
\wt{L}_n^\dagger\wt{L}_n=\frac{1}{4}\sum_{b,b'}\lan \mathrm{MAX}|_{t_n,k_n}\sigma^{(t_n)}_{b'}\sigma^{(t_n)}_b|\mathrm{MAX}\ran_{t_n,k_n}\otimes \left(\wt{\sigma}_{b'}^{(n)}\right)^\dagger \wt{\sigma}_b^{(n)}=I_{aA}.
\ee
$\wt{L}$ is therefore also an isometry. 

We now bound the error of $\wt{L}$ as a reconstruction of $L$.  We can first observe that for any sub-exponential state $|\psi\ran_{aA}$, by the triangle inequality and \eqref{Wbound} we have
\begin{align}\nonumber
||(\wt{L}_nV-VL_n)|\psi\ran||&=\frac{1}{2}||\sum_b\sigma_b^{(t_n)}|\mathrm{MAX}\ran_{t_n,k_n}\otimes\left(\wt{\sigma}^{(n)}_bV-V\sigma_b^{(n)}\right)|\psi\ran||\\
&\leq 2\epsilon.\label{Lnbound}
\end{align}
The same bound holds acting on the state $|\mathrm{MAX}\ran_{aA,\ol{a}\ol{A}}$.  Indeed
defining
\be
|\Delta_{i,n}\ran_{aA}\equiv (\wt{L}_nV-VL_n)|i\ran_{aA},
\ee
we have
\begin{align}\nonumber
\left\lVert \left(\wt{L}_nV-VL_n\right)|\mathrm{MAX}\ran_{aA,\ol{a}\ol{A}} \right\rVert &=\Big|\Big|\frac{1}{\sqrt{|a||A|}}\sum_i|\Delta_{i,n}\ran_{aA}\otimes |i\ran_{\ol{a}\ol{A}}\Big|\Big|\\\nonumber
&=\sqrt{\frac{1}{|a||A|}\sum|||\Delta_{i,n}\ran||^2}\\\nonumber
&\leq \sqrt{\frac{1}{|a||A|}\sum_i 4\epsilon^2}\\
&=2\epsilon.
\end{align}
Here we have used that each computational basis state is sub-exponential since $\log (|a||A|)$ is sub-exponential in $\log |B|$.  We can then derive an analogous bound for $L$ and $\wt{L}$ as follows:
\begin{align}\nonumber
||(\wt{L}V-VL)|\mathrm{MAX}\ran_{aA,\ol{a}\ol{A}}||&=\Big|\Big|\sum_{n=1}^N \wt{L}_N\ldots \wt{L}_{n+1} (\wt{L}_nV-VL_n)L_{n-1}\ldots L_1  |\mathrm{MAX}\ran_{aA,\ol{a}\ol{A}}\Big|\Big|\\\nonumber
&=\Big|\Big|\sum_{n=1}^N \wt{L}_N\ldots\wt{L}_{n+1} (\wt{L}_nV-VL_n)\hat{L}_{n-1}\ldots \hat{L}_1  |\mathrm{MAX}\ran_{aA,\ol{a}\ol{A}}\Big|\Big|\\\nonumber
&=\Big|\Big|\sum_{n=1}^N \wt{L}_N\ldots\wt{L}_{n+1}\hat{L}_{n-1}\ldots \hat{L}_1 (\wt{L}_nV-VL_n)  |\mathrm{MAX}\ran_{aA,\ol{a}\ol{A}}\Big|\Big|\\
&\leq 2\epsilon N.\label{Lbound}
\end{align}
Here $N=\log_2(|a||A|)$ is the number of qubits in $aA$ and
\be
\hat{L}_n\equiv \frac{1}{2}\sum_b \sigma_{b}^{(t_n)}|\mathrm{MAX}\ran_{t_n,k_n} \otimes (\hat{\sigma}^{(n)}_b)^T,
\ee
where $\hat{\sigma}^{(n)}_b$ acts on the $n$th qubit in $\ol{a}\ol{A}$.  $\hat{L}_n$ acts on $|\mathrm{MAX}\ran_{aA,\ol{a}\ol{A}}$ in the same way as does $L_n$, but it commutes with $V$.  In the last step of \eqref{Lbound} we use the triangle inequality, that $\wt{L}_n$ and $\hat{L}_n$ are isometries, and \eqref{Lnbound}.

Finally we can bound the size of the kernel of $V$ as follows. Defining 
\be
\rho_{\ol{a}}=\tr_{BA\ol{A}}\left(V|\mathrm{MAX}\ran\lan \mathrm{MAX}|_{aA,\ol{a}\ol{A}}V^\dagger\right), 
\ee
we have 
\begin{align}\nonumber
\left\lVert \rho_{\ol{a}}-(1+\delta)^2 \frac{I_{\ol{a}}}{|a|}\right\rVert_1&=\left \lVert\tr_{BA\ol{A}tk}\left(\wt{L}V|\mathrm{MAX}\ran\lan \mathrm{MAX}|_{aA,\ol{a}\ol{A}}V^\dagger \wt{L}-VL|\mathrm{MAX}\ran\lan \mathrm{MAX}|_{aA,\ol{a}\ol{A}}L^\dagger V\right)\right\rVert_1\\\nonumber
&\leq \left\lVert\wt{L}V|\mathrm{MAX}\ran\lan \mathrm{MAX}|_{aA,\ol{a}\ol{A}}V^\dagger \wt{L}-VL|\mathrm{MAX}\ran\lan \mathrm{MAX}|_{aA,\ol{a}\ol{A}}L^\dagger V\right\rVert_1\\
&\leq 4\epsilon (1+\delta)\log(|a||A|).\label{abound}
\end{align}
Here we have used that $\wt{L}$ is an isometry, that 
\be
VL|\mathrm{MAX}\ran_{aA,\ol{a}\ol{A}}=|\mathrm{MAX}\ran_{t,\ol{a}\ol{A}}V|\mathrm{MAX}\ran_{aA,k},
\ee
the definition \ref{deltadef} of $1+\delta$, the monotonicity of the trace norm under partial trace, the relation \ref{normrel} between the trace norm and the vector norm, and \eqref{Lbound}.  We can rewrite \eqref{abound} as
\be
\sum_{p}\left|\lambda_p-\frac{1+\delta}{|a|}\right|\leq 4\epsilon(1+\delta) \log(|a||A|),
\ee
where $\lambda_p$ are the eigenvalues of $\rho_{\ol{a}}$.  Finally we can observe that since
\be
\rho_{\ol{a}}=\tr_B\left(V|\mathrm{MAX}\ran\lan \mathrm{MAX}|_{a,\ol{a}}V^\dagger\right),
\ee
at most $|B|$ of the $\lambda_p$ are nonzero.  We therefore have at least $|a|-|B|$ vanishing eigenvalues, and thus we have 
\be
(|a|-|B|)\frac{(1+\delta)^2}{|a|}\leq 4\epsilon(1+\delta) \log(|a||A|),
\ee
from which \eqref{eq:B_lower_bound} immediately follows.  The restriction to single-site Pauli operators follows because in \eqref{Lnbound} we only need \eqref{Wbound} to apply for single-site Pauli operators.
\end{proof}

\section{Calculations in the dynamical model} \label{dyn_app}

In this appendix we provide more concrete calculations for the dynamical model introduced in Section \ref{dynamicsec}. For simplicity, we mostly focus on the less complex block random unitary (BRU) version of the model and only briefly comment on how the analogous calculations would work in the more realistic random pairwise interaction (RPI) version.  To simplify calculations we will assume that $q\gg 1$, but $q$ is still formally $O(1)$ compared to $n_0-t$ so that any black hole we consider will consist of many qudits.

\subsection{Condition for $V_t$ to be an isometry}

Note that $V_t$ defined by \eqref{Vdef2} is an isometry if and only if the state $\ket{\varphi^V}$ defined by 
\be 
\bra{i}_\ell \bra{j}_r \bra{k}_B \ket{\varphi^V} \equiv \frac{1}{\sqrt{|\ell r|}} \bra{k}_B V \ket{i}_{\ell}\ket{j}_{r}  = q^{-(t+m_0)/2} \bra{j}_{R} \bra{k}_B (U_t ... U_1) \ket{i}_\ell \ket{\psi_0}_f    
\ee
is maximally mixed on $\sH_\ell \otimes \sH_r$, which in turn is true if and only if $S_2(\Phi^V_{\ell r})$ is equal to $\log |\ell r| = 3 t\log q $.\footnote{More precisely, what follows from \eqref{S2result} below is that $||V^\dagger V-I||_2\approx \sqrt{\frac{|\ell|^2|r|^2}{|B|}}$.  This only implies that $V$ is approximately isometric ``on average''.  To conclude that it is  approximately isometric in the ``worst case'' sense that $||V^\dagger V-I||_\infty\ll 1$ one can use a measure concentration argument generalizing the one that went into establishing \eqref{isomcond}.}
For notational convenience, we are no longer making the $t$-dependence of the Hilbert spaces $\sH_\ell$, $\sH_r$ and $\sH_B$ explicit.

To evaluate the average of $e^{-S_2(\Phi^V_{\ell r})}$ in the BRU model we use a mapping to a spin model partition function which was introduced in \cite{Hayden:2016cfa} (see also \cite{nahum, frank}). There are two spins associated to each block random unitary $U_m$, one telling us how the ingoing indices of the two copies of $U_m$ are contracted with the outgoing indices of the two copies of $U_m^\dagger$ in the average \eqref{Uresults} and one telling us the same for the outgoing indices of the $U_m$ and the ingoing indices of the $U_m^\dagger$.  It is helpful to introduce the following spin states on 4 copies of the Hilbert space associated with a single qudit: 
\be 
\braket{i_1 i_1' i_2 i_2'|\uparrow} = \frac{1}{q}\delta_{i_1 i_1'} \delta_{i_2 i_2'}, \quad \braket{i_1 i_1' i_2 i_2'|\downarrow} = \frac{1}{q}\delta_{i_1 i_2'}\delta_{i_2 i_1'} \,,
\ee
where each $\ket{i}$, $\ket{i'}$ represent a basis state on one copy. Their inner products are given by 
\be 
\braket{\uparrow| \uparrow}=\braket{\downarrow| \downarrow} =1 , \quad \braket{\uparrow| \downarrow}=\braket{\downarrow| \uparrow} = \frac{1}{q} \, . 
\ee
We can also  define for any operator $O$ the following two states in four copies of the $d-$dimensional Hilbert space that $O$ acts on: 
\be 
\braket{i_1 i_1' i_2 i_2'|O, \uparrow} =  d\, O_{i_1 i_1'}O_{i_1 i_1'}, \quad \braket{i_1 i_1' i_2 i_2'|O, \downarrow} = d\,  O_{i_1 i_2'}O_{i_2 i_1'}\, . 
\ee
With this notation, we can express the second Renyi entropy of $\ket{\varphi^V}$ in $lr$ as   
\begin{gather}  
\exp\left(-S_2(\Phi^V_{lr})\right) =   \bra{\downarrow}_{\ell} \bra{\downarrow}_P  \bra{\uparrow}_B \ket{\Phi^V, \uparrow} \nonumber \\ = \bra{\downarrow}_{p}  \bra{\uparrow}_B \, \,   U \otimes U^{\ast} \otimes U \otimes U^{\ast} \, \, \ket{\downarrow}_{\ell}  \ket{\psi_0, \uparrow}_f 
\label{s2av}
\end{gather} 
where for instance $\ket{\downarrow}_{\ell}$ refers to a tensor product of $\ket{\downarrow}$ on all qudits in $\ell$, and 
\be 
U \equiv U_t U_{t-1} ... U_1 \, . 
\ee
\bfig
    \includegraphics[width=0.7\textwidth]{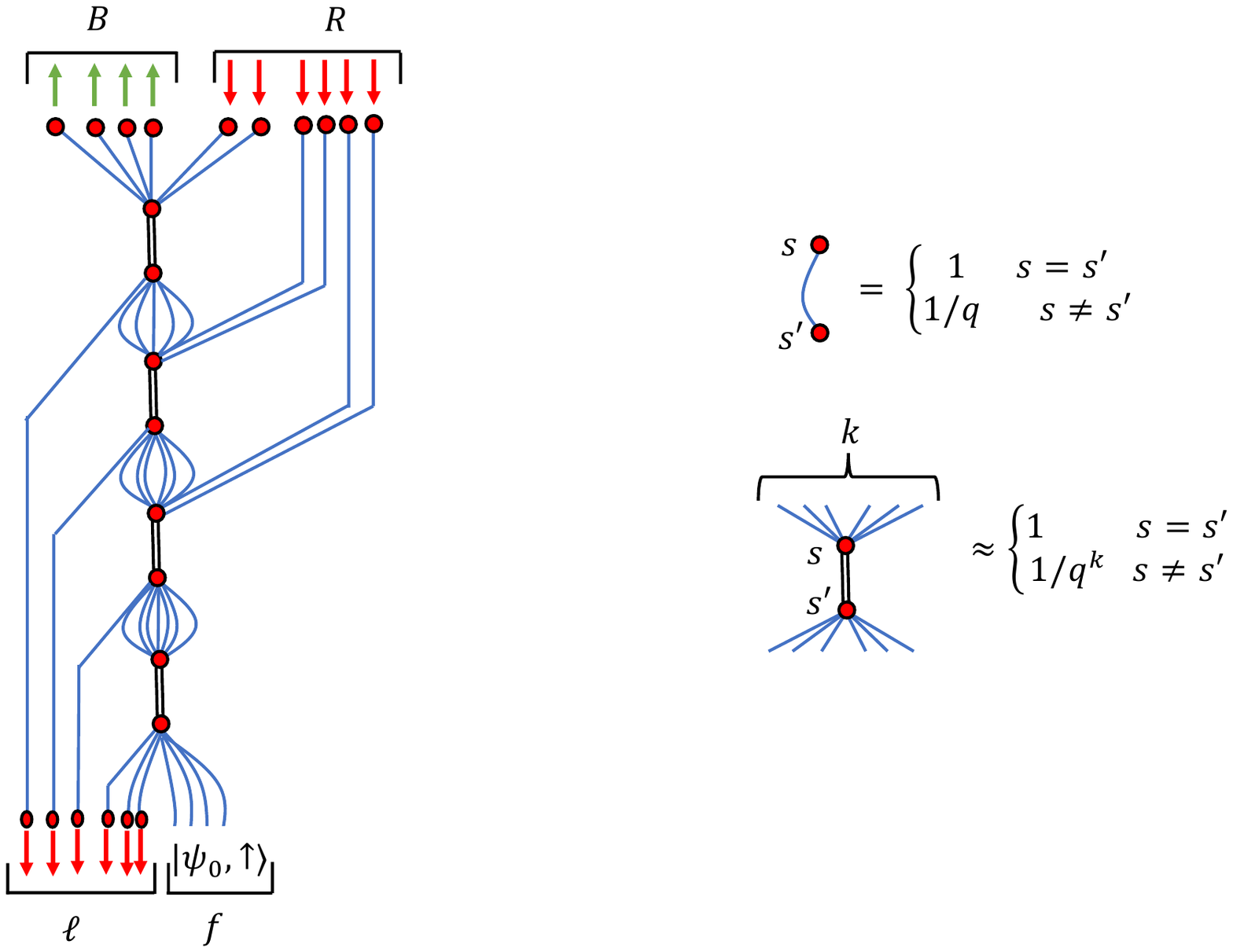}
    \caption{The sums over the $s_m$ and $s'_m$ associated with each $U_m$ appearing the average of \eqref{s2av} can be seen as a partition function on the above lattice (we have shown the case where $n_0=7$, $m_0=3$, and $t=3$).  Each red vertex on the graph, other those fixed at the boundaries, can have either an $\uparrow$ or $\downarrow$ spin, and the total weight of any spin configuration can be obtained by multiplying the weights associated with each of the interactions, which are as shown on the right. The contraction of $\ket{\psi_0, \uparrow}$ with either $\lan\uparrow|$ or $\lan\downarrow|$ contributes a factor of 1.}
    \label{fig:iso}
\efig
The average of \eqref{s2av} can be expressed in terms of the averages $\int d U_m \,  U_m \otimes U_m^{\ast}\otimes U_m \otimes U_m^{\ast}$ for each of the $m$'s. We introduce the spins $s_m$ and $s'_m$ by rewriting \eqref{Uresults} (noting that $U_m$ acts on $n_0-m+2$ qudits) as 
\begin{gather} 
\int dU_m \, U_m \otimes U_m^{\ast}\otimes U_m \otimes U_m^{\ast} = \sum_{s_m, s'_m \in \{\uparrow, \downarrow\}} w_{n_0-m+2} (s_m, s'_m)\, 
(\ket{s_m}\bra{s'_m})^{\otimes{n_0 - m +2}}
\label{umav}
\end{gather} 
where 
\be \label{eq:unitarydomainwalls}
w_k (s, s') = \begin{cases} q^{2k}/(q^{2k}-1) \approx 1  & s = s' \\ 
-q^{k}/(q^{2k}-1) \approx -q^{-k}   & s \neq s'
\end{cases} .
\ee
The approximation comes either from using the fact that we consider times sufficiently before the total evaporation time that $n_0-t\gg 1$, or from taking $q$ to be large. Substituting \eqref{umav} into \eqref{s2av}, the integrals over the $U_m$ are replaced by sums over the $s_m$ and $s'_m$. The overall sum can be interpreted as a partition function on a lattice whose geometry is determined by the tensor network for $U$, as explained in Fig. \ref{fig:iso}.\footnote{This partition function can have negative Boltzmann weights due to the minus sign in \eqref{eq:unitarydomainwalls}.  This can be ``cured'' by integrating out half of the spins and introducing a trivalent vertex, but it is easier to just live with the minus signs since anyways we just need to identify the dominant contribution and this will always be positive.}

\begin{figure}[t]
    \centering
    \includegraphics[width=\textwidth]{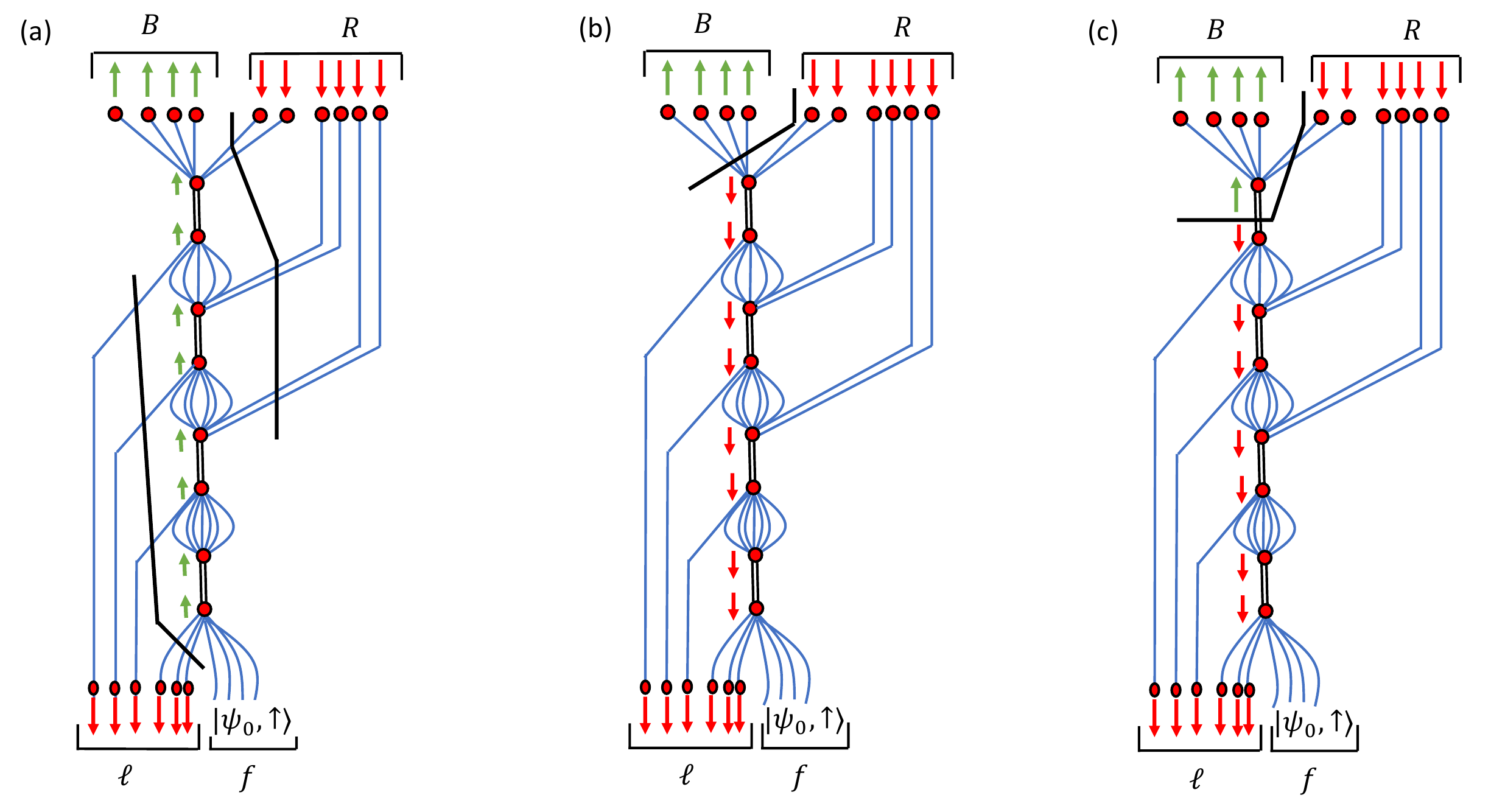}
    \caption{The two configurations, (a) and (b), which can dominate in the partition function \eqref{s2av}, as well as a third (c) which is suppressed compared to (b) by a factor of $1/q$.}
    \label{fig:pf}
\end{figure}
We now evaluate this spin partition function in the regime $1\ll q \ll n_0-t$. There are only two spin configurations that can dominate, which are shown in   Fig. \ref{fig:pf} (a) and (b). These two configurations give
\begin{gather}
\int\mathcal{D}U \, \exp(-S_2(\Phi_{\ell r}^V)) \approx \frac{1}{|\ell r|} + \frac{1}{|B|}, \label{S2result}
\end{gather} 
where $\mathcal{D}U\equiv \prod_{m=0}^t dU_m$.  Other configurations, such as the one shown in Fig. \ref{fig:pf} (c), are further suppressed by factors of $1/q$. 
We therefore find that $\int \mathcal{D}U \, \exp(-S_2(\Phi_{\ell r}^V)) \approx |\ell r|^{-1}$, and hence $V_t$ is typically close to being an isometry, if and only if $t\leq (n_0-m_0)/4$.\footnote{To make this argument quantitatively precise, one can use the fact $e^{-S_2} \geq |\ell r|^{-1}$ and apply Markov's inequality.}

For the RPI model the integrals over unitaries can be performed in an analogous manner, but the resulting partition function will be on a much larger and more complicated lattice. However, for typical random pairings, the dominant contributions are all still analogous to configuration (a) or configuration (b) in Figure \ref{fig:pf}. The primary difference with the BRU calculation is that there are now are large number of degenerate contributions of this form, associated with a freedom to pick the spin of some two qudit unitaries where neither output qudit is emitted as radiation and neither input qudit just fell into the black hole. Naively, each of these configurations would contribute independently to the partition function, increasing the result by a large factor. However, in fact, the minus sign in the $s \neq s'$ term of \eqref{eq:unitarydomainwalls} leads almost all the contributions to cancel, such that the final result is the same as the BRU model, where there is only one configuration of each type.\footnote{This same effect can be seen in a much simpler setting. The distribution of the product of a large number of Haar random unitaries is Haar random, even though the calculations of its moments using \eqref{Uresults} involves many more terms, because the contributions from almost all of those terms cancel.}

\subsection{The QES formula and interior reconstruction}

Let us now see how the second Renyi entropies of the state in the fundamental picture can be expressed in terms of quantities in the effective picture. Suppose we consider a state in the effective description of the form
\be 
\ket{\psi_{\rm eff}} = \ket{\chi^{\{in\}}}_{L \ell} \ket{\mathrm{MAX}}_{rR} \,
\ee
which is the state which would arise from our effective picture dynamics with an ingoing state $\ket{\chi^{\{in\}}}_{L \ell}$,
and ask about the entropies of the state 
$\ket{\psi} = V_t \otimes I_{LR} \ket{\psi_{\rm eff}}$ in different subsystems in the fundamental picture. We will find that these obey the QES formula.

\bfig
    \includegraphics[width=0.7\textwidth]{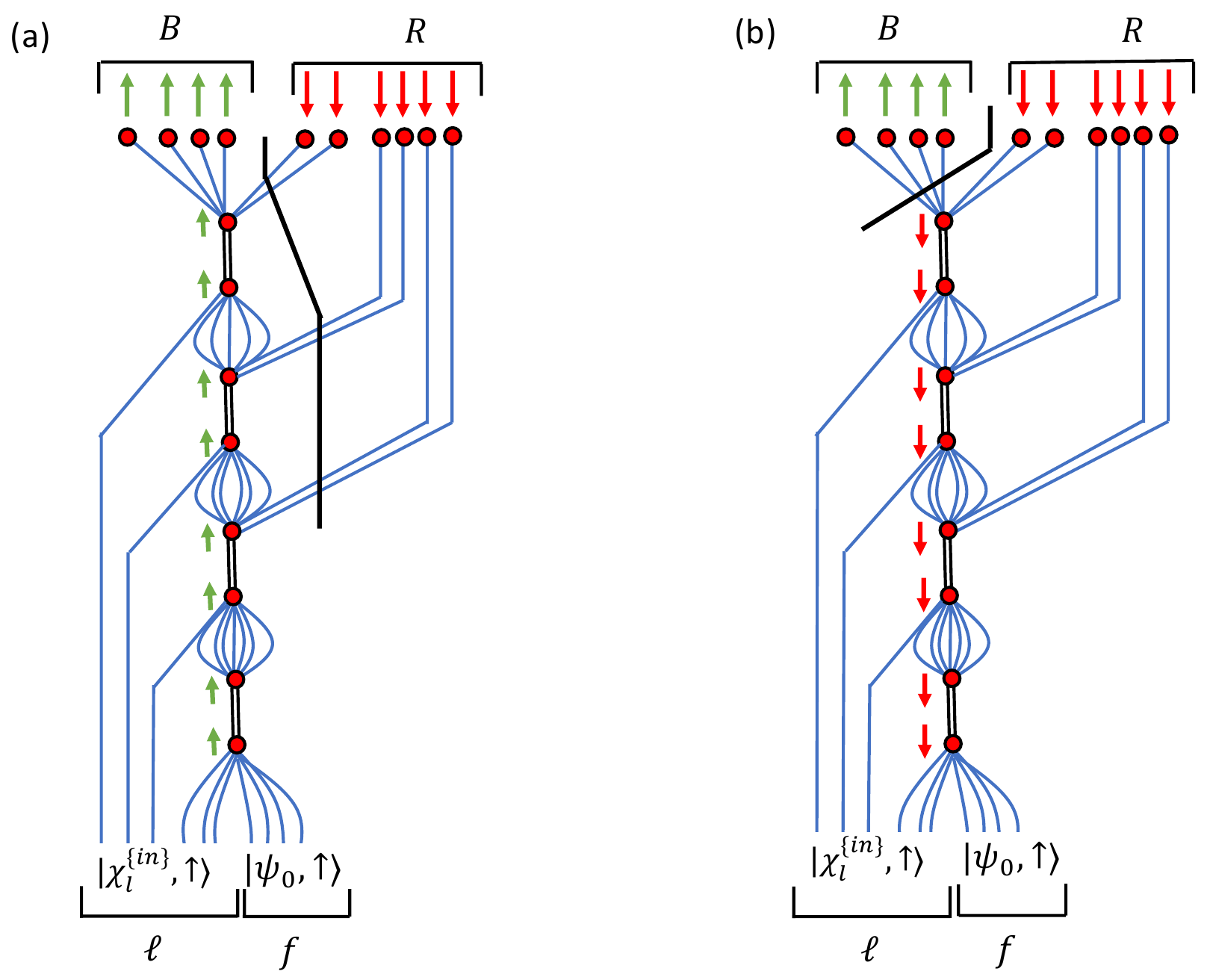}
    \caption{Partition function for $\exp(-S_2(\Psi_R))$, with two possible dominant configurations. Note that  $\ket{\chi_{\ell}^{\{in\}}, \uparrow}$ on the lower boundary gives a factor of 1 when connected to $\uparrow$, and $\exp({-S_2(\chi_{\ell}^{\{in\}})})$ when connected to $\downarrow$.}
    \label{fig:entropy_cases}
\efig
 Using similar steps  to the previous subsection, we find that
the averaged second Renyi entropies of $\ket{\psi}_{LBR}$ in the different subsystems can be expressed as  partition functions. For instance, 
\be 
\exp(-S_2(\Psi_R)) =  \bra{\uparrow}_B  \bra{\downarrow}_{R} \, U \otimes U^{\ast} \otimes U \otimes U^{\ast} \,  \ket{\chi^{\{in\}}_{\ell}, \uparrow}_{\ell} \,  \ket{\psi_0, \uparrow}_f 
\ee
The corresponding lattice and interactions are the same as in Fig. \ref{fig:iso}, but the boundary conditions are now as shown in Fig. \ref{fig:entropy_cases}. When (a) dominates, the entropy is given by $S_2(\Psi_R) = 2 t \log q$. We can interpret this is as the case where only $R$ lies in the entanglement wedge of $R$, and we only have a bulk entropy term in the QES formula. When (b) dominates, $S_2(\Psi_R) = S_2(\chi^{\{ in \}}_{\ell}) + (n_0-t)\log q$.  In this case, the the entanglement wedge of $R$ consists of $\ell rR$. The first term is a bulk entropy term, while the second term is proportional to the number of legs in the tensor network that the quantum extremal surface shown with the solid black line passes through. Overall, we find 
\be \label{eq:S2Rdyn}
S_2(\Psi_R) = \text{min}(2t \log q, ~ S_2(\chi^{\{ in \}}_{\ell}) + (n_0-t)\log q) \, .
\ee

The entanglement entropy of the fundamental picture state in $B$ is given by a similar partition function, with the only change from Fig. \ref{fig:entropy_cases} being that we now have $\downarrow$ spins in $B$ and  $\uparrow$ spins in $R$ in the boundary conditions. Again the two possible dominant contributions have either all bulk spins pointing up or all pointing down, and we find 
\be 
S_2(\Psi_B) = \text{min}\left(2 t \log q +S_2(\chi^{\{ in \}}_{\ell}),~ (n_0-t)\log q\right) \, .
\ee
In the first case, the entanglement wedge of $B$ consists of $\ell$ and $r$ and we have only a bulk entropy term, while in the second case the entanglement wedge of $B$ does not contain either $\ell$ or $r$ and we have only an area term.

A simple variant on the calculations above lets us study when a qudit in $\ell$ can be reconstructed on $B$ or $R$.  Namely if we take 
\be
|\chin\ran=|\mathrm{MAX}\ran_{L_i,\ell_i}|\hat{\chi}^{\mathrm{in}}\ran_{\hat{L}\hat{\ell}},
\ee
where $\ell_i/L_i$ is the $i$th qudit in $\ell/L$ and $\hat{\ell}/\hat{L}$ is its complement, then we will be able to reconstruct $\ell_i$ on $B$ if $\Psi_{RL_i}\approx \Psi_R\otimes \Psi_{L_i}$, while we will be able to reconstruct it on $R$ if  $\Psi_{BL_i}\approx \Psi_B\otimes \Psi_{L_i}$.  We can check these approximate factorizations by showing that the second Renyi entropy is additive.  We can compute $e^{-S_2(\Psi_{R L_i})}$ by setting $\ell_i$ to be down in figure \ref{fig:entropy_cases}, and we can compute $e^{-S_2(\Psi_{B L_i})}$ by setting $\ell_i$ to be down in the analogous figure which flips the spins of $B$ and $R$.  The dominant configurations are still as in figure \ref{fig:entropy_cases}, with the extra qudit in $\ell$ contributing a factor of $1/q$ when the interior spins are up and a factor of $1$ when they are down.  Thus we can reconstruct $\ell_i$ on $B$ whenever the left configuration dominates in figure \ref{fig:entropy_cases}, i.e. whenever the second Renyi of $\Psi_R$ is given by $\log |R|=2t\log q$, and we can reconstruct $\ell_i$ on $R$ whenever the second Renyi of $B$ is given by $\log |B|=(n_0-t)\log q$.  These are just as expected from entanglement wedge reconstruction.

In the RPI model, the second Renyi calculations are essentially identical. However, once we add an $L$ qudit to the computation some new physics arises. Even when the second term in \eqref{eq:S2Rdyn} is minimal, the $L$ qudit is only maximally entangled with $R$ if the entangled $\ell$ qudit is able to influence qudits in $R$. By this we mean that there exists paths from inputs to outputs through two-qudit unitaries that take the $\ell$ qudit to qudits in $R$. If no such path exists then the $L$ qudit will instead be maximally entangled with $B$.  Typically, such paths will only exist if the $\ell$ qudit was thrown into the black hole at least $\log (n_0-t)$ time steps in the past. This is a manifestation of the Hayden-Preskill decoding criterion in the RPI model: information only escapes the black hole (even after the Page time) after a scrambling time delay. The analogue in QES calculations in gravity is the scrambling time delay before an infalling mode enters the ``island''. In contrasts, the delay in the BRU model is only one timestep (for $q \gg 1$), because the BRU model coarse-grains the black hole dynamics such that each timestep corresponds to an entire scrambling time.

\bibliographystyle{jhep}
\bibliography{mybib2022}

\end{document}